\documentclass[11pt]{article}
\usepackage{amsmath}
\usepackage{graphicx}
\usepackage{enumerate}
\usepackage{natbib}
\usepackage{url} 

\usepackage[normalem]{ulem}

\usepackage{color}

\usepackage{xcolor}

\usepackage{floatrow}
\usepackage{graphicx}
\graphicspath{{images/}}
\usepackage{caption}
\usepackage{subcaption}

\usepackage{dirtytalk}

\usepackage{newtxtext}
\usepackage[subscriptcorrection]{newtxmath}

\usepackage{amsthm}

\newtheorem{theorem}{Theorem}
\newtheorem{proposition}{Proposition}

\newtheorem{corollary}{Corollary}

\newtheorem{lemma}{Lemma}
\newtheorem{remark}{Remark}
\newtheorem{assumption}{Assumption}

\usepackage[plain,noend]{algorithm2e}

\makeatletter
\renewcommand{\algocf@captiontext}[2]{#1\algocf@typo. \AlCapFnt{}#2} 
\def\@algocf@capt@plain{top}
\renewcommand{\algocf@makecaption}[2]{%
  \addtolength{\hsize}{\algomargin}%
  \sbox\@tempboxa{\algocf@captiontext{#1}{#2}}%
  \ifdim\wd\@tempboxa >\hsize
    \hskip .5\algomargin%
    \parbox[t]{\hsize}{\algocf@captiontext{#1}{#2}}
  \else%
    \global\@minipagefalse%
    \hbox to\hsize{\box\@tempboxa}
  \fi%
  \addtolength{\hsize}{-\algomargin}%
}
\makeatother

\usepackage{dirtytalk}
\usepackage{comment}

\usepackage{hyperref}
\usepackage{xr-hyper}

\usepackage{multirow}

\newcommand{\blind}{1}

\usepackage{makecell}

\begin{document}

\def\spacingset#1{\renewcommand{\baselinestretch}%
{#1}\small\normalsize} \spacingset{1}


\if1\blind
{
  \title{\bf Overfitted high-dimensional matrix factorizations via adaptive spectral shrinkage}
  \author{Lorenzo Mauri
  \hspace{.2cm}\\
    Department of Statistical Science, Duke University\\
    \small{\texttt{lorenzo.mauri@duke.edu}}\\
    David B. Dunson\\
        Department of Statistical Science, Duke University\\
    }
  \maketitle
} \fi

\if0\blind
{
  \bigskip
  \bigskip
  \bigskip
  \begin{center}
    {\LARGE\bf Overfitted high-dimensional matrix factorizations via adaptive spectral shrinkage}
\end{center}
  \medskip
} \fi

\bigskip
\begin{abstract}
Factor models are popular approaches for analyzing high-dimensional data to extract low-rank signals and estimate covariances. They decompose the covariance matrix as the sum of low-rank and  diagonal components. A key issue is how to choose the latent dimension $k$, which is particularly challenging when the factor model only holds approximately and in low signal-to-noise scenarios. Bayesian overfitted factor models specify an upper bound on $k$ and rely on structured shrinkage priors to effectively remove extra components. Such approaches are popular and effective, but computationally expensive. We propose a much faster \texttt{EigenBayes} approach that provides valid uncertainty quantification, based on spectral estimation of latent factors and adaptive empirical Bayes calibration of key hyperparameters. The resulting posterior distribution factorizes across outcomes and is analytically tractable, bypassing Markov chain Monte Carlo. We show that \texttt{EigenBayes} adapts to the signal-to-noise ratio of each outcome and latent dimension, while  shrinking superfluous latent components to zero. We establish favorable asymptotic properties and demonstrate strong empirical performance in numerical experiments and a genomics application, where EigenBayes outperforms state-of-the-art alternatives. 
\end{abstract}

\noindent%
\vfill

\section{Introduction}

Estimating covariances in high-dimensions is a challenging statistical and computational problem, generally requiring $\mathcal O(p^2)$ parameters to be estimated, where $p$ is the ambient dimension of the data. A popular approach introduces latent factors and models the $i$-th observation as
\begin{equation} \label{eq:fact_model}     \begin{aligned}     y_{i}  =  \Lambda \eta_{i} &+\epsilon_{i},  \quad \epsilon_{i} \sim N_p(0, \Sigma), 
\end{aligned}  \quad (i=1, \dots, n),\end{equation}
where $\eta_i \in \mathbb R^k$ denotes the latent factors for the $i$-th observation, $\Lambda \in \mathbb R^{p \times k}$ is the factor loadings matrix, $ \Sigma = \text{diag}(\sigma_{1}^2, \dots, \sigma_{p}^2)$ is the matrix of residual error variances, $n$ is the sample size, $k$ is the latent dimension, and, generally $k \ll p$. 
Assuming $\eta_{i}  \sim N_{k}(0, {I}_{k})$, and integrating out the latent factors, we obtain an equivalent representation: \begin{equation}
y_{i}  \sim N_p\left(0, \Theta \right),\quad (i=1, \dots, n), \quad \Theta = \Lambda \Lambda^\top + \Sigma.\label{eq:fact_model_int}\end{equation}
Model \eqref{eq:fact_model_int} decomposes the covariance  $\Theta$ as a sum of low-rank and diagonal matrices.
A rich literature has developed corresponding estimators. Letting $ M = \begin{bmatrix}
        \eta_1, \cdots, \eta_n
    \end{bmatrix}^\top \in \mathbb R^{n \times k}$ be the matrix of latent factors, spectral approaches let $M^\top M = nI_k$ and estimate the loading matrix $\Lambda$ with $\hat \Lambda = V_{1:k} D_{1:k} / \sqrt{n}$, where $D_{1:k}$ is a diagonal matrix with the $k$ largest singular values of data matrix $Y$ on the diagonal and $V_{1:k}$ is the matrix of left singular vectors \citep{bai_03, Bai2020SimplerProofs, Doz2012QuasiMaximum}.
    A number of consistent information criteria have been proposed to infer $k$ \citep{ic_bn, ic_ah, chen_jic}. 
In our numerical experiments, information criteria tend to underestimate $k$ in challenging signal-to-noise scenarios. Even when $k$ is specified correctly, spectral estimates can suffer from poor estimation accuracy when the magnitude of some latent components is not large compared to the noise. 
\citet{Barigozzi2018ConsistentEO} propose an estimator that addresses this issue based on scaling the entries of the eigenvectors of the sample covariance matrix for an overfitted rank. However, 
their method lacks uncertainty quantification and is outperformed by alternatives in our experiments.


Early Bayesian approaches \citep{fa_rjmcmc} placed priors on $k$ 
and conducted posterior sampling via inefficient reversible jump Markov chain Monte Carlo sampling. More recently, overfitted factor models have been popular, specifying an upper-bound on $k$ and using shrinkage priors on loadings to remove unnecessary factors \citep{sparse_infinite, durante_note, cusp}. Numerous extensions of overfitted factor models have been proposed and applied in a variety of settings \citep{bmsfa,  kowal2023semiparametric, zorzetto2025multivariate, bayesNMF}. These approaches can be sensitive to hyperparameters and Gibbs samplers can suffer from slow convergence and mixing, particularly when $p$ is large. Alternatively, variational approximations have been proposed \citep{stoch_vi, bb_vi, vi_review, ad_vi, single_cell_fm}, but can massively underestimate uncertainty. \citet{rotate} infer sparsity in the loading matrix via an Indian buffet process \citep{ibp} prior, but focus on point estimation. 

Particularly relevant to our development is 
\citet{fable}, who combine the spectral estimation of latent factors with Bayesian inference for the remaining model parameters, including the loadings and residual error variances. \citet{fable} is related to a rich literature combining pre-estimation of latent factors, via spectral methods or joint maximum likelihood estimation, with Bayesian inference (\cite{fan1, Fan2023BridgingFactor, xie2024eigenvector, wu2025sanvi, wu_xie_graph}). Related approaches tend to have excellent empirical performance and strong theoretical support in high-dimensional asymptotic scenarios \citep{chen_jmle, chen_identfiability, lee24, flair, blast, fama, basil}. However, existing approaches in this area do not provide a mechanism for adaptive overfitted rank selection or adaptation to outcomes with varying signal-to-noise ratios.

Inspired by these considerations, we propose 
\texttt{EigenBayes}, which combines the computational efficiency of spectral estimators with data-adaptive shrinkage and uncertainty quantification of Bayesian approaches. We obtain a spectral estimate of the latent factors using an overestimated latent dimension and infer loadings and residual error variances via surrogate Bayesian linear regression tasks. We adopt normal-inverse gamma priors with outcome- and latent component-dependent shrinkage parameters. Crucially, conditioning on the spectral estimate leads to a factorized posterior across outcomes, reducing inference to $p$  Bayesian regressions, enabling embarrassingly parallel computation. While the point estimators can be interpreted as regularized spectral estimators, sampling from the posterior distribution arising from the surrogate regressions allows for full uncertainty quantification, yielding a computationally efficient alternative to the popular overfitted Bayesian factor models mentioned above. We develop a data-adaptive tuning procedure for shrinkage hyperparameters, which admits interpretable closed-form expressions in terms of the singular value decomposition of the observed data.  The resulting estimates adapt to the signal-to-noise ratio of each outcome and latent dimension. Superfluous components are appropriately shrunk to zero, automatically adapting to the true rank.
Further, we prove posterior contraction around the true parameter, a central limit theorem for point estimators, and a Bernstein-von Mises result characterising the asymptotic behaviour of the induced posterior distribution on the entries of the covariance. To our knowledge, this provides one of the first formal frequentist theoretical justifications for empirical Bayes inference in overfitted  factor models. We show substantial gains in accuracy and uncertainty quantification in numerical experiments and demonstrate utility on a  gene expression dataset.

\section{Method}

\subsection{Notation}
For a matrix $A$, we denote by $a_r$ and $a^{(c)}$ its $r$-th row and $c$-th column, respectively. We denote by $s_l(A)$ the $l$-th largest singular value of $A$ and by $||A||$ and $||A||_\infty$ its spectral and entrywise infinity norms, respectively. For a vector $v$, we denote by $||v||$ its Euclidean norm. For two sequences $(a_n)_{n \geq 1}$ and $(b_n)_{n \geq 1}$, we write $a_n \lesssim b_n$, if $a_n \leq C b_n$ eventually for some $C < \infty $.

\subsection{General methodology}\label{subsec:general_methodology}
We first rewrite the model in its matrix form:
\begin{equation}
    Y = M \Lambda^\top + E, \quad \text{where} \quad Y= \begin{bmatrix}
        y_1, \cdots, y_n
    \end{bmatrix}^\top, \quad M = \begin{bmatrix}
        \eta_1, \cdots, \eta_n
    \end{bmatrix}^\top, \quad E= \begin{bmatrix}
        \epsilon_1, \cdots, \epsilon_n
    \end{bmatrix}^\top
\end{equation}

Let $D_{1:H} = \text{diag}(d_1, \dots, d_H)$ where $d_l = s_l(Y)$ is the $l$-th largest singular value of $Y$, and $U_{1:H} \in \mathbb R^{n \times H}, V_{1:H} \in \mathbb R^{p \times H}$ are the matrices of the corresponding left and right singular vectors, respectively. We select $H$ as a conservative upper bound to the true latent dimension $k$. More details are provided in Section \ref{subsec:hyperparams}. As described in the introduction, 
a spectral estimator for the latent factors is given by the matrix of left singular vectors associated to the leading singular values scaled by $\sqrt{n}$ \citep{bai_03}. Following this approach, we let $\hat M = \sqrt{n} U_{1:H}$ as our initial estimate of the latent factors. 
Given $\hat M$, as in \citet{fable}, we infer the loadings via the surrogate regression 
\begin{equation}\label{eq:surrogate_regression}
    Y = \hat M \tilde \Lambda^\top + \tilde E, \quad \tilde E = \begin{bmatrix}
        \tilde e_1 & \cdots & \tilde e_n 
    \end{bmatrix}^\top, \quad \tilde e_i \sim N(0, \tilde \Sigma), 
\end{equation}
introducing the new parameters $\tilde \Lambda \in \mathbb R^{p \times H}$ and $\tilde \Sigma = \text{diag}\big(\tilde \sigma_1^2, \dots, \tilde \sigma_p^2\big)$. We specify a conjugate normal-inverse gamma prior on the loadings and residual error variance for the $j$-th outcome:
\begin{equation}\label{eq:prior}
   \tilde \lambda_j \mid \tilde \sigma_j^2, \tau_j^2, \tilde \Psi \sim N_H(0, \tilde \sigma_j^2 \tau_j^2 \tilde \Psi), \quad \tilde \sigma_j^2 \sim IG(v/2, vs^2/2), \quad (j = 1, \dots, p),
\end{equation}
where $\Psi = \text{diag}(\psi_1^2, \dots, \psi_H^2)$ with $\psi_l^2$ controlling the strength of the $l$-th latent component globally, while $\tau_j^2$ modulates the signal-to-noise ratio for the $j$-th outcome. This allows latent component and outcome dependent shrinkage. In the following, we describe a data-adaptive procedure for tuning such hyperparameters. Crucially, we set $\psi_l$'s to decreasing values so that the shrinkage increases with the latent component dimension. In Section \ref{sec:theory}, we show that when $H > k$, our estimates of the $\psi_l$'s for $l>k$ converge to $0$, appropriately shrinking superfluous components and automatically adapting to the true rank. The role of the $\tau_j^2$'s is to capture heterogeneity in signal-to-noise ratios across variables. For instance, in genomics it is common that some genes display a strong signal, while others have larger idiosyncratic noise. 
Under the prior specification in \eqref{eq:prior}, we obtain the following conjugate update:
\begin{equation}\label{eq:posterior_update}
    \begin{aligned}
       \tilde \lambda_j \mid \tilde \sigma_j^2, Y, \hat M, \tau_j, \Psi &\sim N_H\big(\hat \lambda_j, \tilde \sigma_j^2 \Psi_{n,j}^{-1}\big),\\
        \tilde \sigma_j^2 \mid Y, \hat M, \tau_j, \Psi &\sim IG\Big(\frac{v + n}{2}, \frac{v s^2 + ||y^{(j)} - \hat M \hat \lambda_j||^2}{2}\Big),
    \end{aligned}\quad (j=1, \dots, p),
\end{equation}
where 
\begin{equation}\label{eq:posterior_params}
       \hat \lambda_j = E[\lambda_j \mid Y, \hat M, \tau_j, \Psi] = \sqrt{n}  \Psi_{n,j}^{-1} D_{1:H} v_{j, 1:H} , \quad \Psi_{n,j} =\text{diag}\big(n+\tau_j^{-2} \psi_1^{-2}, \dots, n+\tau_j^{-2} \psi_H^{-2}  \big),
\end{equation}
and $v_{j, 1:H}$ denotes the $j$-th row of $V_{1:H}$.
We denote by $\hat \Lambda$ the matrix whose $j$-th row is given by $\hat \lambda_j$, that is $\hat \Lambda = [\hat \lambda_1 \cdots \hat \lambda_p]^\top$. 
The posterior mean is equivalent to the solution of the following least squares problem with adaptive $L_2$ regularisation:
\begin{equation}\label{eq:l2_ols}
    \hat \lambda_j = \arg_{\beta \in \mathbb R^{H}}\min ||y^{(j)} - \hat M \beta||^2 + \sum_{l=1}^{H} \frac{1}{\tau_j^2\psi_l^2}\beta_l^2, \quad \text{where} \quad \beta = (\beta_1, \dots, \beta_H)^\top.
\end{equation}
Equation \eqref{eq:l2_ols} shows that our point estimator corresponds to a ridge estimator in the surrogate regression with component-wise penalties, yielding adaptive shrinkage across both variables and latent directions.

\eqref{eq:posterior_update} induces a posterior distribution on the entries of $\Lambda \Lambda^\top$, which suffers from slight undercoverage. To solve a related problem,  \citet{fable} use a posterior variance 
inflation factor $\rho > 1$ having a closed form formula that guarantees asymptotic valid frequentist coverage on average across entries of $\Lambda \Lambda^\top$. The  coverage corrected posterior distribution for $\tilde \lambda_j$ is 
\begin{equation}\label{eq:posterior_update_cc}
    \begin{aligned}
        \tilde \lambda_j &\mid \rho^2,  \tilde \sigma_j^2, Y, \hat M, \tau_j, \Psi \sim N_H(\hat \lambda_j, \rho^2 \tilde \sigma_j^2 \Psi_{n,j}^{-1}).
    \end{aligned}
\end{equation}
Let $\tilde \Lambda = \begin{bmatrix}
    \tilde \lambda_1, \cdots, \tilde \lambda_p
\end{bmatrix}^\top$ and $\tilde \Sigma = \text{diag}(\tilde \sigma_1^2, \cdots, \tilde \sigma_p^2)$ be posterior samples for the loading matrix and residual variances, respectively. We denote by $\tilde L = \tilde \Lambda \tilde \Lambda^\top$ and $\tilde \Theta = \tilde L + \tilde \Sigma$ the corresponding samples for the low-rank component and the overall covariance, respectively.

The spectral estimator for $\Lambda$, $V_{1:H} D_{1:H} / \sqrt{n}$, corresponds to the ordinary least squares estimate of $\tilde \Lambda$ in the surrogate regression \eqref{eq:surrogate_regression}. Taking a Bayesian approach has three major advantages: (i) it leads to a regularized point estimate, which can have superior estimation accuracy, (ii) the prior specification in \eqref{eq:prior} allows for a simple yet effective hyperparameter specification detailed in Section \ref{subsec:hyperparams}, (iii) the induced posterior distribution can be used to propagate uncertainty about any functional of the parameters of interest.

Next, we can estimate $\Lambda \Lambda^\top$ via $\hat L = \hat \Lambda \hat \Lambda^\top$ and the covariance matrix via $\hat L + \hat \Sigma$, where $\hat \Sigma = \text{diag}(\delta_1^2, \dots, \delta_p^2)$, with $\delta_j^2 = \frac{v s^2 + ||y^{(j)} - \hat M \hat \lambda_j||^2}{v+n-2}$ the posterior mean for the $j$-th outcome error variance for $j=1, \dots, p$. We study asymptotic properties of this estimate in Section \ref{sec:theory}. Alternatively, one could estimate $\Lambda \Lambda^\top$ via its posterior mean arising from \eqref{eq:posterior_update_cc}, that is $E[\tilde \Lambda \tilde \Lambda^\top \mid \rho^2, Y, \hat M, \{\tau_j\}_{j=1}^p, \Psi] = \hat L + \rho^2\frac{\text{tr}(\Psi)}{n} \hat \Sigma \mathcal T$, with $\mathcal T = \text{diag}(\tau_1^2, \dots, \tau_p^2)$. 
In practice, we find negligible differences between this posterior mean and $\hat L$. 

While $\hat M$ serves as an initial estimate for $M$, posterior quantities of $M$, including posterior mean and interval estimates, can be easily obtained from the full conditional of $M$, using posterior samples of $\Lambda$ and $\Sigma$. For instance, for given samples of $\tilde \Lambda$ and $\tilde \Sigma$, we can sample from the posterior distribution of the latent factors of the $i$-th sample (the $i$-th row of $M$), via 
\begin{equation}\label{eq:eta_full_cond}
    \eta_i \mid y_i, \tilde \Lambda, \tilde \Sigma \sim N_H\bigg( \big(\tilde \Lambda^\top \tilde \Sigma^{-1} \tilde \Lambda + I_H\big)^{-1}\tilde \Lambda^\top \tilde \Sigma^{-1}y_i,   \big(\tilde \Lambda^\top \tilde \Sigma^{-1}  \tilde \Lambda + I_H\big)^{-1} \bigg).
\end{equation}
Given $N_{MC}$ samples of $\tilde \Lambda$ and $\tilde \Sigma$, $\{\tilde \Lambda^{(s)}\}_{s=1}^{N_{MC}}, \{\tilde \Sigma^{(s)}\}_{s=1}^{N_{MC}}$, an approximation to the posterior mean for $\eta_i$ is available via
\begin{equation}\label{eq:eta_post_mean_approx}
     E[\eta_i \mid y_i] \approx \frac{1}{N_{MC}}\sum_{s=1}^{N_{MC}}   \big(\tilde \Lambda^{(s)\top}\tilde \Sigma^{(s)-1} \tilde \Lambda^{(s)} + I_H\big)^{-1}\tilde \Lambda^{(s)\top}\tilde \Sigma^{(s)-1}  y_i.
\end{equation}
The estimate in \eqref{eq:eta_post_mean_approx} is a refinement of the $i$-th row of the initial estimate $\hat M$, that takes into account the heteroscedasticity across variables and uncertainty in model parameters. Sampling from \eqref{eq:eta_full_cond} allows full uncertainty quantification for the latent factors and differently from Gibbs sampling does not require any post-processing routine to resolve the rotational ambiguity of the Monte Carlo samples.


\subsection{Hyperparameter selection}\label{subsec:hyperparams}
We first describe a strategy for tuning the prior shrinkage hyperparameters. The prior distributions depend on the two sets of hyperparameters via the products $\{\tau_j^2 \Psi\}_{j=1}^p$. Dividing the $\tau_j^2$'s by some constant $c > 0$, and multiplying $\Psi$ by the same constant, we obtain the same specification. Hence, we impose the following normalizing condition: $\text{tr}(\Psi) = H$.
The expected \textit{a priori} squared norms of $\tilde \lambda_j$ and $\tilde  \lambda^{(l)}$, the $j$-th row and $l$-th column of $\tilde  \Lambda$, respectively, are given by $E[||\tilde  \lambda_j||_F^2 \mid \tilde \sigma_j^2, \tau_j^2, \Psi] = \tilde \sigma_j^2 \tau_j^2 \text{tr}(\Psi) = \tilde \sigma_j^2 \tau_j^2 H$, and $E[||\tilde \lambda^{(l)}||_F^2 \mid \{\tilde \sigma_j^2, \tau_j^2\}_{j=1}^p, \psi_l^2] = \psi_l^2 \sum_{j=1}^p\tilde \sigma_j^2 \tau_j^2$. We select hyperparameters via an empirical Bayes moment-matching strategy, calibrating prior expected magnitudes of loadings to their spectral estimates while correcting for overestimation of the latent dimension.
We estimate $\tilde \sigma_j^2$ via $\hat \sigma_j^2 = ||(I_n - U_H U_H^\top) y^{(j)}||_F^2 / (n-H) \vee 
c_{\sigma}$, where $c_{\sigma}$ is some small constant, $||\tilde \lambda_j||_F^2$ via $v_{j, 1:H}^\top (D_{1:H}^2 - d_{H+1}^2 I_H) v_{j, 1:H} / n$, and $||\tilde \lambda^{(l)}||_F^2$ via $(d_l^2 - d_{H+1}^2) / n$. 
Plugging everything together, we obtain 
\begin{equation}\label{eq:hyperparams}
\begin{aligned}
    \hat \psi_l^2 &= \frac{d_l^2 - d_{H+1}^2}{\sum_{l=1}^H(d_l^2 - d_{H+1}^2)} H, \quad &(l=1, \dots, H), \\
    \hat \tau_j^2 &= \frac{v_{j, 1:H}^\top (D_{1:H}^2 - d_{H+1}^2 I_H) v_{j, 1:H} }{n H \hat \sigma_j^2}, \quad &(j=1, \dots, p).
\end{aligned}
\end{equation}
These choices are motivated by the following rationale: if we neglected the term $d_{H+1}^2$ in \eqref{eq:hyperparams}, we would obtain  $E[||\tilde \lambda_j||^2 \mid \hat \sigma_j^2, \hat \tau_j^2, \Psi] = ||\hat \lambda_{svd, j}||^2$, and $E[||\tilde \lambda^{(l)}||_F^2 \mid \{\hat \sigma_j^2, \hat \tau_j^2\}_{j=1}^p, \hat \psi_l^2] = ||\hat \lambda_{svd}^{(l)}||_F^2$, with $\hat \lambda_{svd, j}$ and $\hat \lambda_{svd}^{(l)}$ denoting the $j$-th row and $l$-th column of $\hat \Lambda_{svd} = V_{1:H} D_{1:H} / \sqrt{n}$, the spectral estimate for $\Lambda$ with latent dimension $H$. However, this estimate does not take into account the fact that the true latent dimension is being overestimated. The subtraction of $d_{H+1}^2$ removes the bulk noise contribution associated with overfitted components. Under spiked covariance models, the first few singular values beyond the true rank concentrate around the upper edge of the Marchenko–Pastur distribution \citep{MP_law}. Hence, $d_{H+1}^2$ provides an estimate of the noise level, and subtracting it yields a better estimate of signal strength. A complete derivation is reported in the supplementary material.

Hyperparameter estimates have intuitive forms. $\hat \psi_l^2$ is proportional to the relative contribution of variability explained by the $l$-th latent component after subtracting the variability explained by the $H+1$ component.  $\hat \tau_j^2$ is proportional to the ratio of the amount of variability explained by the $H$ latent components in the $j$-th outcome ($v_{j, 1:H}^\top (D_{1:H}^2 - d_{H+1}^2 I_H) v_{j, 1:H} $) and an estimate of the residual error variance ($\hat \sigma_j^2$), allowing stronger shrinkage for noisy variables and weaker shrinkage for informative ones. We characterize the asymptotic behavior of such estimates in Section \ref{sec:theory} and show that $\psi_l^2 \overset{pr}{\to} 0$, whenever the $l$-th latent component is superfluous, ensuring automatic rank adaptation. As a result, leading singular directions receive moderate shrinkage, while components beyond the true rank are aggressively shrunk toward zero.

Here, we conduct a simple numerical experiment to illustrate the adaptivity of our shrinkage hyperparameters. We generate data via the scenario (a) described in section \ref{sec:numerical_experiments} with $n=1000$, $p=5000$.  
For simplicity, we set $H = 15$, while the true rank $k$ is 10.
\begin{figure}[t]
    \centering
    \begin{subfigure}{0.475\linewidth}
        \centering
        \includegraphics[width=\linewidth]{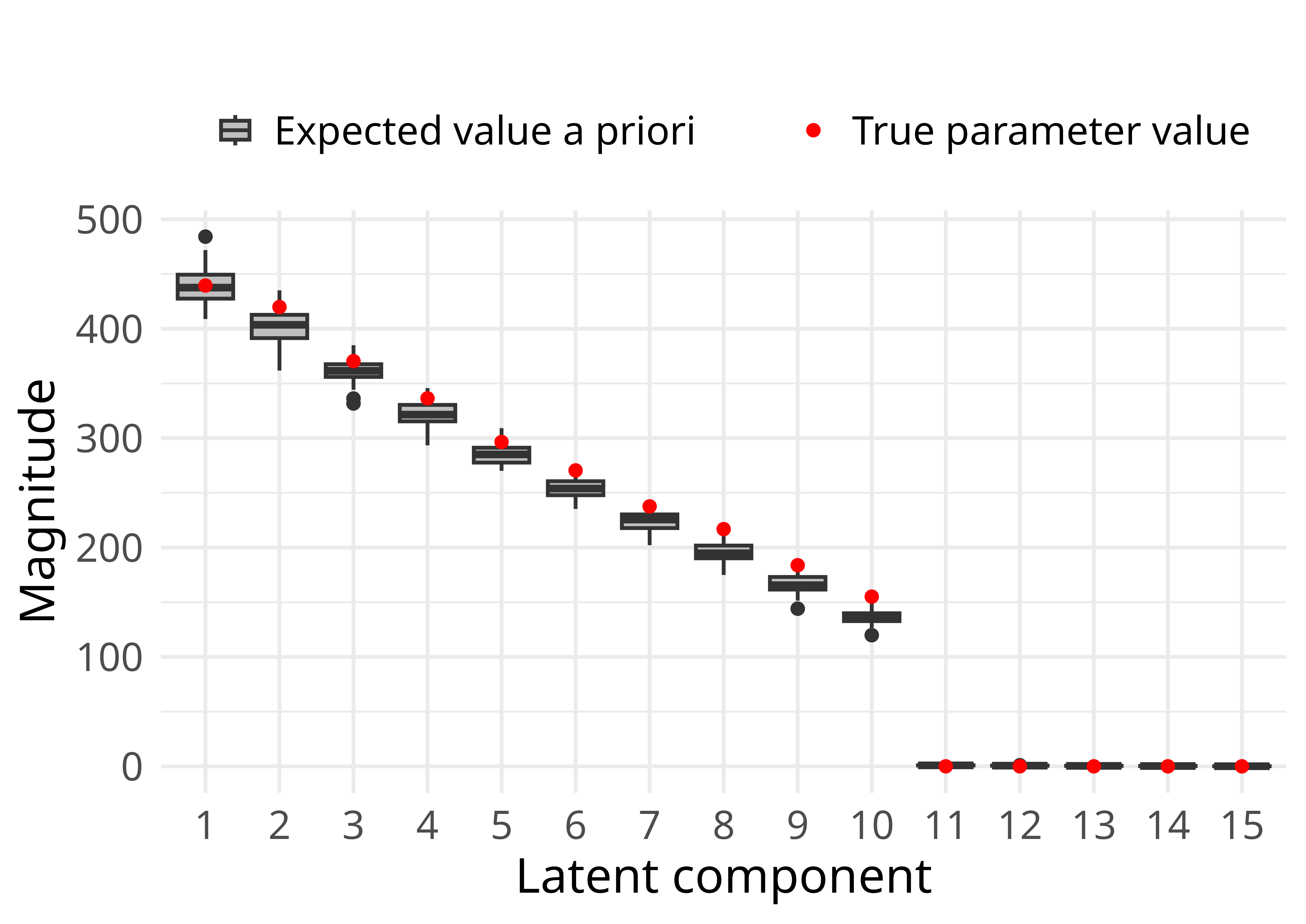}
    \end{subfigure}
    \hfill
    \begin{subfigure}{0.475\linewidth}
        \centering
        \includegraphics[width=\linewidth]{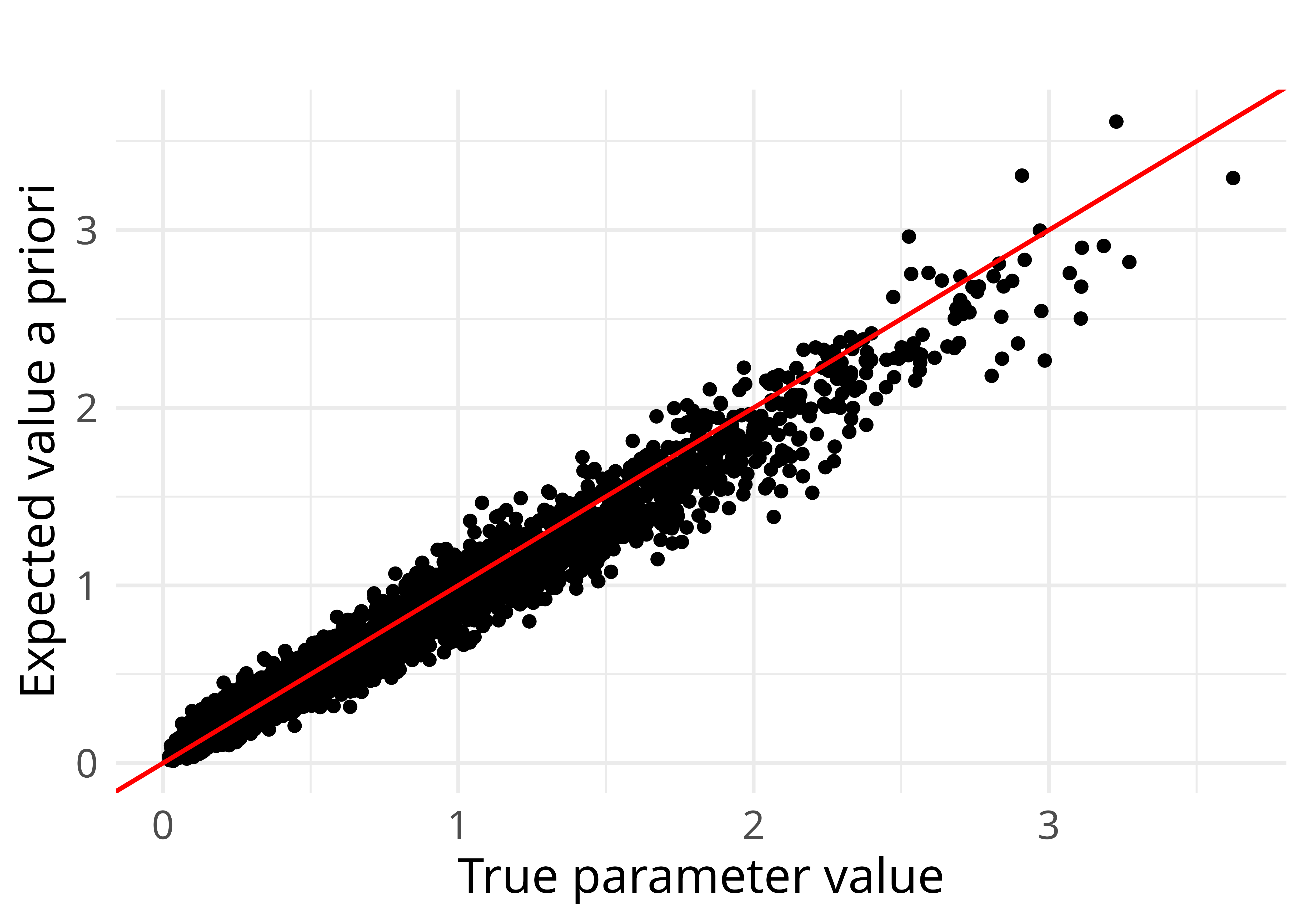}
    \end{subfigure} 
    \caption{Left panel: prior expected squared magnitude (across 50 replications) and true squared magnitude of the columns of the loading matrix. Right panel: squared magnitude of the rows of the loading matrix (x-axis) vs their prior expected value (y-axis) in one replication. }
    \label{fig:ada_shrinkage}
\end{figure}
In the left panel of Figure \ref{fig:ada_shrinkage}, histograms show $E[||\tilde \lambda^{(l)}||^2  \mid \{\hat \sigma_j^2, \hat \tau_j^2\}_{j=1}^p, \hat \psi_l^2 ] = \hat \psi_l^2 \sum_{j=1}^p \hat \tau_j^2 \hat \sigma_j^2$, the prior expected value of the squared magnitude of the $l$-th column of $\tilde \Lambda$ conditionally on the point estimates for the variances $\{\hat \sigma_j^2\}_{j=1}^p$, and the hyperparameters $\{\hat \tau_j^2\}_{j=1}^p$ and $\hat \psi_l^2$, for each component $l=1, \dots, H$ across 50 independent replications, while red dots indicate the squared singular values of the true loading matrix. First, we note that for the active components ($l=1, \dots, 10$), the prior expectation of the squared magnitude of each column is very close to the corresponding squared singular value, demonstrating that our tuning procedure correctly adapts to the amount of signal in the data. Moreover, for the superfluous components ($l=11, \dots, 15$), the prior expectations are very low and such components are successfully shrunk to $0$. The right panel compares the squared magnitude of the rows of $\Lambda$ (x-axis) with their respective prior expectations $\big\{E[||\tilde \lambda_j||^2 \mid \hat \sigma_j^2, \hat \tau_j^2, \{\hat \psi_l\}_{l=1}^H]\big\}_{j=1}^p$ (y-axis) for one random replication of the experiment, providing evidence of the ability of our procedure to adapt to each outcome signal-to-noise ratio.

As for the selection of the latent dimension, various criteria have been proposed, which can be divided into two categories: (a) information criteria \citep{ic_bn, ic_ah, chen_jic} and (b) \say{elbow methods}, based on the comparison of ratio or differences of adjacent pairs of (squared) singular values \citep{Fan2023BridgingFactor}. Our methodology relies on setting $H$ to some overfitted value. Theoretical results presented in Section \ref{sec:theory} require $H = \lesssim p^{2/5}$, however, in practice, we found that setting $H = \lfloor \sqrt{p} \rfloor$ or $H = \widehat{k}_{\mathrm{JIC}}+10$, where $\widehat{k}_{\mathrm{JIC}}$ is found minimizing the joint information criterion by \citet{chen_jic} to have good empirical performance in our numerical experiments and application.

\section{Theoretical Support}\label{sec:theory}
Prior to stating the results, we enumerate some regularity conditions. 
\begin{assumption}\label{assumption:model} The data are generated from model \eqref{eq:fact_model_int}, and the true loading matrix and residual error variance are denoted by $\Lambda_{0} = \begin{bmatrix}
    \lambda_{01}, \cdots, \lambda_{0p} 
\end{bmatrix}^\top$ and $\Sigma_0$, respectively. 
We also let $L_0 = \Lambda_0 \Lambda_0^\top$ and $\Theta_0 = L_0 + \Sigma_0$ be the true low-rank components and the true covariance matrix, respectively. 
\end{assumption}

\begin{assumption}\label{assumption:dimensions}
    We have $p\to \infty$ as $n \to \infty$ with $\frac{p}{n} \to c \in (0, \infty)$.
\end{assumption}

\begin{assumption}\label{assumption:Lambda}
The true loading matrix is such that $s_l(\Lambda_{0}) \asymp \sqrt{p}$ for $l=1, \dots, k$, $||\Lambda_{0}||_\infty \leq C_\Lambda <  \infty$ and $\min_{j=1, \dots, p} ||\lambda_{0j}|| \geq c_{\Lambda} >0$.    
\end{assumption}

\begin{assumption}\label{assumption:homoscedasticity}
    The idiosyncratic error variance is spherical, that is $\Sigma_0 = \sigma_0^2 I_p$, with $\sigma_0^2 < \infty$.
\end{assumption}


The first result establishes automatic rank adaptation. Although we do not explicitly model the latent dimension, the empirical Bayes variance components associated with redundant components vanish asymptotically.

\begin{proposition}\label{prop:psi_l}
    Suppose $H>k$, $H = \mathcal{O}(n^{2/5})$ and Assumptions \ref{assumption:model}--\ref{assumption:homoscedasticity} hold. Then, for each $k < l \leq H$, with probability at least $1-o(1)$, we have 
    \begin{equation*}
        \max_{j=1, \dots, p} \hat \psi_l^2 \hat \tau_j^2 \lesssim \frac{1}{n^{1 + 2/5}},
    \end{equation*}
     where $\hat \psi_l^2$ and $\hat \tau_j^2$ are defined in Section \ref{subsec:hyperparams}. 
\end{proposition}
Proofs of all theoretical results are reported in the supplemental. Proposition \ref{prop:psi_l} states that as the sample size and the outcome dimension diverge and the data are generated according to the model \ref{eq:fact_model}, the prior variance parameters for superfluous latent components (that is, those corresponding to indices strictly larger than $k$) converge to zero at a fast rate. As a result, the estimates for the corresponding elements in the loading matrix converge to $0$, which is formalized in the following corollary.
\begin{corollary}\label{corr:norm_Lambda_l}
    Under the assumptions of Proposition \ref{prop:psi_l}, for each $k < l \leq H$, we have 
     \begin{equation*}
     \big \| \hat \lambda_{j}^{(l)} \big \| \overset{pr}{\to} 0,
    \end{equation*} 
      as $n\to \infty$,  where $\hat \lambda_{j}^{(l)}$ is the $l$-th column of $\hat \Lambda$, which is defined in Section \ref{subsec:general_methodology}.
\end{corollary}

Next, we study the asymptotic performance of our estimator in recovering the covariance matrix. We find similar consistency and distributional results to the ones in \citet{fable}, without assuming that the latent dimension is known a priori, allowing the upper bound of the latent dimension $H$ to grow with $n$, and taking into account the effect of estimating the empirical Bayes hyperparameters from the same data.

\begin{theorem}\label{thm:consistency}
Let $\hat L = \hat \Lambda \hat \Lambda^\top$ and $\hat \Theta = \hat L + \hat \Sigma$ be the point estimates for the low-rank component and the overall covariance matrix defined in Section \ref{subsec:general_methodology}. Then, if $H = \mathcal O(n^{2/5})$ and Assumptions \ref{assumption:model}--\ref{assumption:homoscedasticity} hold, with probability at least $1-o(1)$, we have 
\begin{equation}
    \begin{aligned}
    \frac{||\hat L - L_0||}{||L_0||} & \lesssim \sqrt{\frac{\log n}{n}}, \quad \frac{||\hat \Theta - \Theta_0||}{||\Theta_0||} & \lesssim \sqrt{\frac{\log n}{n}}.
    \end{aligned}
\end{equation}
Moreover, with probability at least $1-o(1)$, 
\begin{equation}
    \begin{aligned}
E\left[  \Pi \left( \frac{||\tilde L - L_0||}{||L_0||} \geq C_1\sqrt{\frac{\log n}{n}} \right)\right]& \to 0, \quad E\left[  \Pi \left(\frac{||\tilde \Theta - \Theta_0||}{||\Theta_0||}  \geq C_2\sqrt{\frac{\log n}{n}} \right)\right]& \to 0,
    \end{aligned}
\end{equation}
for some constants $C_1, C_2 < \infty$.
\end{theorem}
In Theorem \ref{thm:consistency}, we consider the error in spectral norm in estimating the low-rank component of the covariance and the overall covariance matrix. To make results comparable across dimension, we divide the errors by the norms of the true parameter. 
The first result justifies the use of $\hat L$ and $\hat \Theta$ as point estimators for $\Lambda \Lambda^\top$ and $\Omega$. The second result is stronger and characterises the contraction of the entire posterior distribution. 
The central idea of the proof is to use Proposition \ref{prop:psi_l}, along with appropriate bounds on the $\hat \psi_l^2$'s for $l=1, \dots, k$ and on the $\hat \tau_j^2$'s, to show that asymptotically $\hat \Lambda \hat \Lambda^\top$ behaves similarly to the oracle spectral estimator, the one that selects the correct latent dimension, as superfluous components are appropriately shrunk to 0. Importantly, this result does not require $H$ to be fixed but allows it to grow at a rate $n^{2/5}$.
\begin{remark}
 Although our estimator behaves asymptotically as the spectral estimator with the correct latent dimension, our numerical experiments show that it may even perform better in finite samples due to the effectiveness of the adaptive shrinkage.
\end{remark}
\begin{remark}
  Theorem \ref{thm:consistency} requires $\Sigma_0 = \sigma_0^2 I_p$, however, we relax this assumption letting $\Sigma_0$ be a diagonal matrix in the numerical experiments, where our estimator achieves state-of-the-art performance. 
\end{remark}
\begin{remark}
    Here, we require $H \lesssim n^{2/5} \asymp p^{2/5}$ to obtain theoretical support, but, in the numerical experiments, we also consider $H = \lfloor \sqrt{p} \rfloor$ with excellent empirical performance.
\end{remark}

The last two results concern uncertainty quantification about our estimate. The first one is a central limit theorem showing entry-wise asymptotic normality for the point estimator.
\begin{theorem}\label{thm:clt}
Suppose $H = \mathcal{O}(n^{2/5})$ and Assumptions \ref{assumption:model}--\ref{assumption:homoscedasticity} hold. Let $\hat \theta_{jj'}$ and $\theta_{0jj'}$ be the $j,j'$-th entry of $\hat \Theta$ and $\Theta_0$, respectively, for $j,j' \in \{1, \dots, p\}$. Then, as $n \to \infty$, we have
\begin{equation*}
    \begin{aligned}
      \frac{ \sqrt{n}}{S_{0jj'}} \big(\hat \theta_{jj'} - \theta_{0jj'}\big) \Rightarrow N(0, 1), \quad (j, j'=1, \dots, p),
    \end{aligned}
\end{equation*}
where 
       \begin{equation}\label{eq:S_0_sq}
              S_{0jj'}^2 =    \begin{cases}
                \sigma_{0}^2 \big(||\lambda_{0j}||^2  + ||\lambda_{0j'}||^2\big) + (\lambda_{0j}^\top \lambda_{0j'})^2 + ||\lambda_{0j}||^2 ||\lambda_{0j'}||^2,\quad & \text{if } j \neq j' \\
                   2\big(\sigma_{0}^2 + ||\lambda_{0j'}||^2 \big)^2, \quad & \text{if } j=j'.
                 \\
             \end{cases}
       \end{equation}

\end{theorem}
Theorem \ref{thm:clt} states that the $j,j'$-th entry
    of our point estimator for the covariance is approximately normal distributed centered at the corresponding true value of the parameter with variance $S_{0jj'}^2 / n$, that is, for large $n$, $\hat \theta_{jj'} \overset{\cdot}{\sim} N(\theta_{0jj'}, S_{0jj'}^2 / n)$, where $\overset{\cdot}{\sim}$ reads \say{is approximately distributed as}.
Hence, 
an asymptotically valid $(1-\alpha) \%$ confidence interval for $ \theta_{0jj'}$ is 
\begin{equation}\label{eq:conf_int}
    \begin{aligned}
         \text{CI}_{jj'} = \bigg[\hat \theta_{jj'} \pm z_{1-\alpha/2}\frac{\hat S_{jj'}}{\sqrt{n}}\bigg],
    \end{aligned}
\end{equation}
where $z_{q} = \Phi^{-1}(q)$, $\Phi(\cdot)$ denotes the cumulative distribution function of a standard Gaussian random variable, and $\hat S_{jj'}$ is a consistent estimate of $S_{0jj'}$.
Similarly, we can characterise the asymptotic behaviour of the posterior distribution on the entries of $\tilde \Omega$ induced by  \eqref{eq:posterior_update_cc}. 
\begin{theorem}\label{thm:bvm}
    Suppose $H = \mathcal{O}(n^{2/5})$ and Assumptions \ref{assumption:model}--\ref{assumption:homoscedasticity} hold. 
    Then, as $n \to \infty$, we have
\begin{equation*}
    \begin{aligned}
      \sup_{x \in \mathbb R} \left| \Pi\left\{ \frac{\sqrt{n} (\tilde \theta_{jj'} - \hat \theta_{0jj'})}{T_{0jj'}^2(\rho)} \leq x \right\} - \Phi(x) \right| \overset{pr}{\to} 0, \quad (j, j' = 1, \dots, p),
    \end{aligned}
\end{equation*}
where 
       \begin{equation}\label{eq:S_0_sq}
             T_{0jj'}^2(\rho) =   \begin{cases}
             \rho^2 \sigma_0^2 \big(||\lambda_{0j}||^2  + ||\lambda_{0j'}||^2\big),\quad & \text{if } j \neq j' \\
                 2\sigma_0^4 + 4 \rho^2 \sigma_0^2||\lambda_{0j}||^2, \quad & \text{if } j=j'.
             \end{cases}
       \end{equation}

\end{theorem}
According to Theorem \ref{thm:bvm}, the induced posterior distribution on the $j,j'$-th entry of the covariance after appropriate recentering and rescaling converges to a zero-mean Gaussian with variance $T_{0jj'}^2(\rho)$. Estimating $T_{0jj'}(\rho)$ with $\hat T_{jj'}(\rho)$ by replacing true parameters with their respective point estimates, we obtain an analytic asymptotic approximation to a credible interval for $\theta_{jj'}$, 
\begin{equation}\label{eq:cred_int}
    \begin{aligned}
       \text{CR}_{jj'}(\rho) = \bigg[\hat \theta_{jj'} \pm z_{1-\alpha/2} \frac{\hat T_{jj'}(\rho) }{\sqrt{n} }\bigg].
    \end{aligned}
\end{equation}
Note that $T_{0jj'}(1) \leq S_{0jj'}$, which shows the undercoverage suffered by the posterior without correction induced by \eqref{eq:posterior_update}. We adopt the same strategy of \citet{fable} to tune the inflation term $\rho$ to achieve asymptotically valid coverage on average across entries of $\Theta$.

\section{Numerical experiments}\label{sec:numerical_experiments}


We compare our procedure (denoted by \texttt{EB} in the figures) with four alternatives: \citet{Barigozzi2018ConsistentEO} (\texttt{BC}), the principal component estimator (\texttt{PC}), \citet{rotate} (\texttt{ROTATE}),  \citet{fable} (\texttt{FABLE}) and \citet{sparse_infinite} (\texttt{MGSP}).
\texttt{FABLE}'s implementation provided by the authors estimates the rank via the joint information criterion (JIC, \citet{chen_jic}). \texttt{MGSP} is fit using its default overfitted specification, which adaptively shrinks redundant factors. For all other procedures, we consider four choices: the true latent dimension $k$, $\widehat{k}_{\mathrm{JIC}}$, $\widehat{k}_{\mathrm{JIC}}+10$, where $\widehat{k}_{\mathrm{JIC}}$ is the estimate obtained via the JIC, and $\lfloor\sqrt{p}\rfloor$.  

We consider three generating mechanisms for the loading matrix. 
Scenario (a) represents a standard dense factor model, scenario (b) introduces variable-specific signal heterogeneity, and scenario (c) imposes sparsity in the loading matrix. In particular, letting $s^2 = (s_1^2, \dots, s_k^2)$, with $s_l = 0.15 + 0.1 \frac{l-1}{k-1}$ for $l=1, \dots, k$, we generate the elements of the loading matrix according to one of the following three options:
\begin{enumerate}
    \item [(a)] $\lambda_{0jl} \sim N(0, s_{l}^2)$
    \item [(b)] $\lambda_{0jl} \sim N(0, s_{jl}^2)$, with $s_{jl}^2 = r_j^2s_l^2$ and $r_j^2 = \begin{cases}
        3 \quad \text{if } j \leq \lfloor p/3 \rfloor,\\
                1/3 \quad \text{if } j \geq \lfloor 2p/3 \rfloor +1,\\
        1 \quad \text{otherwise}
            \end{cases}$
            
     \item [(c)] $\begin{cases}
         \lambda_{0jl} \sim N(0, s_{l}^2) \quad \text{if } j \leq \lfloor 2 p/3 \rfloor,\\
         0\quad \text{otherwise}
     \end{cases}$
\end{enumerate}
  for $l=1, \dots, k$ and $j=1, \dots, p$. We generate residual error variances as $\sigma_{0j}^2 \sim \text{Unif}(1, 5)$ for $j=1, \dots, p$. We set $k = 10$ and consider every possible combination of $(n, p) \in \{500, 1000\} \times \{1000, 5000\}$. We evaluate the prediction accuracy via the relative Frobenius error for point estimates of the covariance. 

  \begin{figure}[H]
    \centering
    \begin{subfigure}{0.475\linewidth}
        \centering
        \includegraphics[width=\linewidth]{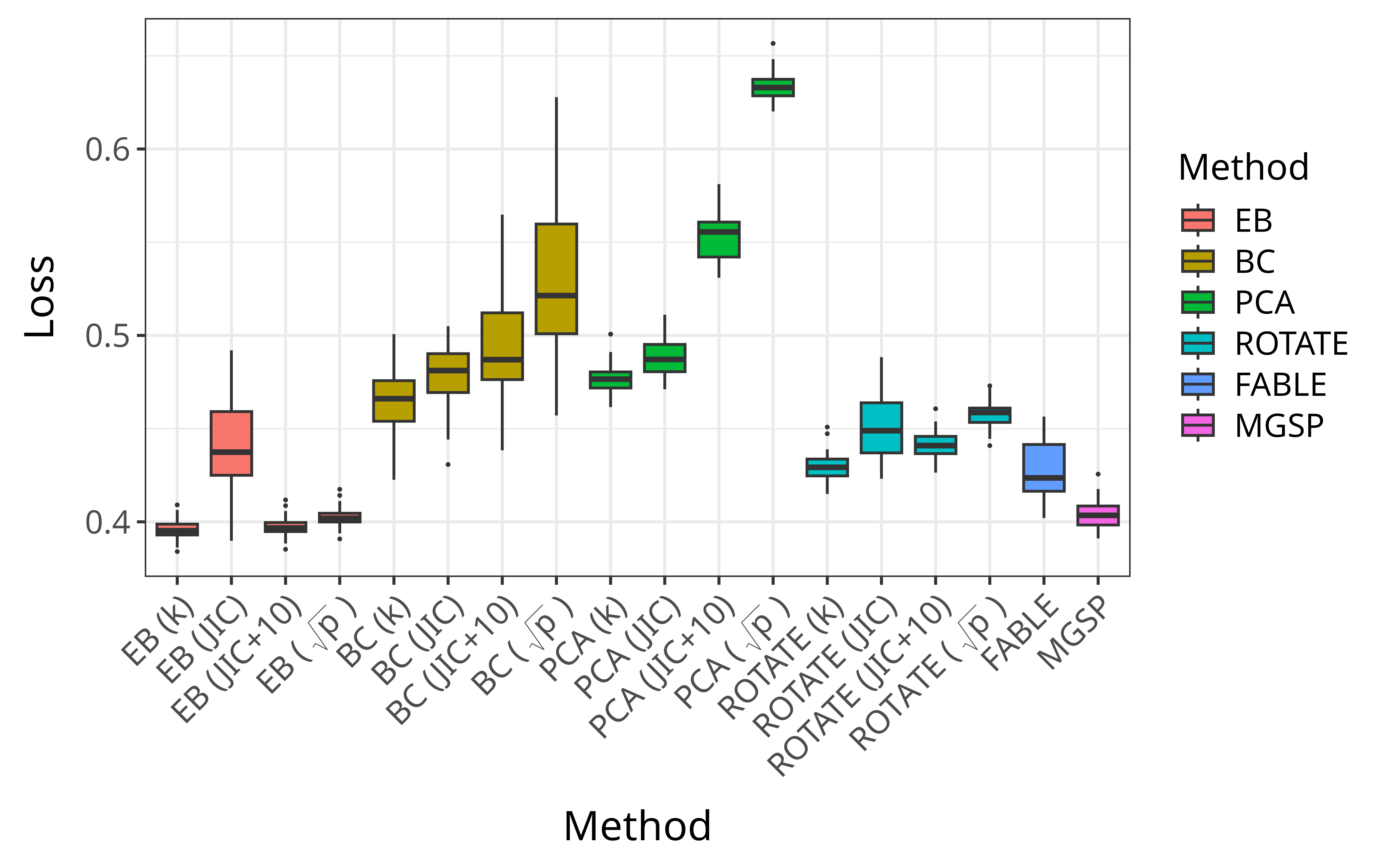}
        \caption{$n=500, p=1000$}
    \end{subfigure}
    \hfill
    \begin{subfigure}{0.475\linewidth}
        \centering
        \includegraphics[width=\linewidth]{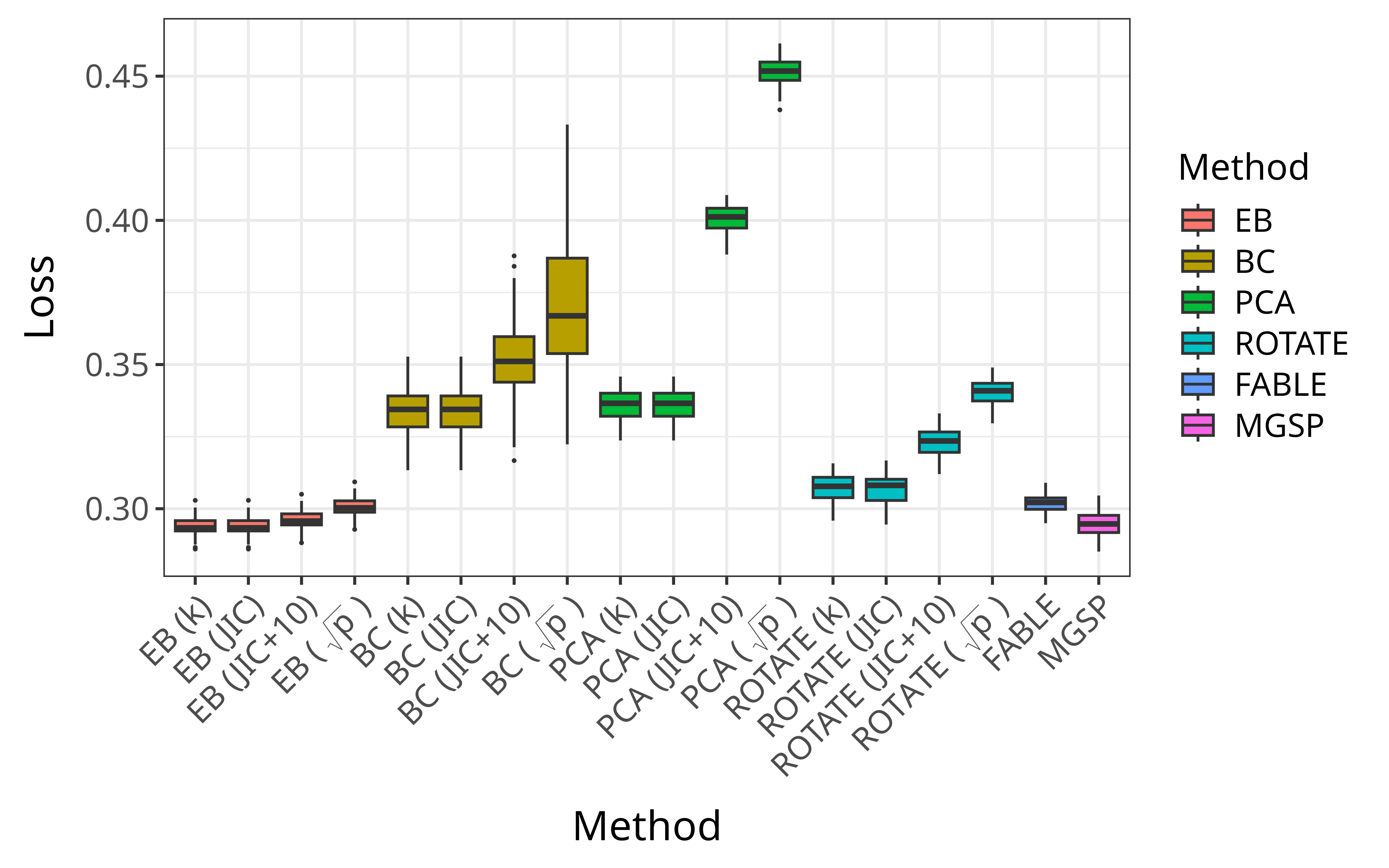}
        \caption{$n=1000, p=1000$}
    \end{subfigure} 
    \vspace{0.5em} 
    \begin{subfigure}{0.475\linewidth}
        \centering
        \includegraphics[width=\linewidth]{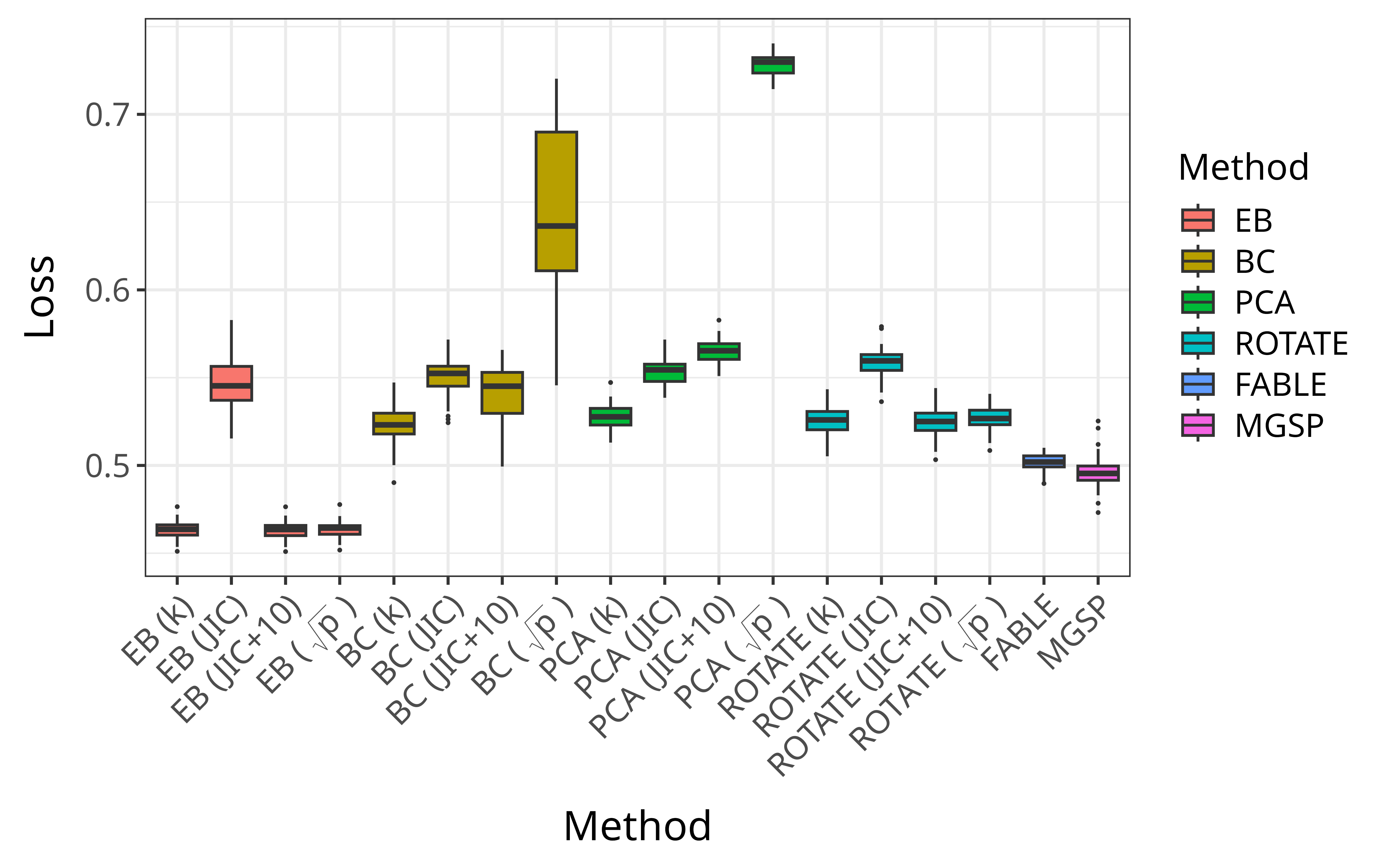}
        \caption{$n=500, p=5000$}
    \end{subfigure}
    \hfill
    \begin{subfigure}{0.475\linewidth}
        \centering
        \includegraphics[width=\linewidth]{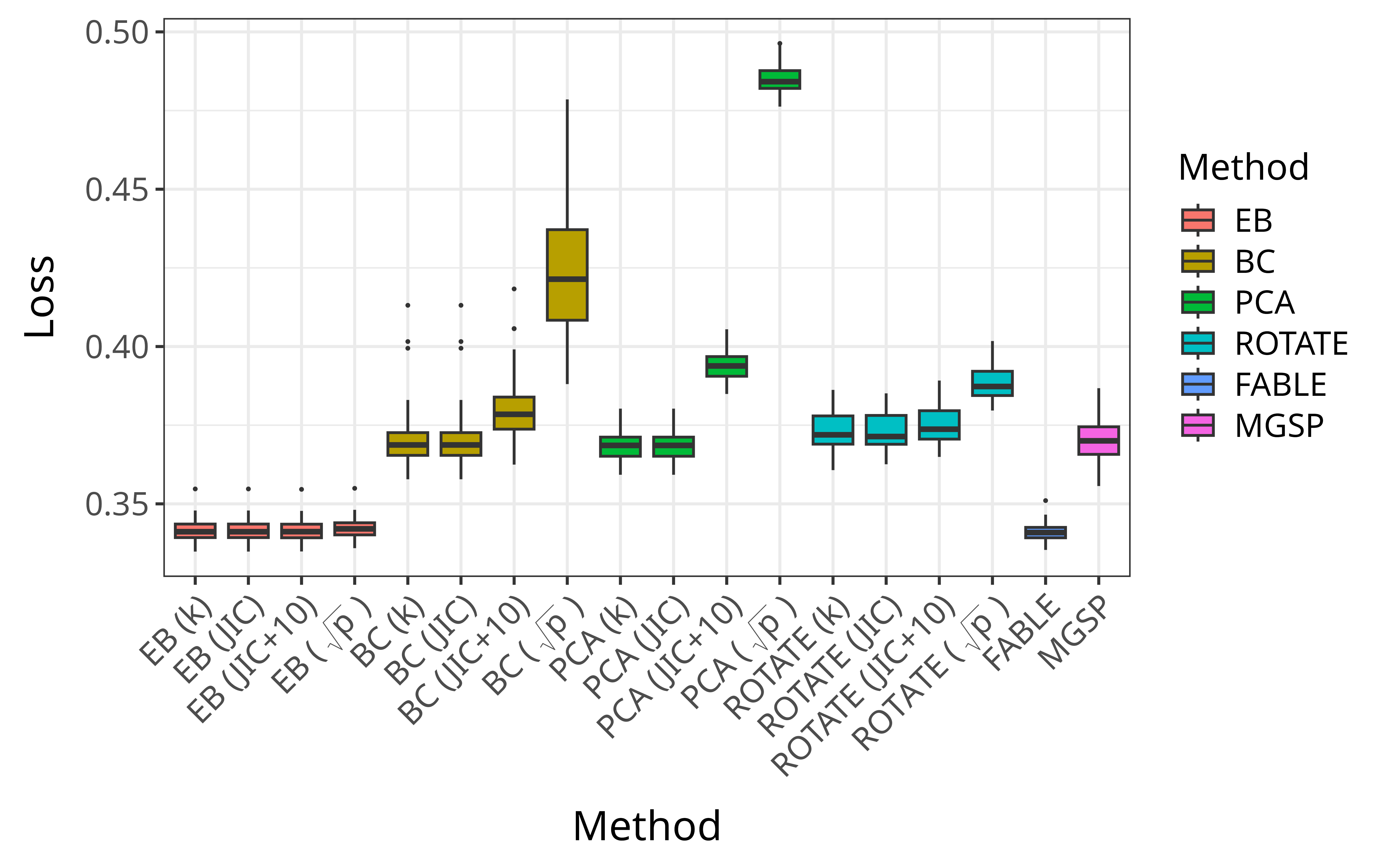}
        \caption{$n=1000, p=5000$}
    \end{subfigure}  
    \caption{Relative Frobenius error in estimating the true covariance matrix under different sample size ($n$) and dimension ($p$) settings in the first scenario.}
    \label{fig:covariance_fr_1}
\end{figure}

Overall, \texttt{EigenBayes} achieves state-of-the-art accuracy in reconstructing the covariance in various settings. Under the first scenario, \texttt{EigenBayes}, \texttt{FABLE}, and \texttt{MGSP} are the best-performing method when $n=1000$, with \texttt{EigenBayes} outperforming the competitors in the lower sample size setting (Figure \ref{fig:covariance_fr_1}). The results are qualitatively similar in the second scenario (Figure \ref{fig:covariance_fr_2}), that is, when the loading magnitudes vary across outcomes, with a slightly more pronounced advantage for \texttt{EigenBayes}. In the sparse setting (\ref{fig:covariance_fr_3}), \texttt{ROTATE}, which explicitly models sparsity in the loading matrix, is the second-best performing method, and is out-performed only by \texttt{EigenBayes}, in particular in the higher-dimensional example. 

Across all simulations, all methods with the latent dimension set to $\widehat{k}_{\mathrm{JIC}}$ (including \texttt{EigenBayes}) perform sub-optimally, when the JIC underestimates the true $k$ (such as when $n=500$). We report the result fitting \texttt{EigenBayes} with $H = \widehat{k}_{\mathrm{JIC}}$ as a reference, as we always recommend setting $H$ to a higher value. Indeed, the accuracy of \texttt{EigenBayes} is remarkably robust to the choice of latent dimension whenever $H$ is chosen to be sufficiently large, with estimates under the most conservative choice ($H = \lfloor\sqrt{p}\rfloor$) performing almost indistinguishably from those having correct latent dimension $(H=k)$, while competitors tend to provide less accurate estimates as the choice of latent dimension increases. Even when $\widehat{k}_{\mathrm{JIC}} = k$, \texttt{EigenBayes} performs just as well and sometimes even better, such as in the sparse setting, than \texttt{FABLE}.

\begin{figure}[h]
    \centering
    \begin{subfigure}{0.475\linewidth}
        \centering
        \includegraphics[width=\linewidth]{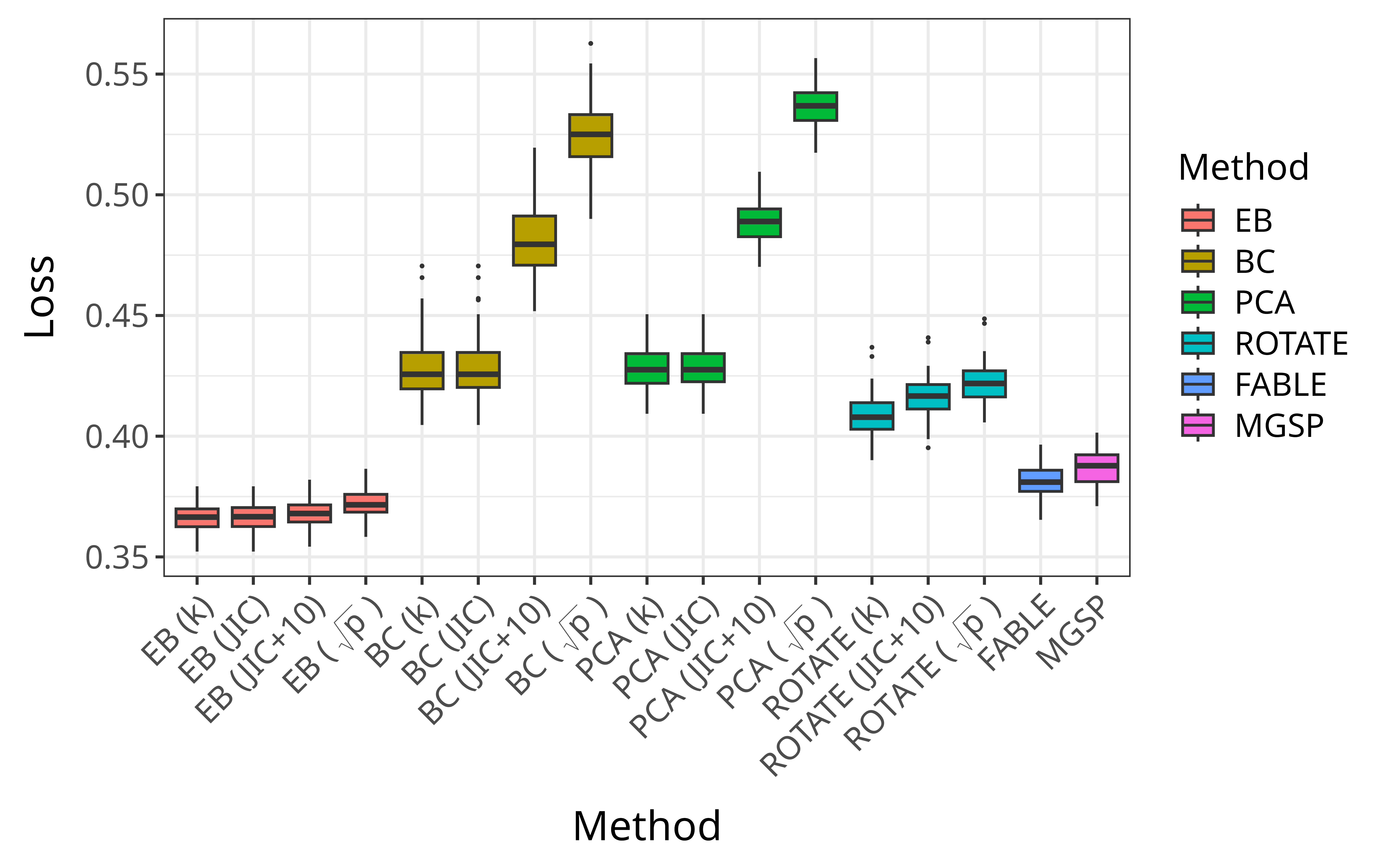}
        \caption{$n=500, p=1000$}
    \end{subfigure}
    \hfill
    \begin{subfigure}{0.475\linewidth}
        \centering
        \includegraphics[width=\linewidth]{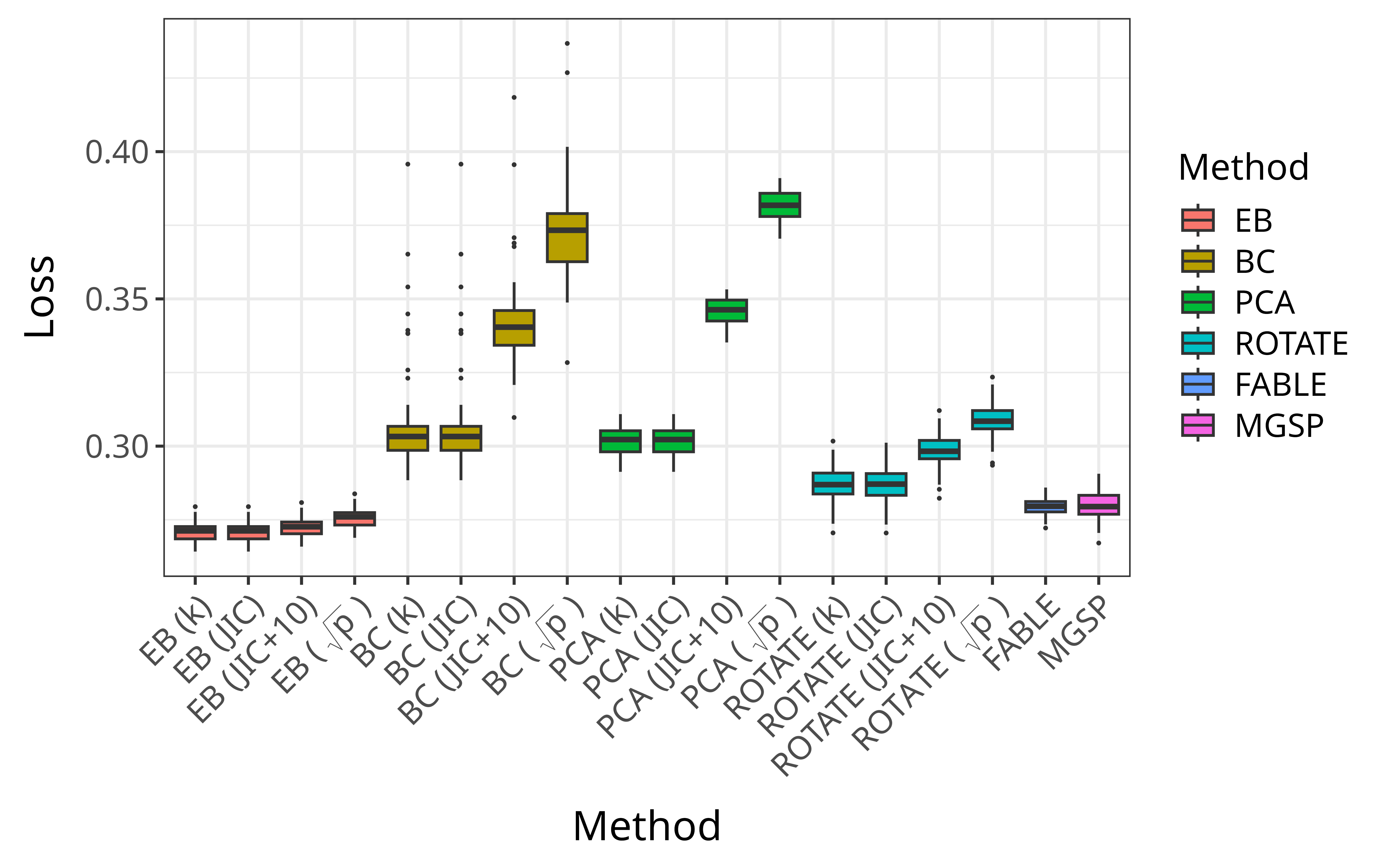}
        \caption{$n=1000, p=1000$}
    \end{subfigure} 
    \vspace{0.5em} 
    \begin{subfigure}{0.475\linewidth}
        \centering
        \includegraphics[width=\linewidth]{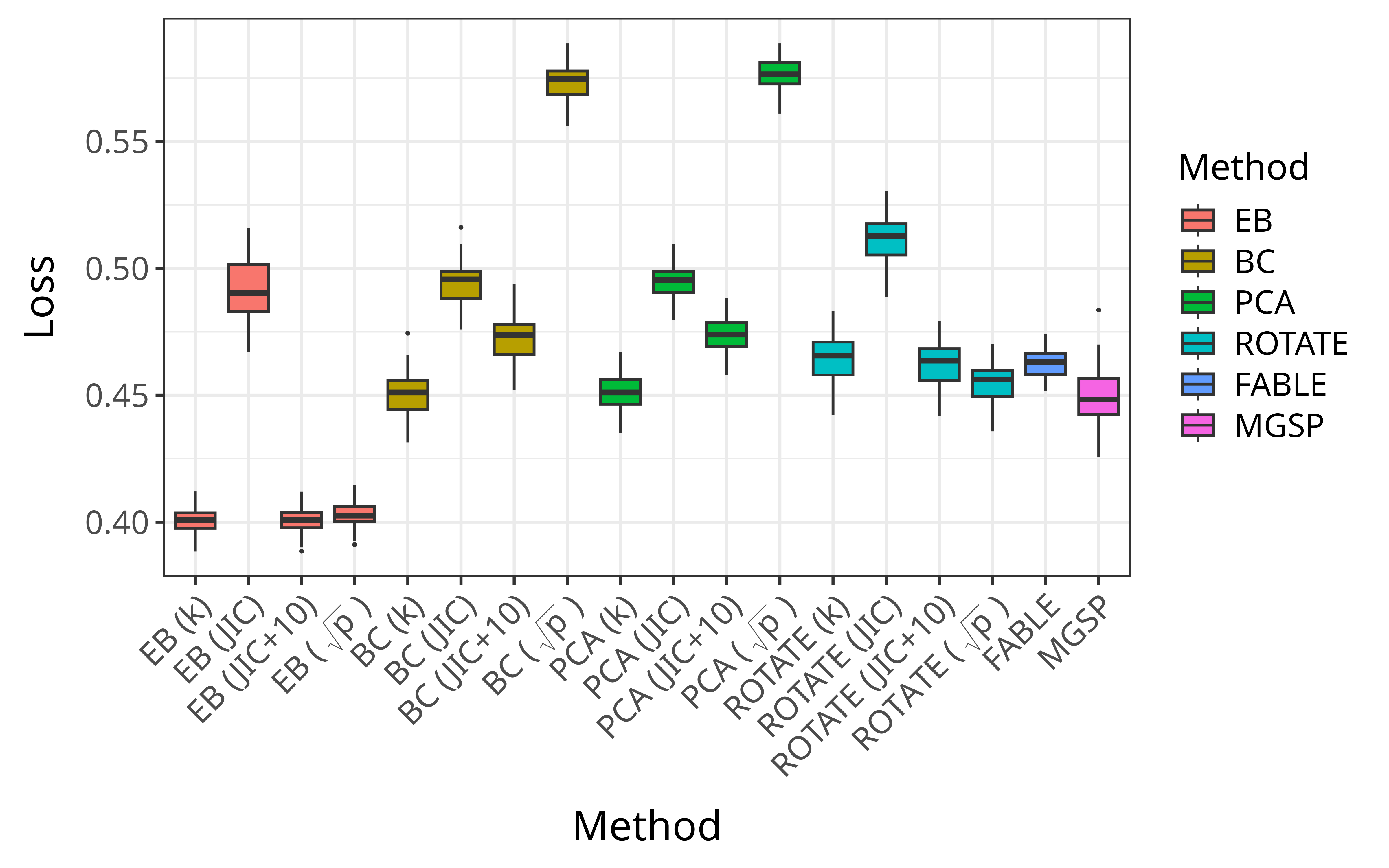}
        \caption{$n=500, p=5000$}
    \end{subfigure}
    \hfill
    \begin{subfigure}{0.475\linewidth}
        \centering
        \includegraphics[width=\linewidth]{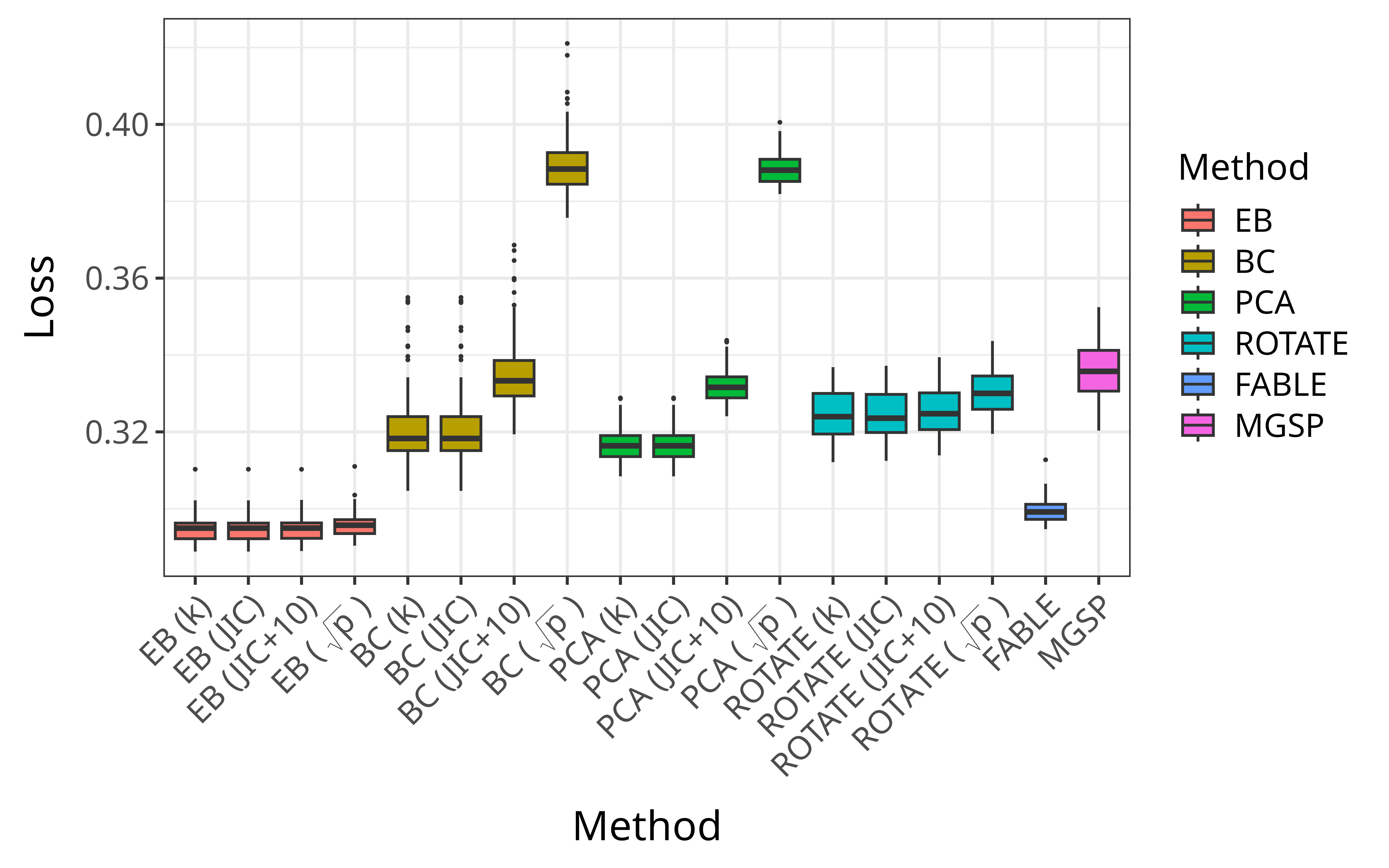}
        \caption{$n=1000, p=5000$}
    \end{subfigure}  
    \caption{Relative Frobenius error in estimating the true covariance matrix under different sample size ($n$) and dimension ($p$) settings in the second scenario.}
    \label{fig:covariance_fr_2}
\end{figure}

\begin{figure}[h]
    \centering
    \begin{subfigure}{0.475\linewidth}
        \centering
        \includegraphics[width=\linewidth]{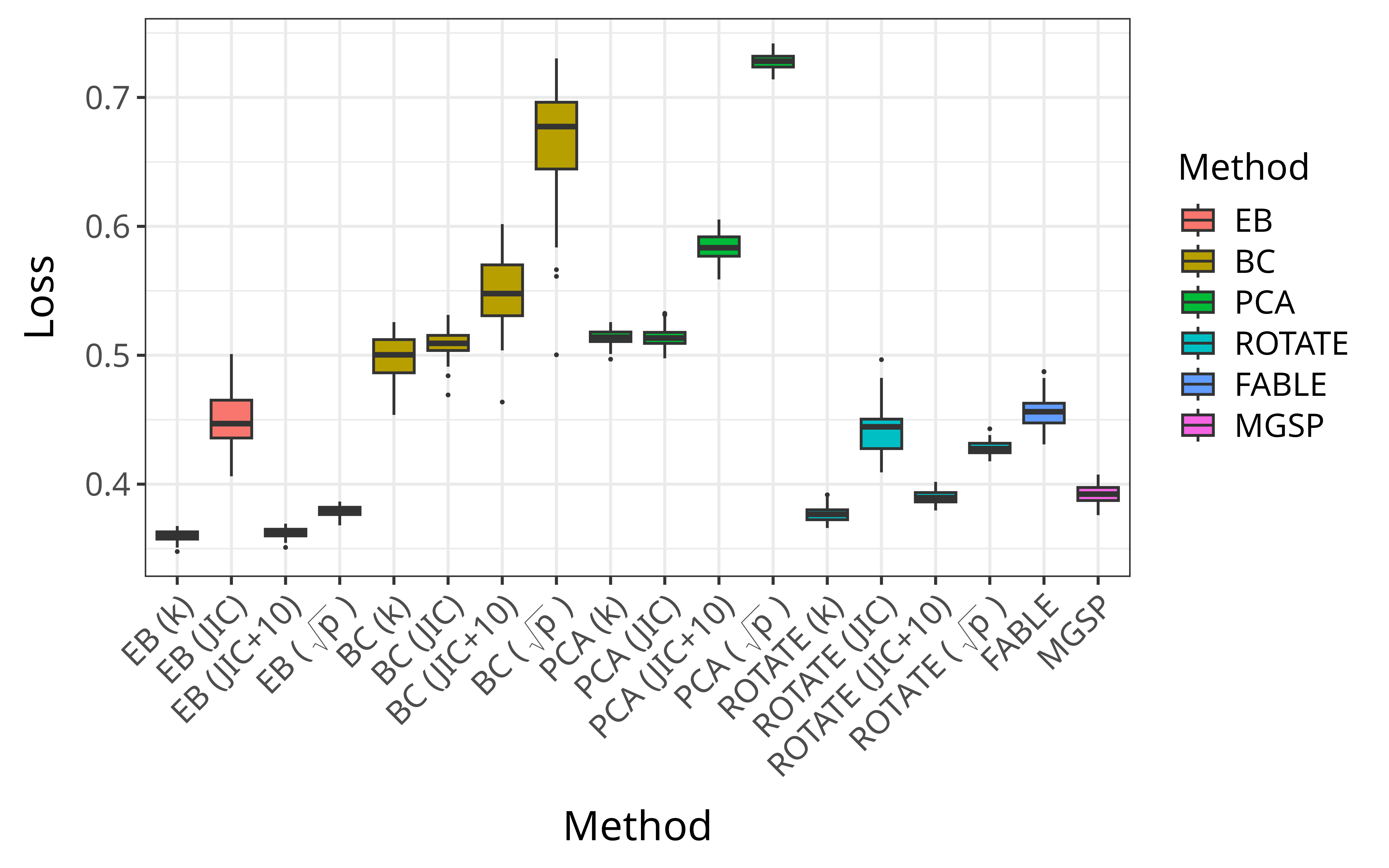}
        \caption{$n=500, p=1000$}
    \end{subfigure}
    \hfill
    \begin{subfigure}{0.475\linewidth}
        \centering
        \includegraphics[width=\linewidth]{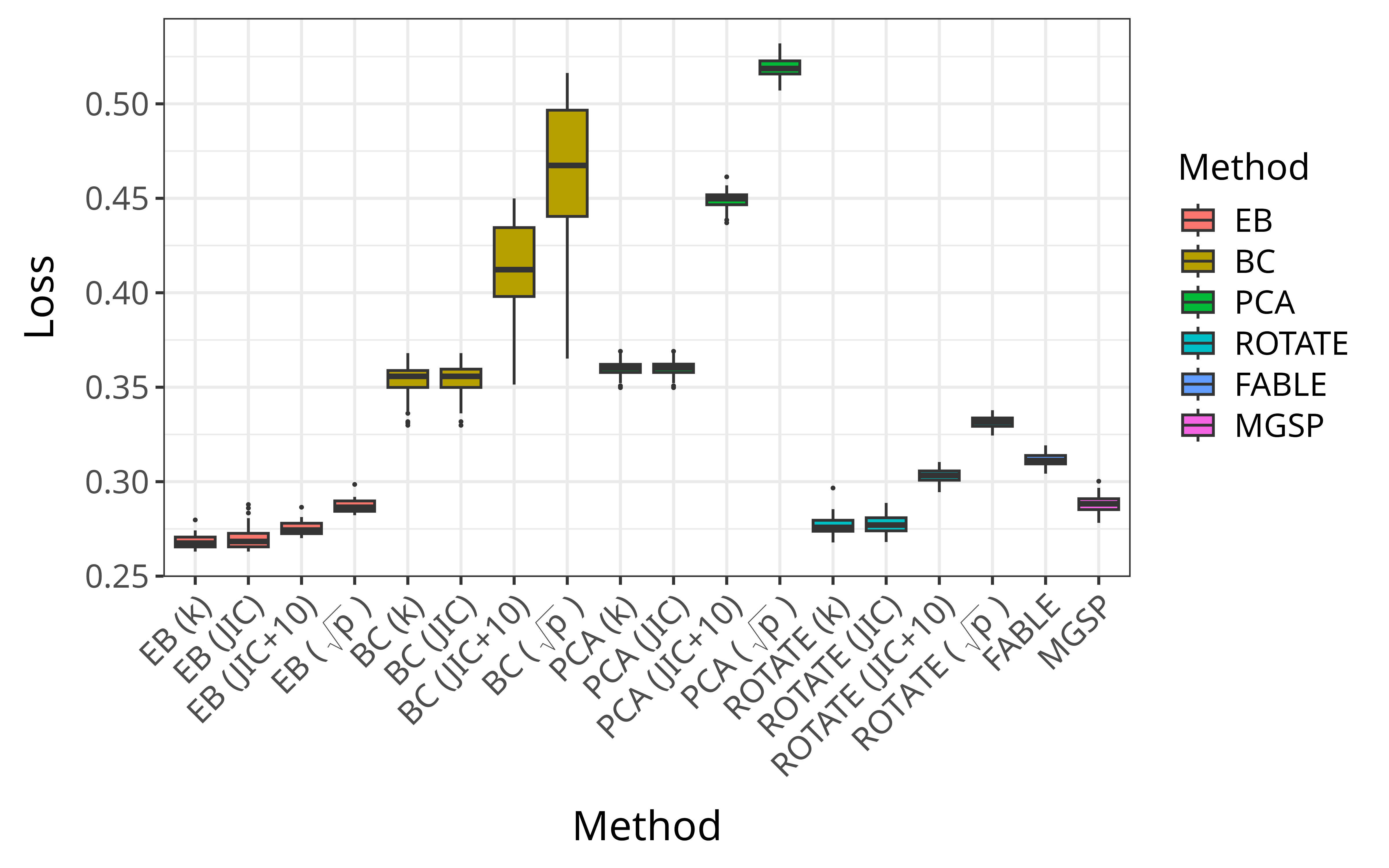}
        \caption{$n=1000, p=1000$}
    \end{subfigure} 
    \vspace{0.5em} 
    \begin{subfigure}{0.475\linewidth}
        \centering
        \includegraphics[width=\linewidth]{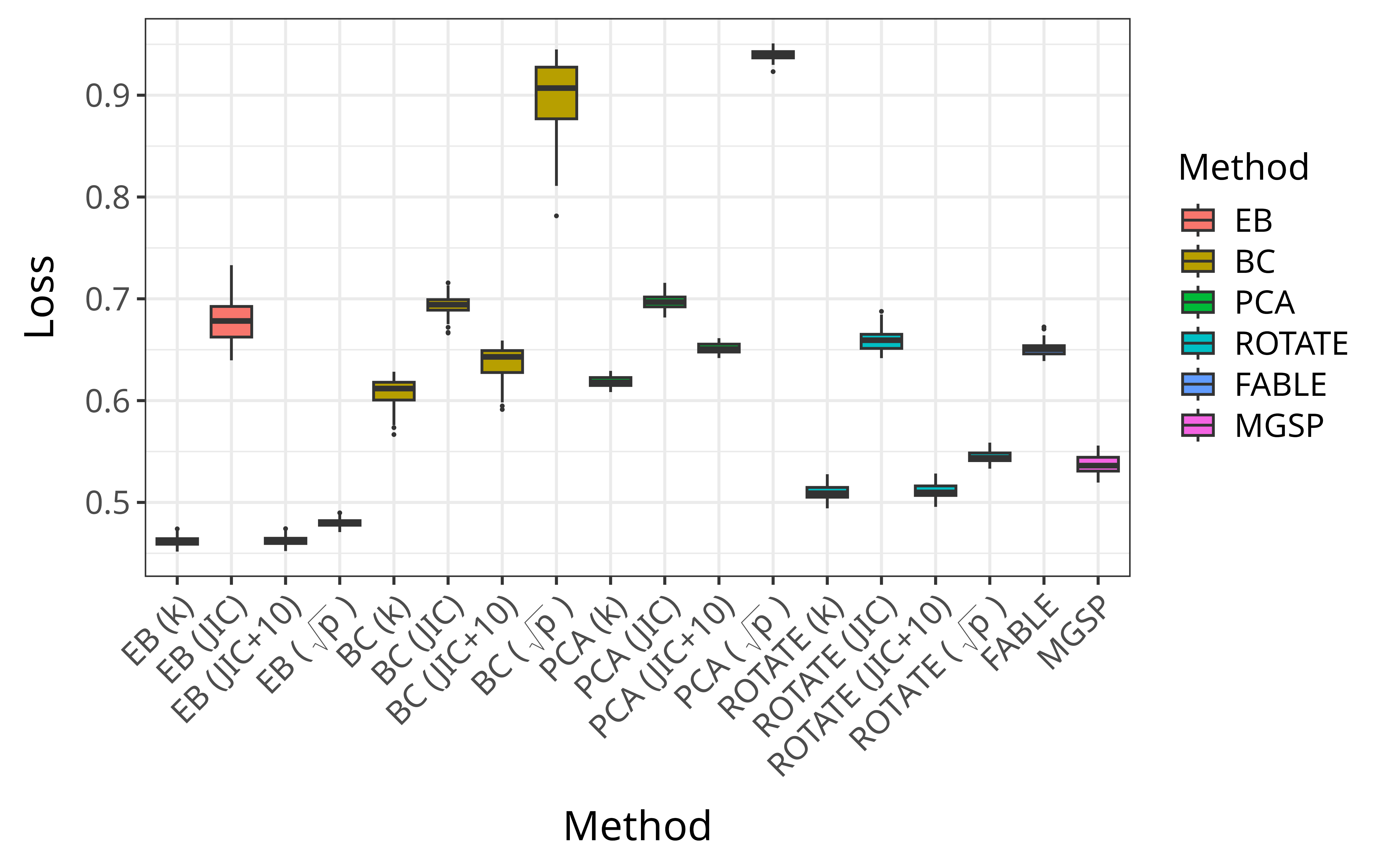}
        \caption{$n=500, p=5000$}
    \end{subfigure}
    \hfill
    \begin{subfigure}{0.475\linewidth}
        \centering
        \includegraphics[width=\linewidth]{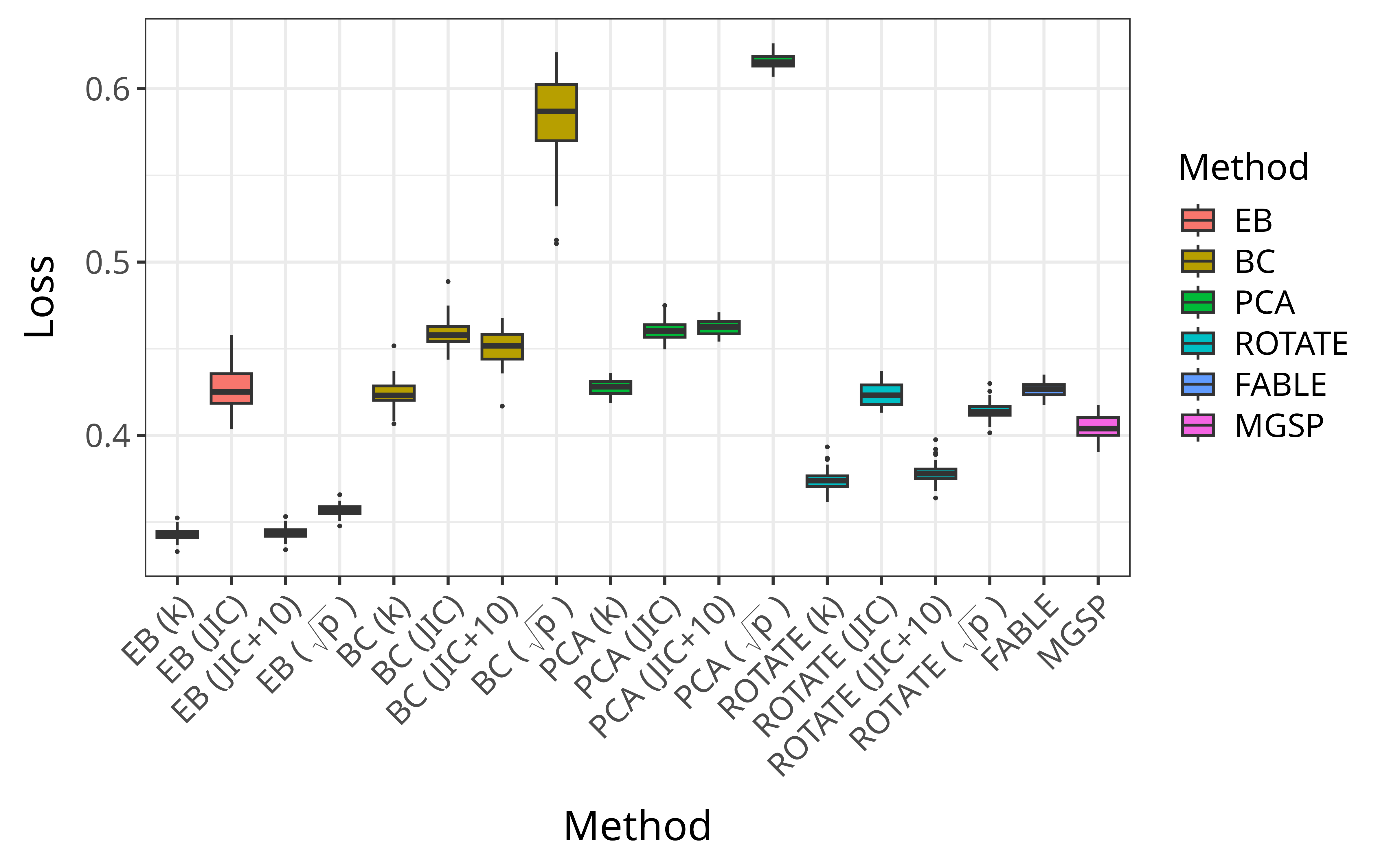}
        \caption{$n=1000, p=5000$}
    \end{subfigure}  
    \caption{Relative Frobenius error in estimating the true covariance matrix under different sample size ($n$) and dimension ($p$) settings in the third scenario.}
    \label{fig:covariance_fr_3}
\end{figure}

Table \ref{tab:uq} reports the average coverage of entries of the covariance matrix for interval estimates. 
For simplicity, we report the results of \texttt{EigenBayes} with $H = \lfloor \sqrt{p} \rfloor$, but similar results were obtained with the other choices. While our procedure provides accurate frequentist coverage on average across entries of the covariance, \texttt{FABLE}'s estimates suffer from non-negligible undercoverage while having comparable width on average, and \texttt{MGSP} slightly undercovers the true parameters when $p=5000$.

\begin{table}[h]
{	\begin{tabular}{cccccccccc}
		& & \multicolumn{4}{c}{\texttt{EigenBayes}} & \multicolumn{2}{c}{\texttt{FABLE}} & \multicolumn{2}{c}{\texttt{MGSP}} \\
		& & \multicolumn{2}{c}{CLT} & \multicolumn{2}{c}{$\Pi$} & \multicolumn{2}{c}{} & \multicolumn{2}{c}{} \\
		Scenario & $(n,p)$ & Cov & Len & Cov & Len & Cov & Len & Cov & Len \\
		\multirow{4}{*}{\textup{(a)}} 
		& $(500,1000)$ & 0.940 & 0.278 & 0.941 & 0.278 & 0.718 & 0.272 & 0.946 & 0.280 \\
		& $(500,5000)$ & 0.940 & 0.204 & 0.941 & 0.204 & 0.678 & 0.201 & 0.945 & 0.201 \\
		& $(1000,1000)$ & 0.944 & 0.272 & 0.944 & 0.272 & 0.700 & 0.258 & 0.876 & 0.304 \\
		& $(1000,5000)$ & 0.945 & 0.198 & 0.946 & 0.198 & 0.679 & 0.205 & 0.810 & 0.199 \\
		\multirow{4}{*}{\textup{(b)}} 
		& $(500,1000)$ & 0.954 & 0.323 & 0.953 & 0.320 & 0.749 & 0.326 & 0.950 & 0.318 \\
		& $(500,5000)$ & 0.948 & 0.237 & 0.947 & 0.234 & 0.685 & 0.231 & 0.944 & 0.230\\
		& $(1000,1000)$ & 0.955 & 0.320 & 0.954 & 0.317 & 0.746 & 0.326 & 0.910 & 0.393 \\
		& $(1000,5000)$ & 0.953 & 0.232 & 0.953 & 0.229 & 0.688 & 0.228 & 0.862 & 0.256 \\
		\multirow{4}{*}{\textup{(c)}} 
		& $(500,1000)$ & 0.965 & 0.218 & 0.964 & 0.218 & 0.786 & 0.184 & 0.963 & 0.216 \\
		& $(500,5000)$ & 0.965 & 0.160 & 0.964 & 0.160 & 0.797 & 0.159 & 0.963 & 0.154 \\
		& $(1000,1000)$ & 0.971 & 0.217 & 0.970 & 0.217 & 0.771 & 0.164 & 0.925 & 0.223 \\
		& $(1000,5000)$ & 0.970 & 0.156 & 0.969 & 0.156 & 0.799 & 0.150 & 0.871 & 0.157 \\
	\end{tabular}}
\caption{
	  Average frequentist coverage (Cov) and average interval length (Len) for interval estimates. For \texttt{EigenBayes}, we report results from interval estimates obtained both using the central limit theorem and the approximation of the posterior distribution. }
    	\label{tab:uq}
    \end{table}

\section{Application}\label{sec:application}
We analyze the ImmGen microarray dataset GSE15907 \citep{painter, desch}, which contains gene-expression measurements from multiple ex-vivo immune lineages, primarily collected from adult B6 male mice and comprising 628 samples. We randomly select 5000 genes out of the half with the highest variance and apply standard pre-processing, which is reported in the supplemental. 

We evaluate competing methodologies via out-of-sample log-likelihood. Specifically, we perform 50 random $80/20$ train-test splits, fit each method on the training samples, and compute the log-likelihood of the held-out samples under the estimated covariance model. 
For all methods that require specifying either the number of factors or an upper bound on it, we fit the method over a grid of values ranging from 40 to 160, assess the robustness of the resulting out-of-sample log-likelihoods to this choice, and compare performance against \texttt{FABLE}.


The left panel of Figure~\ref{fig:application} reports the mean difference in out-of-sample log-likelihood between \texttt{EigenBayes} and each competing method across replications, together with pointwise 95\% bands. Positive values indicate better predictive performance for \texttt{EigenBayes}.  \texttt{EigenBayes} is consistently better than \texttt{BC} and \texttt{PCA} for all values of the latent dimension considered.
As $H$ increases, \texttt{EigenBayes} shows improvements over 
 \texttt{FABLE} and \texttt{MGSP} for a wide range of values for $H$. 
 Finally, the difference with  \texttt{ROTATE} is negligible for smaller values of $H$, but  \texttt{EigenBayes} consistently achieves higher out-of-sample log-likelihood for larger values and is substantially more robust to overfitting $H$.

In real applications, fitting each method over many values of $H$ may be computationally intensive. We therefore also study two default choices, $\widehat{k}_{\mathrm{JIC}}+10$ and $\lfloor \sqrt{p} \rfloor$. For each choice, we report the distribution of $p$-values from one-sided paired $t$-tests assessing whether \texttt{EigenBayes} achieves higher out-of-sample log-likelihood (right panel of Figure \ref{fig:application}). For each competing method except for \texttt{FABLE} and \texttt{MGSP}, we compare \texttt{EigenBayes} against the competitor using the same number of latent components. The resulting $p$-values are small in most comparisons except for the ones with \texttt{MGSP}, providing evidence that \texttt{EigenBayes} yields higher predictive log-likelihood than the competing methods under these default choices of $H$. The comparison with \texttt{MGSP} is more nuanced, since \texttt{MGSP} tends to have a slightly better (worse) out-of-sample performance when we select $K = \widehat{k}_{\mathrm{JIC}}+10$ ($H = \lfloor \sqrt{p} \rfloor$ respectively), but evidence in favour of either approach is not overwhelming for both choices. Nevertheless, \texttt{EigenBayes} is {\em massively faster} with model fitting taking on average less than $1$ second, while posterior computation for \texttt{MGSP} took $\approx 30$ minutes for each replication.

Overall, these results show that \texttt{EigenBayes} achieves excellent accuracy while displaying remarkable robustness to the specification of the number of latent components and appealing computational efficiency. This makes \texttt{EigenBayes} a compelling alternative for practitioners when selecting the factor dimension is difficult, while also reducing the need for extensive tuning.

\begin{figure}[H]
    \centering
    \begin{subfigure}{0.49\linewidth}
        \centering
        \includegraphics[width=\linewidth]{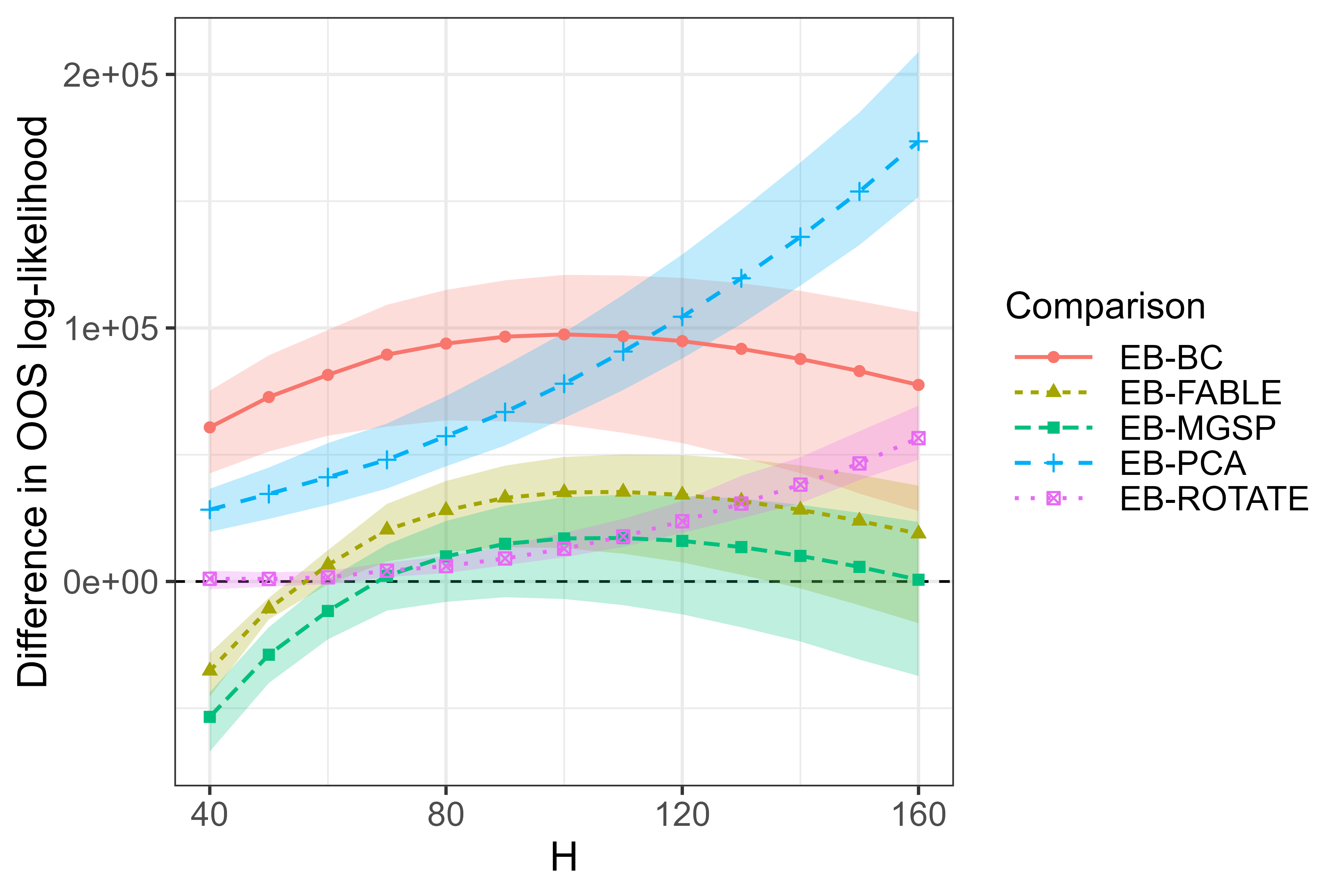}
    \end{subfigure}
    \hfill
    \begin{subfigure}{0.49\linewidth}
        \centering
        \includegraphics[width=\linewidth]{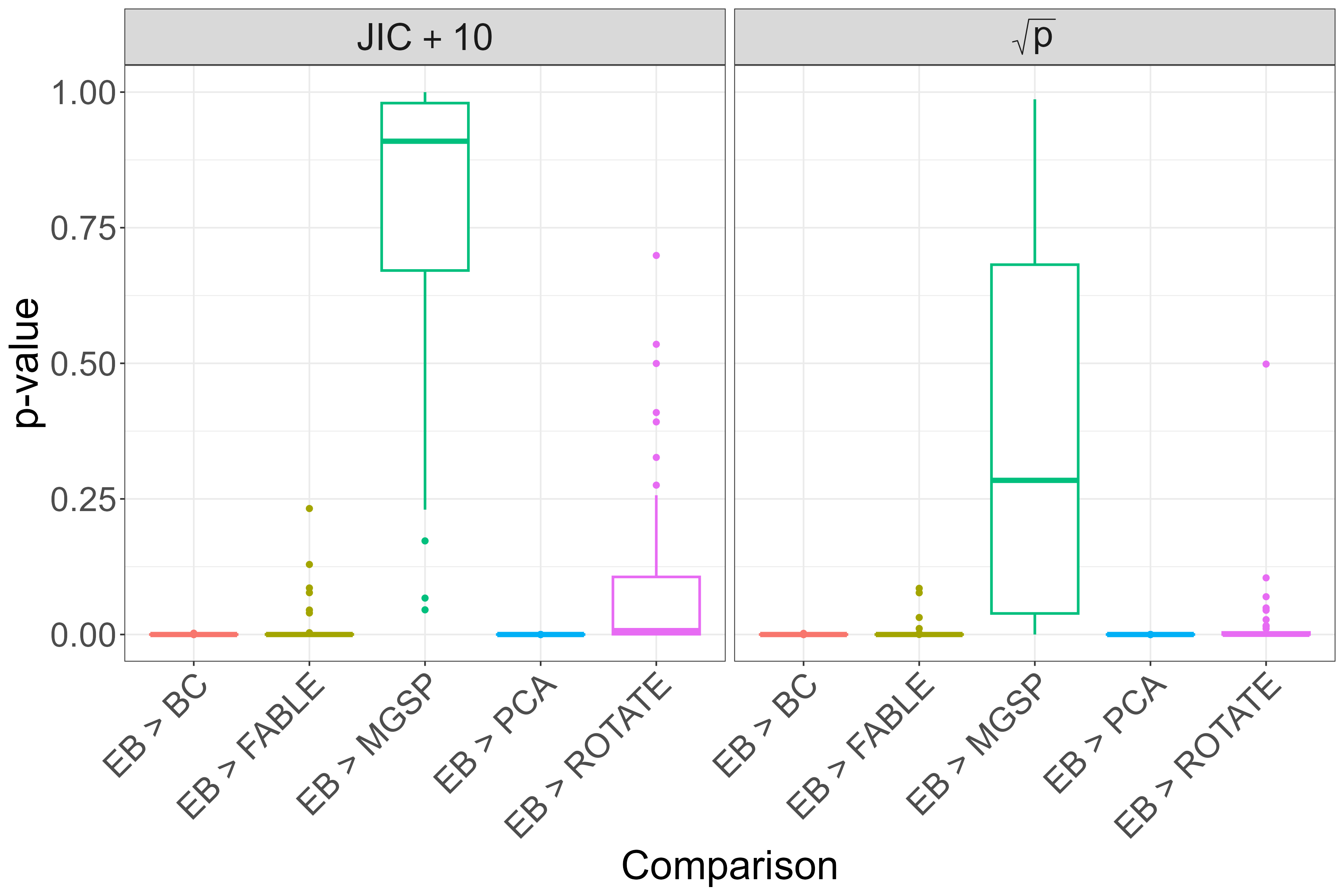}
    \end{subfigure} 
    \caption{Left panel: difference in out-of-sample log-likelihood between \texttt{EigenBayes} and each competing method as $H$ varies, reported as the mean across 50 replications together with pointwise 95\% bands. Right panel: $p$-values across the 50 replications from one-sided paired $t$-tests for the hypothesis that \texttt{EigenBayes} achieves higher out-of-sample log-likelihood, for two default choices of $H$.}
    \label{fig:application}
\end{figure}

\pagebreak
\section{Discussion}\label{sec:conclusion}

We have proposed \texttt{EigenBayes}, a novel methodology for scalable factor analysis, which does not require the rank to be known in advance or to be otherwise estimated. The adaptive shrinkage prior calibrates the regularisation to both latent components and outcomes, allowing superfluous components to be shrunk toward zero and mitigating the sensitivity to the choice of latent dimension. Theoretical results show that \texttt{EigenBayes} recovers the low-rank covariance component at favourable rates and provides asymptotically valid uncertainty quantification, while numerical experiments demonstrate strong finite-sample performance relative to existing alternatives.

Several directions for future research are promising. First, it would be useful to extend the framework beyond Gaussian data. Many scientific applications involve positive, count, binary, or network-valued observations, for which low-rank structure is still central but the Gaussian factor model is inappropriate. Examples include mutational signature analysis and cancer genomics \citep{zito2024compressive}, where nonnegative matrix factorization-type structure is natural \citep{nmf}, as well as generalized linear factor models \citep{gllvm} for binary responses in ecology, social networks, and recommender systems. Developing similar procedures for these settings would require replacing the conjugate Gaussian surrogate regression with scalable approximate Bayesian updates while preserving adaptive component-specific shrinkage. 

Second, another natural direction is to further localize the shrinkage mechanism. The present formulation adapts to heterogeneity across outcomes and latent components through outcome-specific and component-specific shrinkage parameters, but it does not allow the prior variance to vary freely across individual entries of the loading matrix. Introducing entrywise or more locally adaptive shrinkage could improve performance in settings with sparse or highly heterogeneous loading patterns, although doing so would require new strategies beyond the row- and component-level moment matching used here.

Third, extending the theory to allow dependence and misspecification is an important next step. This includes serial or spatial correlation across observations, heavier-tailed noise distributions, heteroscedasticity beyond the settings considered here, and approximate factor models in which the covariance is only approximately low-rank plus diagonal. Such extensions would broaden applicability of the method to longitudinal, spatial, and observational datasets.

Finally, another interesting direction is to study combinations of spectral or algorithmic estimators with Bayesian inference approaches in more generality, to speed up computation while allowing for adaptive regularisation and uncertainty quantification.

\section*{Acknowledgement}
This research was partially supported by the Office of Naval Research (Grant N00014-24-1-2626), by the National Institutes of Health (Grant R01GM163225 and Grant R01ES035625), and by the European Research Council under the European Union’s Horizon 2020 research and innovation programme (grant agreement No 856506).

\bibliographystyle{agsm}
\bibliography{References}

@article{bmsfa,
author = {De Vito, Roberta and Bellio, Ruggero and Trippa, Lorenzo and Parmigiani, Giovanni},
title = {{Bayesian multistudy factor analysis for high-throughput biological data}},
volume = {15},
journal = {The Annals of Applied Statistics},
number = {4},
publisher = {Institute of Mathematical Statistics},
pages = {1723 -- 1741},
year = {2021}
}

@article{xie2024eigenvector,
  title={An eigenvector-assisted estimation framework for signal-plus-noise matrix models},
  author={Xie, Fangzheng and Wu, Dingbo},
  journal={Biometrika},
  volume={111},
  number={2},
  pages={661--676},
  year={2024},
  publisher={Oxford University Press}
}

@article{wu_xie_graph,
  author  = {Dingbo Wu and Fangzheng Xie},
  title   = {Statistical Inference of Random Graphs With a Surrogate Likelihood Function},
  journal = {Journal of Machine Learning Research},
  year    = {2025},
  volume  = {26},
  number  = {230},
  pages   = {1--65}}

@article{wu2025sanvi,
  title={SANVI: A Fast Spectral-Assisted Network Variational Inference Method with an Extended Surrogate Likelihood Function},
  author={Wu, Dingbo and Xie, Fangzheng},
  journal={arXiv preprint arXiv:2509.00562},
  year={2025}
}

@article{chen_identfiability,
author = {Chen, Yunxiao and Li, Xiaoou and Zhang, Siliang},
title = {{Structured Latent Factor Analysis for Large-scale Data: Identifiability, Estimability, and Their Implications}},
journal = {Journal of the American Statistical Association},
volume = {115},
number = {532},
pages = {1756--1770},
year = {2020},
publisher = {Taylor \& Francis}
}

@Article{chen_jmle,
  author={Yunxiao Chen and Xiaoou Li and Siliang Zhang},
  title={{Joint Maximum Likelihood Estimation for High-Dimensional Exploratory Item Factor Analysis}},
  journal={Psychometrika},
  year=2019,
  volume={84},
  number={1},
  pages={124-146},
  month={March}
}

@article{gllvm,
author = {Moustaki, Ir and Knott, Martin},
year = {2000},
month = {02},
pages = {391-411},
title = {Generalized latent trait models},
volume = {65},
journal = {Psychometrika}
}

@article{painter,
    author = {Painter, Michio W. and Davis, Scott and Hardy, Richard R. and Mathis, Diane and Benoist, Christophe and The Immunological Genome Project Consortium },
    title = {Transcriptomes of the B and T Lineages Compared by Multiplatform Microarray Profiling},
    journal = {The Journal of Immunology},
    volume = {186},
    number = {5},
    pages = {3047-3057},
    year = {2011},
    month = {03},
    issn = {0022-1767}
}

@article{desch,
author = {Desch, A. and Murphy, Kenneth and Kedl, Ross and Lahoud, Mireille and Caminschi, Irina and Shortman, Ken and Henson, Peter and Jakubzick, Claudia},
year = {2011},
month = {08},
pages = {1789-97},
title = {CD103+ pulmonary dendritic cells preferentially acquire and present apoptotic cell-associated antigen},
volume = {208},
journal = {The Journal of experimental medicine}
}

@inbook{vershynin_12, place={Cambridge}, title={{Introduction to the non-asymptotic analysis of random matrices}}, booktitle={{Compressed Sensing: Theory and Applications}}, publisher={Cambridge University Press}, author={Vershynin, Roman}, editor={Eldar, Yonina C. and Kutyniok, GittaEditors}, year={2012}, pages={210–268}}

@article{bandeira_van_handel_16,
author = {Afonso S. Bandeira and Ramon van Handel},
title = {{Sharp nonasymptotic bounds on the norm of random matrices with independent entries}},
volume = {44},
journal = {The Annals of Probability},
number = {4},
publisher = {Institute of Mathematical Statistics},
pages = {2479 -- 2506},
year = {2016},

}

@article{lee24,
author = {Sze Ming Lee and Yunxiao Chen and Tony Sit},
title = {A Latent Variable Approach to Learning High-dimensional Multivariate longitudinal Data},
journal = {Journal of the American Statistical Association},
volume = {0},
number = {ja},
pages = {1--21},
year = {2025},
publisher = {Taylor \& Francis},
doi = {10.1080/01621459.2025.2606384}

}

@article{bai_03,
 author = {Jushan Bai},
 journal = {Econometrica},
 number = {1},
 pages = {135--171},
 title = {{Inferential Theory for Factor Models of Large Dimensions}},
 volume = {71},
 year = {2003}
}

@article{blast,
      title={Spectral decomposition-assisted multi-study factor analysis}, 
      author={Lorenzo Mauri and Niccolò Anceschi and David B. Dunson},
      year={2025},
      eprint={2502.14600},
      archivePrefix={arXiv},
      primaryClass={stat.ME},
      journal={arXiv preprint arXiv:2502.14600}, 
}

@article{rotate,
author = {Veronika Rockova and Edward I. George},
title = {{Fast Bayesian Factor Analysis via Automatic Rotations to Sparsity}},
journal = {Journal of the American Statistical Association},
volume = {111},
number = {516},
pages = {1608--1622},
year = {2016},
publisher = {ASA Website}
}

@article{flair,
    author = {Mauri, L and Dunson, D B},
    title = {Factor pre-training in {B}ayesian multivariate logistic models},
    journal = {Biometrika},
    pages = {asaf056},
    year = {2025},
    issn = {1464-3510},
    doi = {10.1093/biomet/asaf056}
}

@article{fable,
  title={Blessing of dimension in {B}ayesian inference on covariance matrices},
  author={Chattopadhyay, Shounak and Zhang, Anru R and Dunson, David B},
  journal={arXiv preprint arXiv:2404.03805},
  year={2024}
}

@article{chen_jic,
    author = {Chen, Y and Li, X},
    title = "{Determining the number of factors in high-dimensional generalized latent factor models}",
    journal = {Biometrika},
    volume = {109},
    number = {3},
    pages = {769-782},
    year = {2021},
}

@article{fama,
  title={Inference on covariance structure in high-dimensional multi-view data},
  author={Mauri, Lorenzo and Dunson, David B},
  journal={arXiv preprint arXiv:2509.02772},
  year={2025}
}

@article{Barigozzi2018ConsistentEO,
  title={Consistent estimation of high-dimensional factor models when the factor number is over-estimated},
  author={Matteo Barigozzi and Haeran Cho},
  journal={Electronic Journal of Statistics},
  year={2018}
}

@article{cusp,
    author = {Legramanti, Sirio and Durante, Daniele and Dunson, David B},
    title = {Bayesian cumulative shrinkage for infinite factorizations},
    journal = {Biometrika},
    volume = {107},
    number = {3},
    pages = {745-752},
    year = {2020},
    month = {05},
    issn = {0006-3444},
    doi = {10.1093/biomet/asaa008}}

@article{sparse_infinite,
    author = {Bhattacharya, A. and Dunson, D. B.},
    title = {Sparse {B}ayesian infinite factor models},
    journal = {Biometrika},
    volume = {98},
    number = {2},
    pages = {291-306},
    year = {2011},
    month = {06},
    issn = {0006-3444},
    doi = {10.1093/biomet/asr013}}

@article{Bai2020SimplerProofs,
  title={Simpler proofs for approximate factor models of large dimensions},
  author={Bai, Jushan and Ng, Serena},
  journal={arXiv preprint arXiv:2008.00254},
  year={2020}
}

@article{Doz2012QuasiMaximum,
  title={A quasi--maximum likelihood approach for large, approximate dynamic factor models},
  author={Doz, C{\'e}cile and Giannone, Domenico and Reichlin, Lucrezia},
  journal={Review of Economics and Statistics},
  volume={94},
  number={4},
  pages={1014--1024},
  year={2012},
  publisher={MIT Press},
}

@article{Fan2023BridgingFactor,
  title={Bridging factor and sparse models},
  author={Fan, Jianqing and Masini, Raffaella P and Medeiros, Marcelo C},
  journal={The Annals of Statistics},
  volume={51},
  number={4},
  pages={1692--1717},
  year={2023}
}

@article{basil,
  title={Pathway-based {B}ayesian factor models for gene expression data},
  author={Mauri, Lorenzo and Stolf, Federica and Herring, Amy H and Miller, Cameron and Dunson, David B},
  journal={arXiv preprint arXiv:2601.13419},
  year={2026}
}

@article{ibp,
  author  = {Thomas L. Griffiths and Zoubin Ghahramani},
  title   = {The Indian Buffet Process: An Introduction and Review},
  journal = {Journal of Machine Learning Research},
  year    = {2011},
  volume  = {12},
  number  = {32},
  pages   = {1185--1224}
}

@article{vi_review,
author = {David Blei and Kucukelbir, Alp and Jon D McAuliffe},
title = {{Variational Inference: A Review for Statisticians}},
journal = {Journal of the American Statistical Association},
volume = {112},
number = {518},
pages = {859--877},
year = {2017},
publisher = {Taylor \& Francis}
}

@InProceedings{bb_vi,
  title = 	 {Black Box Variational Inference},
  author = 	 {Ranganath, Rajesh and Gerrish, Sean and Blei, David},
  booktitle = 	 {Proceedings of the Seventeenth International Conference on Artificial Intelligence and Statistics},
  pages = 	 {814--822},
  year = 	 {2014},
  editor = 	 {Kaski, Samuel and Corander, Jukka},
  volume = 	 {33},
  series = 	 {Proceedings of Machine Learning Research},
  address = 	 {Reykjavik, Iceland},
  month = 	 {22--25 Apr},
  publisher =    {PMLR}
}

@article{stoch_vi,
  author  = {Matthew D. Hoffman and David Blei and Chong Wang and John Paisley},
  title   = {Stochastic Variational Inference},
  journal = {Journal of Machine Learning Research},
  year    = {2013},
  volume  = {14},
  number  = {40},
  pages   = {1303--1347},
}

@article{ad_vi,
  author  = {Alp Kucukelbir and Dustin Tran and Rajesh Ranganath and Andrew Gelman and David M. Blei},
  title   = {Automatic Differentiation Variational Inference},
  journal = {Journal of Machine Learning Research},
  year    = {2017},
  volume  = {18},
  number  = {14},
  pages   = {1--45}
}

@article{single_cell_fm,
author = {Buettner, Florian and Pratanwanich, Naruemon and McCarthy, Davis and Marioni, John and Stegle, Oliver},
journal = {Genome Biology},
year = {2017},
month = {11},
pages = {1–13},
volume = {18},
number = {212},
title = {{f-scLVM: scalable and versatile factor analysis for single-cell RNA-seq.}}
}

@article{fa_rjmcmc,
 ISSN = {10170405, 19968507},
 author = {Hedibert Freitas Lopes and Mike West},
 journal = {Statistica Sinica},
 number = {1},
 pages = {41--67},
 publisher = {Institute of Statistical Science, Academia Sinica},
 title = {BAYESIAN MODEL ASSESSMENT IN FACTOR ANALYSIS},
 urldate = {2026-01-26},
 volume = {14},
 year = {2004}
}

@article{ic_bn,
author = {Bai, Jushan and Ng, Serena},
title = {Determining the Number of Factors in Approximate Factor Models},
journal = {Econometrica},
volume = {70},
number = {1},
pages = {191-221},
doi = {https://doi.org/10.1111/1468-0262.00273},
year = {2002}
}

@article{ic_ah,
author = {Ahn, Seung C. and Horenstein, Alex R.},
title = {Eigenvalue Ratio Test for the Number of Factors},
journal = {Econometrica},
volume = {81},
number = {3},
pages = {1203-1227},
doi = {https://doi.org/10.3982/ECTA8968},
year = {2013}
}

@article{lee_lee_24, title={Phase transition for the generalized two-community stochastic block model}, volume={61}, DOI={10.1017/jpr.2023.44}, number={2}, journal={Journal of Applied Probability}, author={Lee, Sunmin and Lee, Ji Oon}, year={2024}, pages={385–400}}

@article{MP_law,
doi = {10.1070/SM1967v001n04ABEH001994},
year = {1967},
month = {apr},
publisher = {},
volume = {1},
number = {4},
pages = {457},
author = {V A Marchenko and L A Pastur},
title = {DISTRIBUTION OF EIGENVALUES FOR SOME SETS OF RANDOM MATRICES},
journal = {Mathematics of the USSR-Sbornik}
}

@article{bloemenda_et_al,
author = {Bloemendal, Alex and L{\'a}szl{\'o}, Erdős and Antti, Knowles and Horng-Tzer, Yau and Jun, Yin},
title = {{Isotropic local laws for sample covariance and generalized Wigner matrices}},
volume = {19},
journal = {Electronic Journal of Probability},
number = {none},
publisher = {Institute of Mathematical Statistics and Bernoulli Society},
pages = {1 -- 53},
year = {2014},
doi = {10.1214/EJP.v19-3054}
}

@article{Bloemendal_16,
  author  = {Bloemendal, Alex and Knowles, Antti and Yau, Horng-Tzer and Yin, Jun},
  title   = {On the principal components of sample covariance matrices},
  journal = {Probability Theory and Related Fields},
  volume  = {164},
  number  = {1--2},
  pages   = {459--552},
  year    = {2016},
  doi     = {10.1007/s00440-015-0616-x}
}

@article{fan1,
title = {Bayesian factor-adjusted sparse regression},
journal = {Journal of Econometrics},
volume = {230},
number = {1},
pages = {3-19},
year = {2022},
author = {Jianqing Fan and Bai Jiang and Qiang Sun}
}

@article{zito2024compressive,
  title={{Compressive Bayesian non-negative matrix factorization for mutational signatures analysis}},
  author={Zito, Alessandro and Miller, Jeffrey W},
  journal={arXiv preprint arXiv:2404.10974},
  year={2024}
}

@article{nmf,
  title={Learning the parts of objects by non-negative matrix factorization},
  author={Lee, Daniel D. and Seung, H. Sebastian},
  journal={Nature},
  volume={401},
  number={6755},
  pages={788--791},
  year={1999},
  doi={10.1038/44565}
}

@article{durante_note,
title = {A note on the multiplicative gamma process},
journal = {Statistics \& Probability Letters},
volume = {122},
pages = {198-204},
year = {2017},
issn = {0167-7152},
doi = {https://doi.org/10.1016/j.spl.2016.11.014},
author = {Daniele Durante}
}

@article{kowal2023semiparametric,
  title={Semiparametric functional factor models with {B}ayesian rank selection},
  author={Kowal, Daniel R and Canale, Antonio},
  journal={Bayesian Analysis},
  volume={18},
  number={4},
  pages={1161--1189},
  year={2023},
  publisher={International Society for Bayesian Analysis}
}

@article{zorzetto2025multivariate,
  title={Multivariate causal effects: a {B}ayesian causal regression factor model},
  author={Zorzetto, Dafne and Landy, Jenna and Zigler, Corwin and Parmigiani, Giovanni and De Vito, Roberta},
  journal={arXiv preprint arXiv:2504.03480},
  year={2025}
}

@article{bayesNMF,
author = {Jenna M. Landy and Nishanth Basava and Giovanni Parmigiani},
title = {bayes{NMF}: Fast {B}ayesian {P}oisson {NMF} with Automatically Learned Rank Applied to Mutational Signatures},
journal = {Journal of Computational and Graphical Statistics},
volume = {0},
number = {ja},
pages = {1--17},
year = {2026},
publisher = {Taylor \& Francis},
doi = {10.1080/10618600.2026.2657487}
}

@article{zhang_zhou_20,
author = {Zhang, Anru R. and Zhou, Yuchen},
title = {On the non-asymptotic and sharp lower tail bounds of random variables},
journal = {Stat},
volume = {9},
number = {1},
pages = {e314},
doi = {https://doi.org/10.1002/sta4.314},
year = {2020}
}

@book{stewart1990matrix,
  title     = {Matrix Perturbation Theory},
  author    = {Stewart, Gilbert W. and Sun, Ji-guang},
  year      = {1990},
  publisher = {Academic Press},
  address   = {Boston}
}

\newpage
\appendix

\section*{Supplementary material of \say{Overfitted high-dimensional matrix factorizations via adaptive spectral shrinkage}}
\section{Proof of the theoretical results}

\subsection{Proofs of the main results}

\begin{proof}[Proof of Proposition \ref{prop:psi_l}]
By Lemma \ref{lemma:sv_Y}, we have $\sum_{l=1}^H(d_l^2 - d_{H+1}^2) \asymp k n p \asymp n^2$, with probability at least $1-o(1)$. By Lemma \ref{lemma:sv_Y_diff}, for $k<l\leq H$, we have $d_l^2 - d_{H+1}^2 \asymp n^{1/3} (l-k)^{2/3} \lesssim n^{1/3} n^{\frac{2}{5} \frac{2}{3}} = n^{3/5}$, with probability at least $1-o(1)$.
Similarly, $\max_{j=1, \dots, p} v_{j, 1:H}^\top (D_{1:H}^2 - d_{H+1}^2 I_H) v_{j, 1:H} \leq \max_{j=1, \dots, p} ||y^{(j)}|| \lesssim n $, with probability at least $1-o(1)$, by Lemma \ref{lemma:y_j}.   Moreover, we have $\min_{j=1, \dots, p} \hat \sigma_j^2 \geq c_{\sigma}$. 
Putting everything together, we get 
$$
\max_{j=1, \dots, p} \hat \psi_l^2 \hat \tau_j^2 \lesssim \frac{n^{3/5}}{n^2}  \frac{n}{n} \lesssim \frac{1}{n^{1 + 2/5}}.
$$  
\end{proof}

\begin{proof}[Proof of Corollary \ref{corr:norm_Lambda_l}]
    See Lemma \ref{lemma:norm_Lambda_hat}.
\end{proof}

\begin{proof}[Proof of Theorem \ref{thm:consistency}]

We first show consistency of point estimates. Let $\hat \Lambda_{1:k} \in \mathbb R^{p \times k}$ and $\hat \Lambda_{k+1:H} \in \mathbb R^{p \times H-k}$ be the matrices obtained by extracting the first $k$ columns and the remaining $H-k$ ones of $\hat \Lambda$ respectively. Moreover, let %
$\hat L_k = \hat \Lambda_{1:k} \hat \Lambda_{1:k}^\top$, and $\hat L_{svd} = \hat \Lambda_{svd} \hat \Lambda_{svd}^\top$, with $\hat \Lambda_{svd} = V_{1:k} D_{1:k} / \sqrt{n}$, be the estimated low-rank component using only the first $k$ columns of $\hat \Lambda$ and the spectral estimate of the low-rank component selecting $k$ as the rank, respectively.  

By Lemma \ref{lemma:norm_Lambda_hat}  
$||\hat \Lambda_{k+1:H}|| \lesssim \frac{1}{n^{1/5}}$with probability at least $1-o(1)$. This implies that, with probability at least $1-o(1)$, 
\begin{equation*}
    ||\hat L - \hat L_k || = ||\hat \Lambda_{k+1:H} \hat \Lambda_{k+1:H}^\top ||  \lesssim \frac{1}{n^{2/5}}.
\end{equation*}

 Moreover, by Lemma \ref{lemma:hyperparams}, with probability at least $1 - o(1)$, we have 
 $$
 \max_{j=1, \dots, p}\bigg| \frac{\sqrt{n}}{n + \hat \psi_l^{-2} \hat \tau_j^{-2}} - \frac{1}{\sqrt{n}} \bigg|  = \bigg|\min_{j=1, \dots, p} \frac{\sqrt{n}}{n + \hat \psi_l^{-2} \hat \tau_j^{-2}} - \frac{1}{\sqrt{n}} \bigg| \lesssim \frac{1}{n^{3/2}},
 $$
 which implies
$$
||\hat \Lambda_{1:k} - \hat \Lambda_{svd}|| \lesssim \frac{1}{n^{3/2}} \sqrt{np} \asymp \frac{1}{\sqrt{n}},
$$
with probability at least $1-o(1)$, since $d_{l} \lesssim \sqrt{n p}$, with probability at least $1-o(1)$, by Lemma \ref{lemma:sv_Y}. 
Hence, with probability at least $1-o(1)$,
 \begin{equation}
|| \hat L_{svd} -\hat L_{k}|| \lesssim 1 ,  
 \end{equation} 
since $||\hat \Lambda_{1:k}|| \lesssim ||\hat \Lambda|| \lesssim \sqrt{n}$, with probability at least $1-o(1)$, by Lemma \ref{lemma:norm_Lambda_hat}. 
 Moreover, with probability at least $1- o(1)$, we have
\begin{equation*}
    \begin{aligned}
         ||L_0 - \hat L_{svd}|| \lesssim& ||L_{0} - Y^\top U_{1:k} U_{1:k}^\top  Y / n || \leq ||\Lambda_0 \Lambda_0^\top -  Y^\top U_{0,1:k} U_{0, 1:k}^\top  Y / n || +  ||Y^\top \big(U_{1:k} U_{1:k}^\top -  U_{0,1:k} U_{0, 1:k}^\top\big) Y / n  ||\\
      \lesssim& ||\Lambda_0 \big(M^\top M / n - I_k\big) \Lambda_0^\top|| + ||Y^\top \big(U_{1:k} U_{1:k}^\top -   U_{0,1:k} U_{0,1:k}^\top\big)Y ||/ n \\
      &+ || E^\top U_{0,1:k} U_{0, 1:k}^\top  E  ||/ n + ||\Lambda_0 M^\top E|| / n  \\
       \lesssim& p \sqrt{\frac{\log n}{n}} + \big(\frac{1}{\sqrt{n}} + \frac{1}{\sqrt{p}}\big) \frac{np}{n} + \frac{p}{n} + \frac{p}{\sqrt{n}} + \sqrt{p} \asymp   \sqrt{n\log n },
    \end{aligned}
    \end{equation*}
since, with probability at least $1- o(1)$, we have $||M^\top M / n - I_k|| \lesssim \sqrt{\log n / n}$ and $||M|| \lesssim \sqrt{n}$ by Lemma \ref{lemma:sv_F}, $||E|| \lesssim \sqrt{n} + \sqrt{p}$ by Lemma \ref{lemma:E}, and $||U_{1:k} U_{1:k}^\top -   U_{0,1:k} U_{0,1:k}^\top|| \lesssim \frac{1}{\sqrt{n}} + \frac{1}{\sqrt{p}}$ by Lemma \ref{lemma:U}.

Putting everything together, with probability at least $1- o(1)$,  we have $||L_0 - \hat L|| \lesssim \sqrt{n\log n }$, which implies
\begin{equation*}
     \frac{||L_0 - \hat L||}{||L_0||} \lesssim \frac{\sqrt{\log n }}{\sqrt{n}},
\end{equation*}
since $||L_0|| \asymp p \asymp n$.

We proceed by showing consistency for the residual error variance. Let $\delta_j^2 = \frac{v s^2 + ||y^{(j)} - \hat M \hat \lambda_j||^2}{v + n-2}$. Using Lemma \ref{lemma:delta_j}, we can write  $$\delta_j^2  - \sigma_0^2 = \frac{\sigma_0^2}{n} \bigg\{\frac{s_j^2}{\sigma_0^2} - (n-k) \bigg\} - \frac{k \sigma_0^2}{n} + F_j,$$ where $\frac{s_j^2}{\sigma_0^2} \sim \chi_{n-k}^2$ and $\max_{j=1, \dots, p} |F_j| \lesssim \frac{1}{n^{2/5}}$ with probability at least $1-o(1)$. Using Lemma E.7 of \citet{fable}, we have 
$$
\max_{j=1, \dots, p} |\frac{s_j^2}{\sigma_0^2} - (n-k) | \lesssim n \big(\frac{\log n}{n} \big)^{1/3},
$$
with probability at least $1-o(1)$.
Then, recalling $\hat \Sigma = \text{diag}(\delta_1^2, \dots, \delta_p^2)$, we have 
\begin{equation*}
    ||\hat \Sigma - \sigma_0^2 I_p|| \lesssim \big(\frac{\log n}{n} \big)^{1/3},
\end{equation*}
with probability at least $1-o(1)$.  Combining all of the above with $||\Theta_0|| = ||L_0|| \asymp p \asymp n$, we have 
\begin{equation}
    \frac{||\hat \Theta  - \Theta_0||}{||\Theta_0||} \lesssim \sqrt{\frac{\log n}{n}},
\end{equation}
with probability at least $1-o(1)$.

Next, we show posterior contraction around the true parameters $L_0$ and $\Theta_0$ via similar steps to \citet{fable}. Let $\tilde \Lambda$ and $\tilde \Sigma = \text{diag}(\tilde \sigma_1^2, \cdots, \tilde \sigma_p^2)$ be a posterior sample for the loading matrix and the residual variances respectively. We can represent $\tilde \Lambda$ as 
\begin{equation*}
    \tilde \Lambda = \hat \Lambda + \tilde Z, \quad \tilde Z = \begin{bmatrix}
        \tilde z_1, \cdots, \tilde z_p
    \end{bmatrix},
    \quad 
    \text{with } \tilde z_j \sim N(0, \rho^2 \tilde \sigma_j^2 \Psi_{n,j}^{-1} ), \quad (j=1, \dots, p).
\end{equation*}
Thus, $\tilde L = \tilde \Lambda \tilde \Lambda^\top = \hat L + \hat \Lambda \tilde Z^\top + \tilde Z \hat \Lambda^\top + \tilde Z \tilde Z^\top$, which implies $||\tilde L - L_0|| \leq ||\hat L - L_0|| + 2||\hat \Lambda|| ||\tilde Z|| + ||\tilde Z||^2$. Note $\Pi(||\tilde \Sigma|| < C) \geq 1- o(1)$, with probability at least $1-o(1)$, by Lemma \ref{lemma:pi_sigma}, where $\Pi(\cdot)$ denotes the probability measure induced by (\ref{eq:posterior_update}, \ref{eq:posterior_update_cc}). 
Then, with probability at least $1-o(1)$,
$$
||\tilde Z||^2 \lesssim  \rho^2 p ||\Psi_{n,j}^{-1}|| \max_{j=1, \dots, p} \sigma_j^2 \lesssim \frac{p}{n} \asymp 1,
$$
since $||\Psi_{n, l}^{-1}|| \asymp \frac{1}{n}$, and, by Lemma \ref{lemma:norm_Lambda_hat},
$$
||\hat \Lambda|| \lesssim \sqrt{p} \asymp \sqrt{n}.
$$
Next, let $D_j = \tilde \sigma_j^2 - \sigma_0^2$ and note $||\tilde \Sigma - \sigma_0^2 I_p|| = \max_{j=1, \dots, p} |D_j|$. We can express $D_j$ as 
$$D_j = \bigg(1 + \frac{2\Delta_j}{v_n}\bigg)^{-1}\big\{(\delta_j^2 - \sigma_0^2) - \frac{2\Delta_j}{v_n} \sigma_0^2 \big\}, \quad \Delta_j = \frac{v_n \delta_j^2}{2} \tilde \sigma_j^{-2} - \frac{v_n}{2}.$$
Lemma E.7 of \citet{fable} gives $\max_{j=1, \dots, p} |\Delta_j|/ v_n \lesssim \{\log p / n\}^{1/3}$, which also implies $\min_{j=1, \dots, p} |1 + 2\Delta_j/ v_n  | \gtrsim 1/2$, with probability at least $1-o(1)$. 
Hence, with probability at least $1-o(1)$,
$$
||\tilde \Sigma - \sigma_0^2 I_p||\lesssim \bigg(\frac{\log p}{n}\bigg)^{1/3}
$$
since $\max_{j=1, \dots, p}|\delta_j^2 - \sigma_0^2| \lesssim (\log p / n)^{1/3} \asymp (\log n / n)^{1/3}$,  with probability at least $1-o(1)$.
Finally, letting $\tilde \Theta = \tilde L + \tilde \Sigma$, we have
$$
\frac{||\tilde \Theta - \Theta_0||}{||\Theta_0||} \leq \frac{||\tilde L - L_0|| + ||\tilde \Sigma - \sigma_0^2 I_p||}{||\Theta_0||} \lesssim \sqrt{\frac{\log n}{n}},
$$
with probability at least $1-o(1)$.
\end{proof}

\begin{proof}[Proof of Theorem \ref{thm:clt}]
   First, we show that the entries of $\hat L$ are asymptotically appropriately close to the ones of $\hat L_{svd}$. Let $\bar \lambda_u =  D_{1:H} v_{u,1:H} / \sqrt{n}$ for $u=1, \dots, p$.
  Denote by $\hat l_{jj'}$ and $\hat l_{svd, jj'}$ the $j,j'$-th entries of $\hat L$ and $\hat L_{svd}$ respectively. 
  Moreover, let $\Theta_{svd} =\hat L_{svd} + \hat \Sigma_{svd}$, where $\hat \Sigma_{svd} = \text{diag}(\delta_{svd, 1}^2, \cdots, \delta_{svd, p}^2 )$.
Denote by $\hat \theta_{svd, jj'}$ and $\theta_{0jj'}$ the $j,j'$-th entries of $\hat \Theta_{svd}$ and $\Theta_0$ respectively.
  
  Note that $\hat l_{jj'} = \bar \lambda_{j}^\top B_jB_{j'} \bar \lambda_{j'}$ and $\hat l_{svd, jj'} = \bar \lambda_{j}^\top B^{2} \bar \lambda_{j'}$, respectively, where 
  \begin{equation}\label{eq:B_u}
      B_u = n \Psi_{n, j}^{-1} \text{diag}\bigg(\frac{n}{n + \tau_j^{-2} \psi_1^2}, \cdots, \frac{n}{n + \tau_j^{-2} \psi_H^2}\bigg), \quad (u = j,j'),
  \end{equation}
  and 
  \begin{equation}\label{eq:B}
      B = \begin{bmatrix}
       I_k &\mathbf{O}_{k,H-k}\\
       \mathbf{O}_{H-k,k} & \mathbf{O}_{H-k,H-k}\\
   \end{bmatrix},
  \end{equation}
   where $\mathbf 0_{d_1, d_2} \in \mathbb R^{d_1 \times d_2}$ denotes the zero-matrix of dimension $d_1 \times d_2$. 
   Hence, with probability at least $1-o(1)$, we have
   \begin{equation*}
       \begin{aligned}
         \sqrt{n}\left|\hat l_{jj'} - \hat l_{svd, jj'}\right| =&  \sqrt{n} |\bar \lambda_{j}^\top \big(B_jB_{j'} - B^{2} \big) \bar \lambda_{j'}| \leq  \sqrt{n} ||\bar \lambda_j|| ||\bar \lambda_{j'}|| ||B_jB_{j'} - B^{2}|| \\
           & \leq  \sqrt{n} \left(\max \left\{\frac{n}{n + \tau_j^{-2} \psi_k^2}-1, \frac{n}{n + \tau_{j'}^{-2} \psi_k^2}-1, \frac{n}{n + \tau_j^{-2} \psi_{k+1}^2}, \frac{n}{n + \tau_{j'}^{-2} \psi_{k+1}^2}\right\}\right)^2 \\
           &\asymp \sqrt{n} \bigg( \frac{1}{n} + \frac{1}{n^{4/5}} \bigg) \asymp \frac{1}{n^{3/10}},
       \end{aligned}
   \end{equation*}
where we used Proposition \ref{prop:psi_l} and Lemma \ref{lemma:hyperparams} to bound the ratios in the second line and the fact that $||\bar \lambda_j|| ||\bar \lambda_{j'}|| \leq ||y^{(j)}|| ||y^{(j')}|| / n \lesssim 1$, with probability at least $1-o(1)$, by Lemma \ref{lemma:y_j}. 
Next, let $\delta_{svd, j}^2 = ||(I_n - U_{1:k}U_{1:k})y^{(j)}||_F^2 / n$ and, using the steps in the proof of Lemma \ref{lemma:delta_j}, we can also show $$ \sqrt{n}|\delta_j^2-  \delta_{svd, j}^2| \lesssim \sqrt{n}  \frac{1}{n} \asymp \frac{1}{\sqrt{n}}$$ for each $j=1, \dots, p$.
Combining the two bounds, we obtain, with probability at least $1-o(1)$,  
$$
 \sqrt{n}\left|\hat \theta_{jj'}- \hat \theta_{svd, jj'} \right| = o_{pr}(1).
$$
It remains to prove a central limit theorem for \(\hat\Theta_{svd}\). Let
\[
P=U_{1:k}U_{1:k}^\top,\qquad P_0=U_{0, 1:k}U_{0, 1:k}^\top,\qquad m^{(j)}=M_0\lambda_{0j}, \quad (j=1, \dots, p).
\]
For \(j\neq j'\), since \(\hat\theta_{svd,jj'}=n^{-1}y^{(j)\top}Py^{(j')}\), we write
\begin{equation*}
\hat\theta_{svd,jj'}-\theta_{0jj'}=D_{jj'}+R_{jj'},
\end{equation*}
where
\begin{equation*}
D_{jj'}=\frac{1}{n}m^{(j)\top}m^{(j')}-\lambda_{0j}^\top\lambda_{0j'}+\frac{1}{n}m^{(j)\top}\epsilon^{(j')}+\frac{1}{n}m^{(j')\top}\epsilon^{(j)}
\end{equation*}
and
\begin{equation*}
R_{jj'}=\frac{1}{n}m^{(j)\top}(P-P_0)m^{(j')}+\frac{1}{n}m^{(j)\top}(P-P_0)\epsilon^{(j')}+\frac{1}{n}m^{(j')\top}(P-P_0)\epsilon^{(j)}+\frac{1}{n}\epsilon^{(j)\top}P\epsilon^{(j')}.
\end{equation*}
We first show that \(\sqrt n R_{jj'}=o_{pr}(1)\). Since \(m^{(j)},m^{(j')}\in \operatorname{col}(P_0)\), we have \(\|P_0(P-P_0)P_0\|\leq \|P-P_0\|^2\).
Moreover, we have  \(\|P-P_0\|=\mathcal O_{pr}(n^{-1/2})\), by Lemma \ref{lemma:U} and \(\|m^{(u)}\| \leq \|M_0\| ||\lambda_{0j}|| \lesssim \sqrt{n}\) with $u=j,j'$, since $\|M_0\| = O_{pr}(\sqrt n)$ by Lemma \ref{lemma:sv_F}. 
Hence, 
\begin{equation*}
|m^{(j)\top}(P-P_0)m^{(j')}|
\leq
\|m^{(j)}\|\|P_0(P-P_0)P_0\|\|m^{(j')}\|
=
\mathcal O_{pr}(1),
\end{equation*}
and therefore \(n^{-1/2}m^{(j)\top}(P-P_0)m^{(j')}=o_{pr}(1)\). Next, let \(P^{(-j')} = U_{1:k}^{(-j')}U_{1:k}^{(-j')\top}\), where $U_{1:k}^{(-j')} \in \mathbb R^{n \times k}$ is the matrix of the left singular vectors associated to the leading-\(k\) singular values of the matrix $Y$ after deleting the \(j'\)-th column. Then
\begin{equation*}
m^{(j)\top}(P-P_0)\epsilon^{(j')}=m^{(j)\top}(P-P^{(-j')})\epsilon^{(j')}+m^{(j)\top}(P^{(-j')}-P_0)\epsilon^{(j')}.
\end{equation*}
By Lemma \ref{lemma:P_stability}, \(\|P-P^{(-j')}\|=\mathcal O_{pr}(n^{-1})\). Since \(\|m^{(j)}\|=\mathcal O_{pr}(\sqrt n)\) and \(\|\epsilon^{(j')}\|=\mathcal O_{pr}(\sqrt n)\) by Lemma \ref{lemma:y_j},
\begin{equation*}
|m^{(j)\top}(P-P^{(-j')})\epsilon^{(j')}|
\leq
\|m^{(j)}\|\|P-P^{(-j')}\|\|\epsilon^{(j')}\|
=
\mathcal O_{pr}(1).
\end{equation*}
 Conditional on \(P^{(-j')}\), the vector \(\epsilon^{(j')}\) is independent of \(P^{(-j')}\), and
\begin{equation*}
E\!\left[\{m^{(j)\top}(P^{(-j')}-P_0)\epsilon^{(j')}\}^2\mid P^{(-j')}\right]
=
\sigma_0^2\|(P^{(-j')}-P_0)m^{(j)}\|^2.
\end{equation*}
Since \(m^{(j)}\in\operatorname{col}(P_0)\), \(\|(P^{(-j')}-P_0)m^{(j)}\|\leq \|(P^{(-j')}-P_0)P_0\|\|m^{(j)}\|=\mathcal O_{pr}(1)\). Hence the conditional second moment is \(\mathcal O_{pr}(1)\), and Chebyshev's inequality gives
\begin{equation*}
m^{(j)\top}(P^{(-j')}-P_0)\epsilon^{(j')}=\mathcal O_{pr}(1).
\end{equation*}
Thus \(m^{(j)\top}(P-P_0)\epsilon^{(j')}=\mathcal O_{pr}(1)\), and similarly \(m^{(j')\top}(P-P_0)\epsilon^{(j)}=\mathcal O_{pr}(1)\). Finally, consider the following decomposition 
\begin{equation*}
\epsilon^{(j)\top}P\epsilon^{(j')}=\epsilon^{(j)\top}P_0\epsilon^{(j')}+\epsilon^{(j)\top}(P-P_0)\epsilon^{(j')}.
\end{equation*}
The first term is \(\mathcal O_{pr}(1)\), since $\epsilon^{(j)\top}P_0\epsilon^{(j')} = \epsilon^{(j)\top} U_{0,1:k} U_{0,1:k}^\top\epsilon^{(j')}$, and $U_{0, 1:k}^\top\epsilon^{(u)} \sim N_k(0, I_k)$, with $u  = j,j'$. For the second term, let $P^{(-jj')} = U_{1:k}^{(-jj')} U_{1:k}^{(-jj')\top}$, where $U_{1:k}^{(-jj')} \in \mathbb R^{n \times k}$ be the matrix of left singular vectors associated to the leading $k$ singular values of $Y$ after removing the $j,j'$-th columns.
Then
\begin{equation*}
\epsilon^{(j)\top}(P-P_0)\epsilon^{(j')}=\epsilon^{(j)\top}(P-P^{(-jj')})\epsilon^{(j')}+\epsilon^{(j)\top}(P^{(-jj')}-P_0)\epsilon^{(j')}.
\end{equation*}
By Lemma \ref{lemma:P_stability}, \(\|P-P^{(-jj')}\|=\mathcal O_{pr}(n^{-1})\), and by Lemma \ref{lemma:y_j}, \(\|\epsilon^{(j)}\|,\|\epsilon^{(j')}\|=\mathcal O_{pr}(\sqrt n)\). Hence
\begin{equation*}
|\epsilon^{(j)\top}(P-P^{(-jj')})\epsilon^{(j')}|
\leq
\|\epsilon^{(j)}\|\|P-P^{(-jj')}\|\|\epsilon^{(j')}\|
=
\mathcal O_{pr}(1).
\end{equation*}
Conditional on \(P^{(-jj')}\), the vectors \(\epsilon^{(j)}\) and \(\epsilon^{(j')}\) are independent of \(P^{(-jj')}\), and hence
\begin{equation*}
E\!\left[\{\epsilon^{(j)\top}(P^{(-jj')}-P_0)\epsilon^{(j')}\}^2\mid P^{(-jj')}\right]
=
\sigma_0^4\|P^{(-jj')}-P_0\|_F^2.
\end{equation*}
Since \(\|P^{(-jj')}-P_0\|_F=\mathcal O_{pr}(n^{-1/2})\), Chebyshev's inequality gives
\begin{equation*}
\epsilon^{(j)\top}(P^{(-jj')}-P_0)\epsilon^{(j')}=\mathcal O_{pr}(n^{-1/2}).
\end{equation*}
Therefore \(\epsilon^{(j)\top}P\epsilon^{(j')}=\mathcal O_{pr}(1)\). Combining all of the above gives \(\sqrt n R_{jj'}=o_{pr}(1)\).

It remains to derive the limiting distribution of \(D_{jj'}\). Writing \(M_0= \begin{bmatrix}
    \eta_1,\ldots,\eta_n
\end{bmatrix}^\top\), we have
\begin{equation*}
\sqrt n D_{jj'}=\frac{1}{\sqrt n}\sum_{i=1}^n \xi_{i,jj'},
\end{equation*}
where
\begin{equation*}
\xi_{i,jj'}=(\lambda_{0j}^\top\eta_i)(\lambda_{0j'}^\top\eta_i)-\lambda_{0j}^\top\lambda_{0j'}+(\lambda_{0j}^\top\eta_i)\epsilon_{ij'}+(\lambda_{0j'}^\top\eta_i)\epsilon_{ij}.
\end{equation*}
The variables \(\xi_{i,jj'}\) are independent and identically distributed with mean zero. Standard calculations give
\begin{equation*}
\operatorname{Var}(\xi_{i,jj'})=\sigma_0^2(\|\lambda_{0j}\|^2+\|\lambda_{0j'}\|^2)+(\lambda_{0j}^\top\lambda_{0j'})^2+\|\lambda_{0j}\|^2\|\lambda_{0j'}\|^2, \quad j \neq j'.
\end{equation*}
Therefore the classical central limit theorem gives
\begin{equation*}
\sqrt n(\hat\theta_{svd,jj'}-\theta_{0jj'})\Rightarrow N(0,\operatorname{Var}(\xi_{i,jj'})),\qquad j\neq j'.
\end{equation*}

For the diagonal entries, observe that
\begin{equation*}
\hat\theta_{svd,jj}=\frac{1}{n}y^{(j)\top}Py^{(j)}+\frac{1}{n}y^{(j)\top}(I_n-P)y^{(j)}=\frac{1}{n}\|y^{(j)}\|^2.
\end{equation*}
Since \(y_{ij}=\lambda_{0j}^\top\eta_i+\epsilon_{ij}\sim N(0,\|\lambda_{0j}\|^2+\sigma_0^2)\), we have \(\theta_{0jj}=\|\lambda_{0j}\|^2+\sigma_0^2\) and
\begin{equation*}
\sqrt n(\hat\theta_{svd,jj}-\theta_{0jj})=\frac{1}{\sqrt n}\sum_{i=1}^n\{y_{ij}^2-(\|\lambda_{0j}\|^2+\sigma_0^2)\}\Rightarrow N\!\left(0,2(\|\lambda_{0j}\|^2+\sigma_0^2)^2\right).
\end{equation*}
Thus, for all fixed \(j,j'\),
\begin{equation*}
\sqrt n(\hat\theta_{svd,jj'}-\theta_{0jj'})\Rightarrow N(0,S_{0jj'}^2),
\end{equation*}
where 
       \begin{equation}\label{eq:S_0_sq}
              S_{0jj'}^2 =    \begin{cases}
                \sigma_{0}^2 \big(||\lambda_{0j}||^2  + ||\lambda_{0j'}||^2\big) + (\lambda_{0j}^\top \lambda_{0j'})^2 + ||\lambda_{0j}||^2 ||\lambda_{0j'}||^2,\quad & \text{if } j \neq j' \\
                   2\big(\sigma_{0}^2 + ||\lambda_{0j'}||^2 \big)^2, \quad & \text{if } j=j' 
                 \\
             \end{cases}
       \end{equation}
Combining this result with \(\sqrt n|\hat\theta_{jj'}-\hat\theta_{svd,jj'}|=o_{pr}(1)\), Slutsky's theorem yields
\begin{equation*}
\sqrt n(\hat\theta_{jj'}-\theta_{0jj'})\Rightarrow N(0,S_{0jj'}^2),
\end{equation*}
which concludes the proof.

\end{proof}

\begin{proof}[Proof of Theorem \ref{thm:bvm}]
Fix $j,j'\in\{1,\ldots,p\}$. From \eqref{eq:posterior_update_cc}, we can write
\begin{equation*}
\tilde\lambda_j=\hat\lambda_j+\rho\,\tilde\sigma_j\Psi_{n,j}^{-1/2}z_j,\qquad 
z_j\sim N_H(0,I_H),
\end{equation*}
independently across $j$, for $j=1, \dots, p$.
Let
\begin{equation*}
\tilde\theta_{jj'}=\tilde\lambda_j^\top\tilde\lambda_{j'}+\tilde\sigma_j^2\mathbf 1(j=j'),\qquad 
\hat \theta_{jj'}=\hat\lambda_j^\top\hat\lambda_{j'}+\delta_j^2\mathbf 1(j=j').
\end{equation*}
First suppose $j\ne j'$. Expanding around $\hat \theta_{jj'}$ gives
\begin{align}\label{eq:decomp_tilde_theta}
\tilde\theta_{jj'}-\hat \theta_{jj'}
&=
\rho\tilde\sigma_{j'}\hat\lambda_j^\top\Psi_{n,j'}^{-1/2}z_{j'}
+\rho\tilde\sigma_j\hat\lambda_{j'}^\top\Psi_{n,j}^{-1/2}z_j
+\rho^2\tilde\sigma_j\tilde\sigma_{j'}z_j^\top\Psi_{n,j}^{-1/2}\Psi_{n,j'}^{-1/2}z_{j'} .
\end{align}
The first two terms are conditionally Gaussian with mean zero and, after multiplying by $\sqrt n$, have conditional variance
\begin{equation*}
T_{jj'}^2(\rho) = n\rho^2\tilde\sigma_{j'}^2\hat\lambda_j^\top\Psi_{n,j'}^{-1}\hat\lambda_j
+n\rho^2\tilde\sigma_j^2\hat\lambda_{j'}^\top\Psi_{n,j}^{-1}\hat\lambda_{j'} .
\end{equation*}
By proposition \ref{prop:psi_l} and Lemma \ref{lemma:hyperparams}, with probability at least $1-o(1)$, we have 
$$
||B_u-  B|| \lesssim \frac{1}{\sqrt{n}} + \frac{1}{n^{2/5}} \asymp \frac{1}{n^{2/5}}, \quad (u = j,j'),
$$
where $B_u$ and $B$ are defined in \eqref{eq:B_u} and \eqref{eq:B} respectively, which, together with $l_{svd, uu} \overset{pr}{\to} ||\lambda_{0u}||^2$ for $u=j,j'$, implies 
$$
 n\hat \lambda_u^\top \Psi_{n, u}^{-1}\hat\lambda_u \overset{pr}{\to} ||\lambda_{0u}||^2, \quad (u=j,j').
$$
Moreover, by the proof of Theorem \ref{thm:consistency}, we have 
$\tilde \sigma_u^2 \overset{pr}{\to} \sigma_0^2.
$
Thus, for $j\neq j'$, we have $T_{jj'}^2(\rho) \overset{pr}{\to} \rho^2\sigma_0^2 \big(||\lambda_{0j}||^2+ ||\lambda_{0j'}||^2\big)$, which implies the sum of the first two terms (scaled by $\sqrt{n}$) is asymptotically Normally distributed with zero-mean and variance $T_{0jj'}^2(\rho)$. The last term in the expansion is negligible since $||\sqrt n \rho^2\tilde\sigma_j\tilde\sigma_{j'}z_j^\top\Psi_{n,j}^{-1/2}\Psi_{n,j'}^{-1/2}z_{j'} || \leq  |\rho^2\tilde\sigma_j\tilde\sigma_{j'} | ||\sqrt n \Psi_{n,j}^{-1/2} \Psi_{n, j'}^{-1/2}||z_j |||| z_{j'}||$ and $||\sqrt n \Psi_{n,j}^{-1/2} \Psi_{n, j'}^{-1/2}|| \lesssim \frac{1}{\sqrt{n}}$ with probability at least $1- o(1)$. 
Therefore, by Lemma F.2 of \citet{fable}, we have
$$
\sup_{x \in \mathbb R} \left| \Pi \left\{\frac{\sqrt{n} (\tilde \theta_{jj'} - \hat \theta_{jj'})}{T_{0jj'}(\rho)} \leq x\right\} - \Phi(x)  \right| \overset{pr}{\to} 0.
$$
Now suppose $j=j'$. Then
\begin{align*}
\tilde\theta_{jj}- \hat \theta_{jj'}
&=
\|\tilde\lambda_j\|^2-\|\hat\lambda_j\|^2+\tilde\sigma_j^2-\delta_j^2 \\
&=
2\rho\tilde\sigma_j\hat\lambda_j^\top\Psi_{n,j}^{-1/2}z_j
+\rho^2\tilde\sigma_j^2 z_j^\top\Psi_{n,j}^{-1}z_j
+\tilde\sigma_j^2-\delta_j^2,.
\end{align*}
with $z_j$ and $\tilde \sigma_j$ being independent. Similarly as above, the first term is conditionally Gaussian with mean zero and, after multiplying by $\sqrt n$, variance
\begin{equation*}
T_{jj}^2(\rho) = 4n\rho^2\tilde\sigma_j^2\hat\lambda_j^\top\Psi_{n,j}^{-1}\hat\lambda_j \overset{pr}{\to} 4\rho^2\sigma_0^2\|\lambda_{0j}\|^2,
\end{equation*}
while the second term is negligible, since $| \sqrt{n}\rho^2\tilde\sigma_j^2 z_j^\top\Psi_{n,j}^{-1}z_j| \leq | \rho^2\tilde\sigma_j^2|||z_j||  ||z_j|| || \sqrt n \Psi_{n,j}^{-1} || $ and $|| \sqrt n \Psi_{n,j}^{-1} ||  \lesssim \frac{1}{\sqrt{ n}}$, with probability at least $1-o(1)$. 
We now handle the last term. Using the decomposition of the proof of Theorem \ref{thm:consistency}, we write
$$
\sqrt{n} (\tilde \sigma_j^2 - \delta_j^2) = - \frac{2 \sqrt{n} \Delta_j}{v_n} \sigma_0^2 + \mathcal G_j, \quad (j=1, \dots, p)
$$
with 
$$
\mathcal G_j = - \frac{2 \sqrt{n} (\delta_j^2 - \sigma_0^2) \Delta_j / v_n}{1 +2 \Delta_j / v_n } + \frac{4 \sqrt{n} \sigma_0^2 \Delta_j^2 / v_n^2}{1 + 2 \Delta_j / v_n}. 
$$
Note that $ \frac{ 2 \sqrt{n}\Delta_j}{v_n} \sigma_0^2 \Rightarrow N(0, 2 \sigma_0^4)$, due to the central limit theorem, which also implies $|\mathcal G_j| = o(1)$ with probability at least $1-o(1)$. Note that by construction $z_j$ is independent of $\frac{ 2 \sqrt{n}\Delta_j}{v_n} \sigma_0^2$. Hence, we can complete the proof using Lemmas F.1 and F.2 of \citet{fable}, which give
$$
\sup_{x \in \mathbb R} \left| \Pi \left\{\frac{\sqrt{n} (\tilde \theta_{jj} - \hat \theta_{jj})}{T_{0jj}(\rho)} \leq x\right\}  - \Phi(x) \right| \overset{pr}{\to} 0,
$$
with $T_{0jj}^2(\rho) = 2 \sigma_0^4 + 4 \rho^2 \sigma_0^2||\lambda_{0j}||^2$.

\end{proof}

\subsection{Auxiliary lemmas}

\begin{lemma}\label{lemma:sv_F}
    Let $F \in \mathbb R^{n \times k}$ be a matrix with independent standard Gaussian entries. Then, with probability at least $1-o(1)$, we have
    \begin{equation*}
      \sqrt{n} - C \sqrt{k} \leq   s_l(F) \leq   \sqrt{n} + C\sqrt{k}, \quad (l=1,\dots,k),
    \end{equation*}
    for a universal constant $C < \infty$. 
    Moreover, with probability at least $1-o(1)$, we have
    $$
    \left| \left| \frac{F^\top F}{n} - I_k \right|\right| \lesssim \sqrt{\frac{\log n}{n}}.
    $$
\end{lemma}
\begin{proof}[Proof of Lemma \ref{lemma:sv_F}]
    See Corollary 5.35 and Lemma 5.36 of \citep{vershynin_12}.
\end{proof}

\begin{lemma}\label{lemma:sv_signal}
Suppose Assumptions \ref{assumption:model} and \ref{assumption:Lambda} hold, then, with probability at least $1-o(1)$, we have
    \begin{equation*}
      \sqrt{np} - C \sqrt{kp} \leq   s_l(M \Lambda_0^\top) \leq   \sqrt{np} + C\sqrt{kp}, \quad (l=1,\dots,k),
    \end{equation*}
    for a universal constant $C < \infty$.
\end{lemma}
\begin{proof}[Proof of Lemma \ref{lemma:sv_signal}]
    It follows from Lemma \ref{lemma:sv_F} and Assumption \ref{assumption:Lambda}
\end{proof}

\begin{lemma}\label{lemma:E}
  Suppose Assumptions \ref{assumption:model}, \ref{assumption:homoscedasticity} hold. Then, with probability at least $1-o(1)$,
  $$||E|| \lesssim \sqrt{n} + \sqrt{p}.$$
\end{lemma}
\begin{proof}[Proof of Lemma \ref{lemma:E}]
    It follows from Corollary 3.11 of \citet{bandeira_van_handel_16}. 
\end{proof}

\begin{lemma}\label{lemma:sv_Y}
     Suppose Assumptions \ref{assumption:model}--\ref{assumption:homoscedasticity} hold, then, for $n$ sufficiently large, with probability at least $1-o(1)$,
     \begin{equation*}
         \frac{d_l^2}{n p} = \mathcal \mathcal O( \mathbf{1}_{\{l\leq k\}}) + o(\mathbf{1}_{\{l> k\}}), \quad (l=1, \dots, H)
     \end{equation*}
\end{lemma}
\begin{proof}[Proof of lemma \ref{lemma:sv_Y}]
    It follows from Lemmas \ref{lemma:sv_signal} and \ref{lemma:E}
\end{proof}

\begin{lemma}\label{lemma:y_j}
    Suppose Assumptions \ref{assumption:model}, \ref{assumption:homoscedasticity} hold. Then, with probability at least $1-o(1)$,
    \begin{equation*}
       \max_{j=1, \dots, p} ||y^{(j)}|| \lesssim \sqrt{n}, \quad  \max_{j=1, \dots, p} ||e^{(j)}|| \lesssim \sqrt{n}.
    \end{equation*}
\end{lemma}
\begin{proof}[Proof of Lemma \ref{lemma:y_j}]
First, note that 
$||e^{(j)}|| \lesssim \sqrt{n }$ with probability at least $1-o(1)$, by Lemma \ref{lemma:y_j}. Then, note that $||y^{(j)}|| \leq ||M|| ||\lambda_j|| + ||e^{(j)}||$ and $ \max_{j=1, \dots, p}||M|| ||\lambda_j|| \lesssim \sqrt{n}\max_{j=1, \dots, p}||\lambda_j|| \asymp \sqrt{n}$, since $||M|| \lesssim \sqrt{n}$  with probability at least $1-o(1)$, by Lemma \ref{lemma:sv_F} and $\max_{j=1, \dots, p}||\lambda_j|| \leq k C_\Lambda < \infty$ by Assumption \ref{assumption:Lambda}.
\end{proof}

\begin{lemma}\label{lemma:U}
Suppose Assumptions \ref{assumption:model}--\ref{assumption:homoscedasticity} hold. Let $M\Lambda_0^\top = U_{0, 1:k}D_{0, 1:k}V_{0, 1:k}^\top$ 
be the singular value decomposition of the signal matrix, with
$U_0\in\mathbb R^{n\times k}$ having orthonormal columns. Let $U_{1:k}$ be the matrix of left singular vectors of $Y$ associated to its leading $k$ singular values. Then, 
\begin{equation*}
pr\left\{
\left\|U_{1:k}U_{1:k}^\top - U_{0, 1:k}U_{0, 1:k}^\top\right\|
>
C\left(\frac{1}{\sqrt n}+\frac{1}{\sqrt p}\right)
\right\}\to 0,
\end{equation*}
for some constant $C < \infty$. 
\end{lemma}
\begin{proof}[Proof of Lemma \ref{lemma:U}]
Let $\Delta=U_{1:k}U_{1:k}^\top-U_{0, 1:k}U_{0, 1:k}^\top$. Since $X_0=M_0\Lambda_0^\top$,
\begin{equation*}
\|X_0^\top\Delta X_0\|=\|\Lambda_0 M_0^\top\Delta M_0\Lambda_0^\top\|\ge s_k^2(\Lambda_0)\|M_0^\top\Delta M_0\|,
\end{equation*}
where the inequality follows because $\Lambda_0$ has full column rank. We now lower bound the second factor. Since $\Lambda_0$ has full column rank, $\operatorname{col}(X_0)=\operatorname{col}(M_0)$; moreover, from the singular value decomposition $X_0=U_{0, 1:k}D_{0, 1:k}V_{0, 1:k}^\top$, $\operatorname{col}(X_0)=\operatorname{col}(U_{0, 1:k})$. Hence $\operatorname{col}(M_0)=\operatorname{col}(U_{0, 1:k})$, and since both matrices have rank $k$, there exists an invertible $k\times k$ matrix $R$ such that $M_0=U_{0, 1:k}R$. Thus
\begin{equation*}
\|M_0^\top\Delta M_0\|=\|R^\top(U_{0, 1:k}^\top\Delta U_{0, 1:k})R\|\ge s_k^2(R)\|U_{0, 1:k}^\top\Delta U_{0, 1:k}\|=s_k^2(M_0)\|U_{0, 1:k}^\top\Delta U_{0, 1:k}\|.
\end{equation*}
Also, $s_k(R)=s_k(M_0)$ because $M_0^\top M_0=R^\top U_{0, 1:k}^\top U_{0, 1:k}R=R^\top R$.
 Since \( (U_{0, 1:k}U_{0, 1:k}^\top) U_{0,1:k}=U_{0,1:k}\), 
\begin{equation*}
U_{0,1:k}^\top(U_{1:k}U_{1:k}^\top -U_{0, 1:k}U_{0, 1:k}^\top)U_{0,1:k}
=U_{0,1:k}^\top U_{1:k}U_{1:k}^\top U_{0,1:k}-I_k
=C^\top C-I_k,
\end{equation*}
where \(C=U_{1:k}^\top U_{0,1:k}\). 
Let \(\alpha_1,\ldots,\alpha_k\) be the principal angles between \(\operatorname{col}(U_{1:k})\) and \(\operatorname{col}(U_{0,1:k})\). The singular values of \(C\) are \(\cos(\alpha_1),\ldots,\cos(\alpha_k)\), so the eigenvalues of \(C^\top C-I_k\) are \(-\sin^2(\alpha_1),\ldots,-\sin^2(\alpha_k)\). Hence
\begin{equation*}
\|U_{0,1:k}^\top(U_{1:k}U_{1:k}^\top -U_{0, 1:k}U_{0, 1:k}^\top )U_{0,1:k}\|
=\max_{r = 1, \dots, k}\sin^2(\alpha_r)
=\|U_{1:k}U_{1:k}^\top -U_{0, 1:k}U_{0, 1:k}^\top\|^2,
\end{equation*}
where the last equality uses the identity \(\|U_{1:k}U_{1:k}^\top -U_{0, 1:k}U_{0, 1:k}^\top\|=\max_{1\le r\le k}\sin(\alpha_r)\) (see \citet{stewart1990matrix}).
Moreover,
\begin{equation*}
\|U_{0, 1:k}^\top\Delta U_{0, 1:k}\|=\|U_{0, 1:k}^\top U_{1:k}U_{1:k}^\top U_{0, 1:k}-I_k\|=
\|U_{1:k}U_{1:k}^\top-U_{0, 1:k}U_{0, 1:k}^\top\|^2.
\end{equation*}
Combining the previous results gives
\begin{equation*}
\|X_0^\top\Delta X_0\|\ge s_k^2(\Lambda_0)s_k^2(M_0)\|U_{1:k}U_{1:k}^\top-U_{0, 1:k}U_{0, 1:k}^\top\|^2.
\end{equation*}
We next upper bound the same quantity. Since $U_{0, 1:k}U_{0, 1:k}^\top X_0=X_0$,
\begin{equation*}
\|X_0^\top\Delta X_0\|=\|X_0^\top(U_{1:k}U_{1:k}^\top-I_n)X_0\|=\|X_0^\top(I_n-U_{1:k}U_{1:k}^\top)X_0\|=\|(I_n-U_{1:k}U_{1:k}^\top)X_0\|^2.
\end{equation*}
Because $U_{1:k}$ contains the leading $k$ left singular vectors of $Y=X_0+E$, $U_{1:k}U_{1:k}^\top Y$ is a best rank-$k$ approximation to $Y$ in spectral norm. Hence
\begin{equation*}
\|(I_n-U_{1:k}U_{1:k}^\top)Y\|\le \|(I_n-U_{0, 1:k}U_{0, 1:k}^\top)Y\|=\|(I_n-U_{0, 1:k}U_{0, 1:k}^\top)E\|\le \|E\|.
\end{equation*}
It follows that
\begin{equation*}
\|(I_n-U_{1:k}U_{1:k}^\top)X_0\|=\|(I_n-U_{1:k}U_{1:k}^\top)(Y-E)\|\le \|(I_n-U_{1:k}U_{1:k}^\top)Y\|+\|(I_n-U_{1:k}U_{1:k}^\top)E\|\le 2\|E\|,
\end{equation*}
and therefore $\|X_0^\top\Delta X_0\|\le 4\|E\|^2$. Combining this with the lower bound yields
\begin{equation*}
\|U_{1:k}U_{1:k}^\top-U_{0, 1:k}U_{0, 1:k}^\top\|\le \frac{2\|E\|}{s_k(\Lambda_0)s_k(M_0)}.
\end{equation*}
Using the fact that, with probability at least $1-o(1)$, $\|E\| \lesssim \sqrt n+\sqrt{p}$, $s_k(M_0)\asymp \sqrt n$, by Lemma \ref{lemma:E} and \ref{lemma:sv_F}, respectively, and $s_k(\Lambda_0)\asymp \sqrt{p}$, by Assumption \ref{assumption:Lambda}, we obtain, with probability at least $1-o(1)$, 
\begin{equation*}
\|U_{1:k}U_{1:k}^\top-U_{0, 1:k}U_{0, 1:k}^\top\|\lesssim \frac{\sqrt n+\sqrt{p}}{\sqrt{n p}}\asymp \frac{1}{\sqrt n}+\frac{1}{\sqrt p}.
\end{equation*}
\end{proof}

\begin{lemma}\label{lemma:P_stability}
Suppose Assumptions \ref{assumption:model}--\ref{assumption:homoscedasticity} hold. Fix $j, j' \in {1, \dotsm p}$, and denote by $Y^{(-j)}$ and $Y^{(-jj')}$ the matrices obtained by removing the $j$-th and the $j, j'$-th columns from $Y$ respectively. Let $P = U_{1:k} U_{1:k}^\top$, $P^{(-j)} = U_{1:k}^{(-j)} U_{1:k}^{(-j)\top}$, and $P^{(-jj')} = U_{1:k}^{(-jj')} U_{1:k}^{(-jj')\top}$, where $U_{1:k}^{(-j)} \in \mathbb R^{n \times k}$ and $U_{1:k}^{(-jj')} \in \mathbb R^{n \times k}$ are the matrices of left singular vectors associated to the leading $k$ singular values of $Y^{(-j)}$ and $Y^{(-jj')}$, respectively. 
 Then, with probability at least $1-o(1)$, 
\[
\|P-P^{(-j)}\| \lesssim \frac{1}{n}, \quad \|P-P^{(-jj')}\| \lesssim \frac{1}{n}
\]
\end{lemma}
\begin{proof}[Proof of Lemma \ref{lemma:P_stability}]
We only prove the second result, as the first one can be proven by the same steps.  
Let
\[
A=YY^\top,\qquad A^{(-jj')}=Y^{(-jj')}Y^{(-jj')\top}.
\]
Since \(Y^{(-jj')}\) is obtained from \(Y\) by deleting the $j,j'$-th columns,
\begin{equation*}
A-A^{(-jj')}=y^{(j)}y^{(j)\top}+y^{(j')}y^{(j')\top}.
\end{equation*}
Therefore, with proabability at least $1-o(1)$
\begin{equation*}
\|A-A^{(-jj')}\|
\leq \|y^{(j)}\|^2+\|y^{(j')}\|^2 \lesssim n 
\end{equation*}
since $\max_{j=1, \dots, p}||y^{(j)}|| \lesssim \sqrt{n }$, with probability at least $1-o(1)$, by Lemma \ref{lemma:y_j}.

Note that the $k$-th and $k+1$-th eigenvalues of $A$ are $d_k^2$ and $d_{k+1}^2$, respectively and $d_k^2 \asymp n^2$ and $d_{k+1}^2 \asymp n$, with probability at least $1-o(1)$, by Lemma \ref{lemma:sv_Y}. Hence, since $A$ has a positive eigengap between the $k$-th and $k+1$-th eigenvalues, with probability at least $1-o(1)$, by Lemma \ref{lemma:sv_Y}, by the Davis--Kahan theorem applied to the leading \(k\)-dimensional eigenspaces of \(A\) and \(A^{(-jj')}\), 
we have
\begin{equation*}
\|P-P^{(-jj')}\|
\leq
C\frac{\|A-A^{(-jj')}\|}
{d_k^2 - d_{k+1}^2}
\lesssim \frac{n}{n^2} = \frac{1}{n},
\end{equation*}
with probability at least $1-o(1)$. 
\end{proof}


\begin{lemma}\label{lemma:norm_Lambda_hat}
    Suppose Assumption \ref{assumption:model}--\ref{assumption:homoscedasticity} hold and $H = \mathcal O(n^{2/5})$.  Let $\hat \Lambda_{1:k} \in \mathbb R^{p \times k}$ and $\hat \Lambda_{k+1:H} \in \mathbb R^{p \times H-k}$ be the matrices obtained by extracting the first $k$ columns and the remaining $H-k$ ones of $\hat \Lambda$ respectively. Then, with probability at least $1-o(1)$, we have
    $$
    ||\hat \Lambda_{1:k} ||_F  \lesssim \sqrt{n}, \quad ||\hat \Lambda_{k+1:H}||_F \lesssim  \frac{1}{n^{1/5}},
    $$
    which, in turn, imply
    $$
    ||\hat \Lambda||_F \lesssim \sqrt{n}.
    $$
\end{lemma}
\begin{proof}[Proof of Lemma \ref{lemma:norm_Lambda_hat}]
  Note     $$||\hat \Lambda||_F  \leq ||\hat \Lambda_{1:k} ||_F + ||\hat \Lambda_{k+1:H}||_F.$$ By Proposition \ref{prop:psi_l}, for $l > k$, we have $\max_{j=1, \dots, p} \frac{\sqrt{n}}{n + \hat \psi_l^{-2} \hat \tau_j^{-2}} \asymp \frac{1}{\sqrt{n}(1+n^{2/5})} \asymp \frac{1}{n^{9/10}}$, with probability at least $1-o(1)$. Thus, since, by Lemma \ref{lemma:sv_Y}, $d_{l} \lesssim \sqrt{n}$ for $l > k$, we have 
    $$||\hat \Lambda_{k+1:H}||_F \lesssim \sqrt{H}\frac{1}{n^{2/5}} \lesssim \frac{1}{n^{1/5}},$$ 
    with probability at least $1-o(1)$.
    Moreover, with probability at least $1-o(1)$, we have
    $$
    ||\hat \Lambda_{1:k} ||_F \leq \left|\left| \frac{V_{1:K} D_{1:K}}{\sqrt{n}}\right|\right|_F \leq k \frac{d_1}{\sqrt{n}} \lesssim k \sqrt{\frac{p n}{n}} \asymp \sqrt{n},
     $$
since, by Lemma \ref{lemma:sv_Y}, $d_1 = s_1(Y) \asymp \sqrt{np}$, with probability at least $1-o(1)$. 
    
\end{proof}

\begin{lemma}\label{lemma:sv_Y_diff}

Suppose Assumptions \ref{assumption:model}--\ref{assumption:homoscedasticity} hold and $r = O(n^{2/5})$, 
then
\begin{equation*}
d_{k+1}^2- d_{k+r}^2=\mathcal O_{pr}(n^{1/3} r^{2/3}).
\end{equation*}
\end{lemma}
\begin{proof}[Proof of Lemma \ref{lemma:sv_Y_diff}]
Denote the marginal variance of each $y_i$ with $\Theta_0 = \Lambda_0 \Lambda_0^\top + \sigma_0^2 I_p$. 
Let $Z\in\mathbb R^{n\times p}$ have independent standard Gaussian entries. Then $Y\stackrel{d}{=}Z\Theta_0^{1/2}$.
Define the normalized covariance
\begin{equation*}
\bar\Theta_0=\sigma_0^{-2}\Theta_0 = I_p+\sigma_0^{-2}\Lambda_0\Lambda_0^\top,
\end{equation*}
which is a rank-$k$ perturbation of $I_p$. 
Its nontrivial eigenvalues are
\begin{equation*}
 \gamma_l=1+\frac{s_l(\Lambda_0)^2}{\sigma_0^2} 
 ,\qquad (l=1,\dots,k),
\end{equation*}
so $ \gamma_l^2\to\infty$ and the spikes are uniformly supercritical \citep{lee_lee_24}. 
Denote by $\widetilde\gamma_1\ge\cdots\ge \widetilde\gamma_n\ge 0$ the eigenvalues of $S = \frac{1}{n \sigma_0^2} Y Y^\top$, and note that for each $l=1, \dots, n$, we have $d_l^2 = n \sigma_0^2 \widetilde \gamma_l$. Similarly, let $\widetilde\gamma^{(0)}_1\ge\cdots\ge \widetilde\gamma^{(0)}_n\ge0$ be the eigenvalues of $
\widetilde S_0=\frac{1}{n \sigma_0^2}EE^\top $.
Theorem 2.7 of \citet{Bloemendal_16} with $n \asymp p$ gives 
\begin{equation*}\label{eq:eigenvalues_stiking}
\widetilde\gamma_{k+r}=\widetilde\gamma^{(0)}_{r} + \mathcal O(n^{-2}),
\end{equation*}
with probability at least $1-o(1)$.

Let $\gamma_l^{(0)}$ denote the $l/p$-th  upper-tail quantile of $\rho_\gamma(\cdot)$, defined as in Eq.~(2.17) of \citet{bloemenda_et_al}, where $\rho_\gamma(\cdot)$ is the Marchenko-Pastur law \citep{MP_law}, for instance, defined in Eq.~(2.4) of \citet{bloemenda_et_al}. Then, for each fixed $r$, by Theorem 2.10 of \citet{bloemenda_et_al}, with probability at least $1-o(1)$,
\begin{equation*}
\widetilde\gamma^{(0)}_{r}=\gamma_r^{(0)} + \mathcal O\!\big(n^{-2/3}r^{-1/3}\big),
\end{equation*}
and also $\widetilde\gamma^{(0)}_{1}=\gamma_1^{(0)} + \mathcal O(n^{-2/3})$.
Therefore, with probability at least $1-o(1)$, 
\begin{equation*}
\label{eq:null-gap}
\widetilde\gamma^{(0)}_{1}-\widetilde\gamma^{(0)}_{r}
= (\gamma_1^{(0)}-\gamma_r^{(0)})+\mathcal O(n^{-2/3})\end{equation*}
From the explicit formula of $\rho_\gamma(\cdot)$, near the upper edge $\gamma_+$ the density
vanishes like a square root, so there exists $C<\infty$ such that for all sufficiently small $\delta>0$,
\begin{equation*}
\label{eq:mp-tail-bound}
\int_{\gamma_+-\delta}^{\gamma_+}\rho_{\mathrm{\gamma}}(x)\,dx \le C\,\delta^{3/2}.
\end{equation*}
By definition of $\gamma_r^{(0)}$ as an upper-tail quantile,
\begin{equation*}
\frac{r}{p}=\int_{\gamma_r^{(0)}}^{\gamma_+}\rho_\gamma(x)\,dx
\le \int_{\gamma_+-(\gamma_+-\gamma_r^{(0)})}^{\gamma_+}\rho_\gamma(x)\,dx
\le C(\gamma_+-\gamma_r^{(0)})^{3/2},
\end{equation*}
so $\gamma_+-\gamma_r^{(0)} = O\big((r/n)^{2/3}\big)$ and hence
\begin{equation*}
\gamma_1^{(0)}-\gamma_r^{(0)} = O\big((r/n)^{2/3}\big).
\end{equation*}
Combining all of the above, with probability at least $1-o(1)$, we get 
\begin{equation*}
d_{k+1}^2-d_{k+r}^2
= n \sigma_0^2(\widetilde\gamma_{k+1}-\widetilde\gamma_{k+r})
= \mathcal O\!\big(n^{1/3}r^{2/3}\big)
= \mathcal O(n^{3/5}),
\end{equation*}
In particular, if $r \asymp n^{2/5}$, with probability at least $1-o(1)$,
\begin{equation*}
d_{k+1}^2-d_{k+r}^2
= \mathcal O(n^{3/5}),
\end{equation*}
\end{proof}

\begin{lemma}\label{lemma:hyperparams}
    Suppose $H>k$, $H = \mathcal{O}(n^{2/5})$ and Assumptions \ref{assumption:model}--\ref{assumption:homoscedasticity} hold. Then, for each $ l \leq k$, with probability at least $1-o(1)$, we have 
    \begin{equation*}
        \min_{j=1, \dots, p} \hat \psi_l^2 \hat \tau_j^2 \gtrsim 1,  
    \end{equation*}
     where $\hat \psi_l^2$ and $\hat \tau_j^2$ are defined in Section \ref{subsec:hyperparams}. 
\end{lemma}
\begin{proof}[Proof of Lemma \ref{lemma:hyperparams}]
   Let $l \leq k$, by Lemma \ref{lemma:sv_Y}, we have $d_l^2 \asymp n^2 $, $d_{H+1}^2 \asymp n$, and $\sum_{l=1}^H(d_l^2 - d_{H+1}^2) \asymp k n p \asymp n^2$, with probability at least $1-o(1)$, implying $(d_l^2 -d_{H+1}^2) / \sum_{l'=1}^H(d_{l'}^2 - d_{H+1}^2) \asymp 1$. Moreover, we have $\max_{j = 1, \dots, p} \hat \sigma_j^2 \leq \max\{||y^{(j)}||^2 / (n-H), c_{\sigma}\} \lesssim 1$, with probability at least $1-o(1)$, by Lemma \ref{lemma:y_j}.
   Moreover, let $v_{j, 1:k} \in \mathbb R^k$, $v_{j,(k+1):H} \in \mathbb R^{H-k}$ be the vectors obtained by taking leading $k$ entries and the remaining $H-k$ ones of $v_{j, 1:H}$ respectively, and $D_{1:k} \in \mathbb R^{k \times k}$, $D_{(k+1):H} \in \mathbb R^{H-k \times H-k}$ be the diagonal matrices with the leading $k$ singular values and the remaining $H-k$ of $Y$ and note
\begin{equation*}
\begin{aligned}
      v_{j, 1:H}^\top (D_{1:H}^2 - d_{H+1}^2 I_H) v_{j, 1:H} =&  v_{j, 1:k}^\top D_{1:k}^2 v_{j, 1:k} - v_{j, 1:k}^\top d_{H+1}^2 I_k v_{j, 1:k} \\
      &+ v_{j, k+1:H}^\top (D_{(k+1): H}^2- d_{H+1}^2 I_{H-k}) v_{j, k+1:H}.
\end{aligned}
\end{equation*}
First,  \begin{equation*}
     v_{j, 1:k}^\top D_{1:k}^2 v_{j, 1:k} = y^{(j)\top} U_{1:k} U_{1:k}^\top y^{(j)} = y^{(j)\top} U_{0,1:k} U_{0, 1:k}^\top y^{(j)} + y^{(j)\top} \big(U_{1:k} U_{1:k}^\top -U_{0,1:k} U_{0, 1:k}^\top\big) y^{(j)},
     \end{equation*}   
For the first term, we have 
\begin{equation*}
\begin{aligned}
    y^{(j)\top} U_{0,1:k} U_{0, 1:k}^\top y^{(j)} =& \lambda_{0j}^\top M^\top M \lambda_{0j} + \epsilon_{j}^\top  M \lambda_{0j} + \lambda_{0j}^\top M^\top \epsilon_{j} +  \epsilon_{j}^\top U_{0,1:k} U_{0, 1:k}^\top \epsilon_{j} \\
    & =n\lambda_{0j}^\top \lambda_{0j} + \lambda_{0j}^\top (M^\top M - nI_k) \lambda_{0j} +\epsilon_{j}^\top  M \lambda_{0j} + \lambda_{0j}^\top M^\top \epsilon_{j} +  \epsilon_{j}^\top U_{0,1:k} U_{0, 1:k}^\top \epsilon_{j} 
\end{aligned} 
\end{equation*}
and, with probability at least $1-o(1)$,
\begin{equation}
    \frac{1}{n}||\lambda_{0j}^\top (M^\top M - nI_k) \lambda_{0j}|| \lesssim \frac{1}{n} ||\lambda_{0j}||^2 || M^\top M - nI_k|| \lesssim \frac{\sqrt{ n\log n}}{n} = \sqrt{\frac{\log n}{n}},
\end{equation}
since, with probability at least $1-o(1)$, $|| M^\top M - nI_k|| \lesssim \sqrt{n \log n}$, by Lemma \ref{lemma:sv_F}.
\begin{equation*}
   \max_{j=1, \dots, p} \frac{1}{n}|\epsilon_{j}^\top  M \lambda_{0j}| =  \frac{1}{n}|\lambda_{0j}^\top M^\top \epsilon_{j} | = \frac{1}{n}  \sum_{i=1}^n \lambda_{0j}^\top (\eta_i^\top \epsilon^{(j)}) \lesssim 1.
\end{equation*}
Note that , uniformly over \(j=1,\ldots,p\), conditional on \(M_0\), $e^{(j)\top}M_0\lambda_{0j}\sim N\!\left(0,\sigma_0^2\|M_0\lambda_{0j}\|^2\right)$. Since, $\max_{j = 1, \dots, p}\|M_0\lambda_{0j}\|^2\lesssim n$, because $||M_0|| \lesssim \sqrt{n}$, by Lemma \ref{lemma:sv_F} and $\max_{j = 1, \dots, p} ||\lambda_{0j}||$ by Assumption \ref{assumption:Lambda}, then, for any \(t>0\),
\[
\mathbb P\left(\max_{j = 1, \dots, p}\left|\frac{1}{n}e^{(j)\top}M_0\lambda_{0j}\right|>t\mid M_0\right)
\le
2p\exp\left(-\frac{nt^2}{2\sigma_0^2 C_1}\right),
\]
for some finite constant $C_1$. Taking \(t=C_2\sqrt{\log p/n}\) with \(C_2\) large enough gives
\begin{equation*}
\max_{j = 1, \dots, p}\left|\frac{1}{n}e^{(j)\top}M_0\lambda_{0j}\right|
=
\mathcal O_{pr}\left(\sqrt{\frac{\log p}{n}}\right) = o_{pr}(1).
\end{equation*}
Moreover, 
\begin{equation}
  \max_{j=1, \dots, p}   \epsilon_{j}^\top U_{0,1:k} U_{0, 1:k}^\top \epsilon_{j}  =  \max_{j=1, \dots, p} ||U_{0, 1:k}^\top \epsilon_{j}||^2 \lesssim \log n,
\end{equation}
by Lemma \ref{lemma:sv_F}, since $U_{0, 1:k}^\top \epsilon_{j} \sim N_k (0, \sigma_0^2 I_k)$. 
For the second term, with probability at least $1-o(1)$, we have
\begin{equation*}
  \max_{j=1, \dots, p}  |y^{(j)\top} \big(U_{1:k} U_{1:k}^\top -U_{0,1:k} U_{0, 1:k}^\top\big) y^{(j)}  | \leq  \max_{j=1, \dots, p}  ||y^{(j)}||^2 || U_{1:k} U_{1:k}^\top -U_{0,1:k} U_{0, 1:k}^\top || \lesssim \sqrt{n},
\end{equation*}
 since, with probability at least $1-o(1)$, $\max_{j=1, \dots, p} ||y^{(j)}|| \lesssim \sqrt{n}$ by Lemma \ref{lemma:y_j}, $|| U_{1:k} U_{1:k}^\top -U_{0,1:k} U_{0, 1:k}^\top || \lesssim \frac{1}{\sqrt{n}}$.
 Thus, with probability at least $1-o(1)$, for each $j=1, \dots, p$, 
 
\begin{equation*}
 \frac{1}{n}  v_{j, 1:k}^\top D_{1:k}^2 v_{j, 1:k}  = ||\lambda_{0j}||^2 + \mathcal O\big(\sqrt \frac{\log n}{n}\big).
\end{equation*}
Thus, since the diagonal entries of $D_{1:K}^2$ are $\asymp n^2$, with probability at least $1-o(1)$, by Lemma \ref{lemma:sv_Y}, the entries of $v_{j, 1:k}$ are $\asymp \frac{1}{\sqrt n}$ with probability at least $1-o(1)$. Hence, with probability at least $1-o(1)$,
\begin{equation*}
    \frac{1}{n}  v_{j, 1:k}^\top d_{H+1}^2 I_{k} v_{j, 1:k} \lesssim \frac{1}{n},
\end{equation*}
since $d_{H+1} \lesssim \sqrt{n}$, with probability at least $1-o(1)$, by Lemma \ref{lemma:sv_Y}. Finally,  with probability at least $1-o(1)$, 
$$
 \frac{1}{n}  v_{j, k+1:H}^\top (D_{(k+1): H}^2- d_{H+1}^2 I_{H-k}) v_{j, k+1:H} \lesssim \frac{n^{3/5}}{n} = n^{-2/5},
$$
since $||D_{(k+1): H}^2- d_{H+1}^2|| \lesssim n^{3/5}$, with probability at least $1-o(1)$, by Lemma \ref{lemma:sv_Y_diff}. 
Combining all of the above, together with $\max_{j=1, \dots, p} ||\lambda_{0j}|| \leq C_\Lambda k$, since $\|\Lambda\|_\infty \leq C_\Lambda$ by Assumption \ref{assumption:Lambda}, we prove the desired result. 
\end{proof}

\begin{lemma}\label{lemma:pi_sigma}
    Suppose $H \lesssim n^{2/5}$ and Assumptions \ref{assumption:model}--\ref{assumption:homoscedasticity}. Then, with probability at least $1-o(1)$, we have 
    \begin{equation*}
        \Pi\big(||\tilde \Sigma|| < C\big) \geq 1 - o(1),
    \end{equation*}
    for some finite constant $C$, where $\Pi(\cdot)$ denotes the probability measure induced by (\ref{eq:posterior_update}, \ref{eq:posterior_update_cc}).
\end{lemma}
\begin{proof}[Proof of Lemma \ref{lemma:pi_sigma}]
    Note that $||\tilde \Sigma|| = \max_{j=1, \dots, p} \tilde \sigma_j^2$, where $ \tilde \sigma_j^2 \mid Y, \hat M, \tau_j, \Psi \sim IG\Big(\frac{v + n}{2}, \frac{v s^2 + ||y^{(j)} - \hat M \hat \lambda_j||^2}{2}\Big)$. Note that $\max_{j = 1, \dots, p} \frac{v s^2 + ||y^{(j)} - \hat M \hat \lambda_j||^2}{2} \asymp \max_{j = 1, \dots, p}||y^{(j)} - \hat M \hat \lambda_j||^2 \leq \max_{j = 1, \dots, p} ||y^{(j)}||^2 \lesssim n$,  with probability at least $1-o(1)$, by Lemma \ref{lemma:y_j}. Theorem 5 of \citet{zhang_zhou_20} concludes the proof. 
\end{proof}

\begin{lemma}\label{lemma:delta_j}
    Let $\delta_j^2 = \frac{v s^2 + ||y^{(j)} - \hat M \hat \lambda_j||^2}{v + n}$. Under Assumptions \ref{assumption:model}--\ref{assumption:homoscedasticity}, we have 
\begin{equation*}
    \delta_j^2 = \frac{s_j^2}{n} + F_j,
    \end{equation*}
    where $\frac{s_j^2}{\sigma_0^2} \sim \chi_{n-k}^2$ and $\max_{j=1, \dots, p}|F_j| \lesssim \frac{1}{n^{2/5}}$, with probability at least $1-o(1)$.
\end{lemma}
  \begin{proof}[Proof of Lemma \ref{lemma:delta_j}]
As in the proof of Theorem \ref{thm:clt}, let $B_j = n \Psi_{n, j}^{-1} = \mathrm{diag}\!\big( \frac{n}{n + \hat \tau_j^{-2} \hat \psi_1^{-2}}, \dots, \frac{n}{n + \hat \tau_j^{-2} \hat \psi_H^{-2}} \big)$. 
Consider the singular value decomposition of the true signal $M \Lambda_0^\top = U_{0, 1:k}D_{0, 1:k}V_{0, 1:k}^\top$ and let $P_0 = U_{0, 1:k}U_{0, 1:k}^\top$ denote the projection onto its columns space. 
Let $\delta_j^2 = \frac{v s^2 + ||y^{(j)} - \hat M \hat \lambda_j||^2}{v_n} = \frac{v s^2 +y^{(j)\top} y^{(j)}
-
y^{(j)\top} U_{1:H} B_j U_{1:H}^\top y^{(j)},}{v_n}$, where $v_n = v + n$. 
Note that 
\begin{equation*}
\begin{aligned}
        y^{(j)\top} (I_n - U_{1:H} B_j U_{1:H}^\top) y^{(j)}  =& y^{(j)\top} (I_n - P_0) y^{(j)} + y^{(j)\top} (P_0 - U_{1:k} U_{1:k}^\top ) y^{(j)}\\ &  + y^{(j)\top} (U_{1:k} U_{1:k}^\top - U_{1:H} B_j U_{1:H}^\top) y^{(j)}  
\end{aligned}
\end{equation*}
Moreover, we have 
\begin{equation*}
Z_j = y^{(j)\top} (I_n - P_0) y^{(j)}
=
\epsilon^{(j)\top} (I_n - P_0) \epsilon^{(j)} \sim  \sigma_0^2 \chi^2_{n-k},
\end{equation*}
since $I_n - P_0$ is an orthogonal projection of rank $n-k$, and, with probability at least $1- o(1)$, we have
     $$\max_{j=1, \dots, p}| y^{(j)\top} (P_0 - U_{1:k} U_{1:k}^\top ) y^{(j)}| \leq \max_{j=1, \dots, p}||y^{(j)}||^2 ||P_0 - U_{1:k} U_{1:k}^\top || \lesssim \sqrt{n}, $$
    and
     \begin{equation*}
    \begin{aligned}
        \max_{j=1, \dots, p} ||y^{(j)\top} (U_{1:k} U_{1:k}^\top - U_{1:H} B_j U_{1:H}^\top) y^{(j)}   || &\leq \max_{j=1, \dots, p}||y^{(j)}||^2 || U_{1:k} U_{1:k}^\top - U_{1:H} B_j U_{1:H}^\top||\\
        &\lesssim n^{3/5},
    \end{aligned}   
    \end{equation*}
since, with probability at least $1- o(1)$, $\max_{j=1, \dots, p} ||y^{(j)}|| \lesssim \sqrt{n}$ by Lemma \ref{lemma:y_j}, $||P_0 - U_{1:k} U_{1:k}^\top || \lesssim \frac{1}{\sqrt{n}}$, by Lemma \ref{lemma:U}, and $||U_{1:k} U_{1:k}^\top - U_{1:H} B_j U_{1:H}^\top|| \lesssim   \frac{1}{n^{2/5}} $, since $||B_j - B || \lesssim  \frac{1}{n^{2/5}}$ by Proposition \ref{prop:psi_l} and Lemma \ref{lemma:hyperparams}, where $B$ is defined in \eqref{eq:B}. 
Combining the results above with $\frac{1}{v_n} = \frac{1}{n} + O\big(\frac{1}{n^2}\big)$, we obtain the desired result.

\end{proof}

\section{Derivation of the shrinkage hyperparameters}
\label{sec:hyp_derivation}

We provide the derivation of the empirical Bayes shrinkage hyperparameters used in Section \ref{subsec:hyperparams}. 
Since \(H\) is intentionally chosen as an upper bound on the true factor dimension, the ordinary spectral loading estimate \(V_{1:H}D_{1:H}/\sqrt n\) contains contributions from overfitted noise directions. We therefore correct for such contribution and start from the initial loading estimate
\begin{equation*}
\bar\Lambda=\frac{1}{\sqrt n}V_{1:H}\left(D_{1:H}^2-d_{H+1}^2I_H\right)^{1/2},
\label{eq:bc_loading_estimate}
\end{equation*}
so that \(\bar\Lambda\bar\Lambda^\top=n^{-1}V_{1:H}(D_{1:H}^2-d_{H+1}^2I_H)V_{1:H}^\top\). The subtraction of \(d_{H+1}^2\) removes the empirical bulk-noise contribution associated with overfitted components.

We use \(\bar\Lambda\) to match prior expected loading magnitudes. Under the prior \(\widetilde\lambda_j \mid \widetilde\sigma_j^2,\tau_j^2,\Psi \sim N_H(0,\widetilde\sigma_j^2\tau_j^2\Psi)\), with \(\Psi=\operatorname{diag}(\psi_1^2,\ldots,\psi_H^2)\)
\begin{equation}
E\{\|\widetilde\lambda_j\|^2 \mid \widetilde\sigma_j^2,\tau_j^2,\Psi\}=\widetilde\sigma_j^2\tau_j^2\operatorname{tr}(\Psi)=\widetilde\sigma_j^2\tau_j^2H,
\label{eq:row_prior_moment}
\end{equation}
where the last equality is due to the normalization condition $\operatorname{tr}(\Psi)=H$. The corresponding spectral estimate of \(\|\lambda_j\|^2\) is the squared norm of the \(j\)-th row of \(\bar\Lambda\), namely
\begin{equation*}
\|\bar\lambda_j\|^2=\frac{1}{n}v_{j,1:H}^\top\left(D_{1:H}^2-d_{H+1}^2I_H\right)v_{j,1:H}.
\label{eq:row_spectral_moment}
\end{equation*}
Matching this quantity with \eqref{eq:row_prior_moment}, after replacing \(\widetilde\sigma_j^2\) by \(\hat \sigma_j^2=\|(I_n-U_{1:H}U_{1:H}^\top)y^{(j)}\|^2/(n-H)\vee c_\sigma\), gives
\begin{equation*}
\widehat \tau_j^2=\frac{v_{j,1:H}^\top\left(D_{1:H}^2-d_{H+1}^2I_H\right)v_{j,1:H}}{nH\hat \sigma_j^2}, \qquad j=1,\ldots,p .
\label{eq:tau_derivation}
\end{equation*}

Similarly, matching the squared norm of the \(\ell\)-th column of \(\bar\Lambda\)  \(\|\widehat\lambda^{(\ell)}\|^2=(d_\ell^2-d_{H+1}^2)/n\) with \(E[\|\widetilde\lambda^{(\ell)}\|^2 \mid \{\widetilde\sigma_j^2,\tau_j^2\}_{j=1}^p,\psi_\ell^2]=\psi_\ell^2\sum_{j=1}^p\widetilde\sigma_j^2\tau_j^2\) and replacing the \(\widetilde\sigma_j^2\)'s and $\tau_j^2$'s with \(\hat \sigma_j^2\)'s and \(\widehat \tau_j\)'s, we obtain 
\begin{equation*}
\widehat\psi_\ell^2=H\frac{d_\ell^2-d_{H+1}^2}{\sum_{m=1}^H(d_m^2-d_{H+1}^2)}, \qquad \ell=1,\ldots,H .
\label{eq:psi_derivation}
\end{equation*}
Thus \(\widehat\psi_\ell^2\) controls global component-wise shrinkage according to the bias-corrected strength of the \(\ell\)-th empirical factor, while \(\widehat\tau_j^2\) controls outcome-wise shrinkage according to the signal-to-noise ratio of the \(j\)-th variable.

\section{Additional details about the numerical experiments and the application}

To estimate $\sigma_0^2$ for constructing the approximation of the confidence and credible intervals, provided in equations \eqref{eq:conf_int} and \eqref{eq:cred_int} respectively, we take $\hat \sigma^2 = \frac{1}{p} \sum_{j=1}^p \delta_j$, where the $\delta_j$'s are defined in Section \ref{subsec:general_methodology}.
For \texttt{MGSP}, we obtain 3000 Monte Carlo samples and discard the first 1000 as burn-in. We fit \texttt{ROTATE} and \texttt{MGSP} with default hyperparameters. For methods that do not provide an estimate for the residual error variances, we estimate them by the empirical variance of the fitted residuals obtained by subtracting the estimated signal from the data matrix. For \texttt{BC}, we estimate $\Lambda \Lambda^\top$ via $\hat X \hat X^\top /n$, where $\hat X$ is the estimate of the low-rank signal $M \Lambda^\top$. 
Coverage results presented in Table \ref{tab:uq} refer to a random $100 \times 100$ submatrix of $\Theta_0$.  

For the application reported in Section \ref{sec:application}, we pre-processed the data by transforming each variable as $\tilde y_{ij} = \Phi^{-1}(\hat F_{j}(y_{ij}^{m}))$, where $\hat F_{mj}(\cdot)$ is the empirical cumulative distribution function of the $j$-th variable. 

All experiments were run on a laptop with 11th Gen Intel(R) Core(TM) i7-1165G7 @ 2.80GHz and 16GB RAM.

\end{document}